\documentclass[submission,copyright,creativecommons]{eptcs}
\providecommand{\event}{Information and Computation}

\usepackage{microtype}
\usepackage{amsmath}
\usepackage{amssymb}
\usepackage{amsthm}
\usepackage{verbatim} 
\usepackage{stmaryrd}
\usepackage{url}
\usepackage{float}
\floatstyle{ruled}
\restylefloat{table}
\restylefloat{figure}
\allowdisplaybreaks
\usepackage{basic}
\usepackage{ambients}
\usepackage{appendix}
\usepackage{paralist}
\usepackage{graphicx}
\usepackage{proof}
\newif\ifdraft\drafttrue
\newcommand{\ntrans}[1]{\mathrel{{\trans{#1}}\makebox[0em][r]{$\not$\hspace{2ex}}}{\!}}

\newcommand{\on}{\mathsf{on}}
\newcommand{\off}{\mathsf{off}}

\newcommand{\confCPS}[2]{#1 \, {\Join} \, #2}

\newcommand{\rsens}[2]{\mathsf{read}\, #2(#1)}

\newcommand{\wact}[2]{\mathsf{write}\, #2 \langle #1 \rangle}

\definecolor{darkred}{RGB}{128,0,0}
\definecolor{darkgreen}{RGB}{0,128,0}
\definecolor{lightgreen}{RGB}{224,255,224}

\newcommand{\dint}[2]{#1_{#2}}
\newcommand{\env}{\mathit{Env}}
\newcommand{\state}{\mathit{S}}

\newcommand{\nome}{\cname{}}
\newcommand{\statefun}{\xi_{\mathrm{x}}}
\newcommand{\actuatorfun}{\xi_{\mathrm{a}}}

\newcommand{\sensorfun}{\xi_{\mathrm{s}}}

\newcommand{\evolmap}{\mathit{evol}}
\newcommand{\evolmapP}{\mathit{evol}}

\newcommand{\measmap}{\mathit{meas}}
\newcommand{\measmapP}{\mathit{meas}}
\newcommand{\invariantfun}{\mathit{inv}}

\renewcommand{\operatorname}[1]{\mathit{#1}}
\newcommand{\dirac}[1]{\overline{#1}}

\newcommand{\CPS}{CPS}
\newcommand\restrict[1]{\raise-.5ex\hbox{\ensuremath|}_{#1}}

\usepackage{marvosym}

\newcommand{\ActComp}{\mathit{Act}}
\newcommand{\subdistr}[1]{{\mathcal D}_{\mathrm{sub}}(#1)}
\newcommand{\distr}[1]{{\mathcal D}(#1)}
\newcommand{\support}{\mathsf{supp}}
\DeclareMathOperator{\Kantorovich}{\mathbf{K}} % Kantorovich's functional
\newcommand{\size}[1]{\mid\!\!{#1}\!\!\mid}
\newcommand{\urg}[1][]{\transS[#1]}
\newcommand{\transS}[1][]{\xrightarrow{\, {#1} \, }} % transition from state to distribution
\newcommand{\ntransS}[1][]{\mathrel{{\transS[#1]}\makebox[0em][r]{$\not$\hspace{2ex}}}{\!}} % negative transition
\newcommand{\nTransS}[1][]{\mathrel{{\TransS[#1]}\makebox[0em][r]{$\not$\hspace{2ex}}}{\!}} % negative transition

\newcommand{\TransS}[1][]{\xRightarrow{\, {#1} \, }}
\makeatletter
\newcommand{\xRightarrow}[2][]{\ext@arrow 0359\Rightarrowfill@{#1}{#2}}
\makeatother

\newcommand{\dummyN}{\mathsf{Dead}}

\newcommand{\metric}{\ensuremath{\mathbf{d} }}
\DeclareMathOperator{\zeroF}{{\bf 0}}
\DeclareMathOperator{\Bisimulation}{\mathbf{B}} 

\newtheorem{theorem}{Theorem}
\newtheorem{proposition}{Proposition}
\newtheorem{definition}{Definition}

\newtheorem{remark}{Remark}
\newtheorem{lemma}{Lemma}

\usepackage{hyperref}

\usepackage{enumitem}

\title{A Probabilistic Calculus of Cyber-Physical Systems~\thanks{A preliminary version appeared in the proceedings of LATA 2017, LNCS 10168, pp.\ 115-127, Springer~\cite{LaMe17}.}}
\author{Ruggero Lanotte 
\institute{Dipartimento di Scienza e Alta Tecnologia\\ 
 Universit\`a degli Studi dell'Insubria, Como, Italy}
\and
Massimo Merro
\institute{Dipartimento di Informatica\\ Universit\`a degli Studi di Verona, Italy}
\and
Simone Tini 
\institute{Dipartimento di Scienza e Alta Tecnologia\\ 
 Universit\`a degli Studi dell'Insubria, Como, Italy}
}
 
\def\titlerunning{A Probabilistic Calculus of Cyber-Physical Systems}
\def\authorrunning{R. Lanotte,  M. Merro,  S. Tini}

\date{}
\begin{document} 
\maketitle

\begin{abstract}
\emph{Cyber-Physical Systems} (\CPS{s}) are integrations of 
networking and distributed computing systems with physical processes, where feedback loops allow physical processes to affect computations and vice versa. 
Although \CPS{s} can be found in several real-world domains (automotive, avionics, energy supply, etc), their verification often relies on \emph{simulation test systems} rather then \emph{formal methodologies}. This is because  there is still a lack of research on the modelling and the definition of formal semantics  to compare non-trivial \CPS{s} in terms of their runtime behaviours up to an acceptable \emph{tolerance}. 

We propose a \emph{hybrid probabilistic process calculus} for modelling  
and reasoning on
\emph{cyber-physical systems\/} (\CPS{s}). 
The dynamics of the calculus is expressed in terms of 
a \emph{probabilistic labelled transition system\/} in the SOS style of Plotkin. 
This is used to define a \emph{bisimulation-based\/}  probabilistic behavioural semantics which supports compositional reasonings. For a more careful comparison between \CPS{s}, we provide two compositional \emph{probabilistic metrics}
to formalise the notion of behavioural distance between systems, also in the case of bounded computations. 
Finally, we provide a non-trivial case study, taken from an engineering 
application, and use it to illustrate our definitions and our compositional 
behavioural theory for \CPS{s}. 
\end{abstract}
%%\keywords{Process calculus, cyber-physical system, semantics.}

%%%%%%%%%%%%%%%%%%%%%%%%%%%%%%%%
%%%%%%                                                                      %%%%%%%
%%%%%%                 I N T R O D U C T I O N                %%%%%%%
%%%%%%                                                                      %%%%%%%
%%%%%%%%%%%%%%%%%%%%%%%%%%%%%%%%

\section{Introduction}

\emph{Cyber-Physical Systems} (\CPS{s}) are integrations of 
networking and distributed computing systems with physical processes, where feedback loops allow physical processes to affect computations and vice versa. 
\CPS{s} can be considered as an evolution of \emph{embedded systems}, where components are immersed in and interact with the physical world, via physical devices (such as \emph{sensors} and \emph{actuators}). 
They can  also be seen as an evolution of \emph{networked control systems}, where physical processes and controllers interact via a communication system.

%%
%%\begin{figure}[t]
%%\centering
%%\includegraphics[width=7cm,keepaspectratio=true,angle=0]{./figures/threat-model.pdf}
%%\caption{Structure of a \CPS}
%%\label{fig:cps-model}
%%\end{figure}

\setlength{\unitlength}{0.4cm}
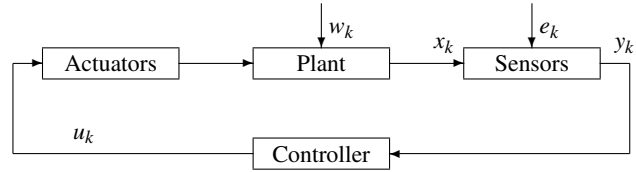
\begin{figure}[t]
\vspace*{-3mm}
\begin{center}
\footnotesize 
\begin{picture}(24,9)(0,0)
\put(1,5){\frame{\makebox(4.5,1){Actuators}}}
\put(8,5){\frame{\makebox(4.5,1){Plant}}}
\put(10.25,7.5){\vector(0,-1){1.5}}
\put(10.5,6.5){$w_k$}
\put(15,5){\frame{\makebox(4.5,1){Sensors}}}
\put(17.25,7.5){\vector(0,-1){1.5}}
\put(17.5,6.5){$e_k$}
\put(8,2){\frame{\makebox(4.5,1){Controller}}}
\put(5.5,5.5){\vector(1,0){2.5}}
\put(12.5,5.5){\vector(1,0){2.5}}
\put(14,6){$x_k$}
\put(19.5,5.5){\line(1,0){1}}
\put(20,6){$y_k$}
\put(20.5,5.5){\line(0,-1){3}}
\put(20.5,2.5){\vector(-1,0){8}}
\put(8,2.5){\line(-1,0){8}}
\put(2,3){$u_k$}
\put(0,2.5){\line(0,1){3}}
\put(0,5.5){\vector(1,0){1}}
\end{picture}
\end{center}
\caption{Structure of a CPS}
\label{fig:cps-model}
\vspace*{-6mm}
\end{figure}

The \emph{physical plant} of a \CPS{} is often 
represented by means of a \emph{discrete-time state-space
model\/}\footnote{We refer to~\cite{survey-CPS-security-2016} for a taxonomy of time-scale models used to represent \CPS{s}.} consisting of two 
equations of the form
\begin{center}
\begin{math}
\begin{array}{rcl}
x_{k+1} & = & Ax_{k} + Bu_{k} + w_{k}\\[2pt]
y_k & = & Cx_{k} + e_k
\end{array}
\end{math}
\end{center}
where
$x_k \in \mathbb{R}^n$ is the current \emph{(physical) state}, $u_k \in
\mathbb{R}^m$ is the \emph{input} (i.e., the control actions implemented
through actuators) and $y_k \in \mathbb{R}^p$ is the \emph{output} (i.e.,
the measurements obtained from the sensors).
The \emph{uncertainty} $w_k \in  \mathbb{R}^n$ and the \emph{measurement error} $e_k \in  \mathbb{R}^p$ represent perturbation and sensor noise, 
respectively. The parameters $A$, $B$, and $C$ are matrices modelling the dynamics of the physical system.  The \emph{next state} $x_{k+1}$ depends on the current state $x_k$ and the corresponding control actions $u_k$, at the sampling instant $k \in \mathbb{N}$. Note that, the state $x_k$ cannot be directly observed: only its measurement $y_k$ can be observed.

The physical plant is supported by a communication network through which
the sensor measurements and actuator data are exchanged with the
\emph{controller(s)\/}, i.e., the \emph{cyber} component, also called
logics,  of a \CPS{} (see \autoref{fig:cps-model}).

In general terms, \CPS{s} can be considered as both  \emph{nondeterministic} and \emph{probabilistic} systems. Nondeterminism arises as they consist of distributed networks in which the activities of specific components occur nondeterministically, whereas the probabilistic behaviour is due to the presence of the uncertainty in the model and the measurement error, which are usually  represented as \emph{probability distributions}.

The range of \CPS{s} applications is rapidly increasing and 
already covers several domains~\cite{CPS-applications}: 
advanced automotive systems, %%process control, 
energy conservation, environmental monitoring, avionics, %% instrumentation, 
critical infrastructure control (for instance, electric power, water resources, and communications systems), etc. 

However, there is still a lack of research on the modelling and validation of \CPS{s} through formal methodologies that allow us to model the interactions among the system components, and to verify the correctness of a \CPS{}, as a whole, before its practical implementation. A straightforward utilisation of these techniques is for \emph{model-checking}~\cite{Clarke:2000:MC:332656}, or even better, for
\emph{probabilistic model-checking}~\cite{PMC}, to statically assess whether the current system deployment can guarantee the expected behaviour. However, they can also be an important aid for system planning, for instance to decide whether 
 different deployments for a given application are  behaviourally equivalent.

Process calculi have been successfully used to model and analyse concurrent, distributed and mobile systems  (see, e.g., the \emph{$\pi$-calculus}~\cite{Mil91},
\emph{Ambients}~\cite{Ambients} and the  \emph{Distributed  $\pi$-calculus}~\cite{dpi}). However, to better describe systems based on a particular paradigm, dedicated calculi are needed.
\emph{Hybrid process algebras}~\cite{CuRe05,BergMid05,vanBeek06,RouSong03,HYPE} 
have been proposed for reasoning about physical systems and provide techniques for analysing and verifying protocols for hybrid automata. In order to enrich hybrid models with probabilistic or stochastic behaviour,  a number of different approaches have been proposed in the last years \cite{Spro2000,Hu2000,Buj04,Abate08,Franzle2011,Hahn2013,Wang17}. However, to our knowledge,  none of these formalisms provide  bisimulation  metrics semantics to estimate the deviation in terms of behaviour of different \CPS{s} in a process-algebra setting. 
The definition of these instruments represents the main goal of the current paper. 

\paragraph{Contribution.}
In this paper, we propose  a \emph{hybrid probabilistic process calculus}, called \cname{} (\emph{Probabilistic Calculus of Cyber-Physical Systems}), with a clearly-defined \emph{probabilistic behavioural semantics\/} for specifying and reasoning on \CPS{s}.  
In  \cname{}, cyber-physical systems are represented by making a neat distinction between  the \emph{physical component\/} describing the physical process (consisting in state variables, sensors, actuators, evolution law, measurement law, etc.) and  the \emph{cyber component}, i.e., the \emph{logics} 
(i.e., controllers, IDS, supervisors, etc.) that governs sensor reading and actuator writing, as well as channel-based communication with other cyber 
components. Thus, channels are used for logical interactions between cyber
components, whereas sensors and actuators make possible the interaction
between cyber and physical components. 
Despite this conceptual similarity, 
messages transmitted via channels are ``consumed'' upon reception,  
whereas  actuators' states (think of a valve)
remains unchanged until its controller modifies it.

 \cname{} adopts a \emph{discrete notion of time\/}~\cite{HR95} and it 
  is equipped with a \emph{probabilistic labelled transition semantics}
(pLTS) in the style of~\cite{Seg95}. 
 We prove that our probabilistic labelled transition semantics satisfies 
some standard time properties such as: \emph{time determinism\/}, 
\emph{patience}, \emph{maximal progress}, and \emph{well-timedness\/}.
Based on our pLTS, we define a natural notion of \emph{weak probabilistic 
bisimilarity\/}, written $\approx$. 
As a main result, we prove that  bisimilarity  in \cname{} is preserved by 
appropriate system contexts 
and it is hence suitable for \emph{compositional reasoning\/}.
Then, we provide a non-trivial \emph{case study}, taken from an
engineering application, and use it to illustrate our definitions and 
our compositional behavioural theory for \CPS{s}. 
We also use our case study  to show that the probabilistic bisimilarity 
is only partially satisfactory to reason on \CPS{s} as it can only establish whether two \CPS{s} behave exactly in the same way or not.
Any tiny variation of the probabilistic behaviour of one of the two  systems under consideration will break the equality without any further 
information on the ``distance'' of their behaviours. 
To this end, \emph{bisimulation metric semantics} have been successfully employed to formalise the \emph{behavioural distance} between two systems~\cite{DJGP02,DGJP04,BW05,DCPP06}.
We generalise our probabilistic bisimilarity by providing a  notion of \emph{weak bisimulation metric}  for \cname{} along the lines of~\cite{DJGP02}. We will write $M \approx_{p} N$, 
 if the weak bisimilarity between \CPS{s} $M$ and $N$ holds with a \emph{distance} $p$,  with probability $p \in [0,1]$.
Intuitively, $\approx_{0}$ will coincide with the weak 
probabilistic bisimilarity $\approx$, whereas $\bigcup_{p \in [0,1]}\approx_{p}$ will correspond to the cartesian product  $\cname{} \times \cname{}$.
% (maximum distance). 
We also provide a notion of \emph{$n$-bisimilarity metric} which takes 
into account bounded computations of systems~\cite{DGJP04}. 
 This kind of  metric, denoted with $\approx_p^n$, for $n\in \mathbb{N}^+$, 
says that the distance $p$ of the systems under considerations is ensured only for  the first $n$ computation steps. Said in other words, 
if $M \approx_p^n N$ then for the first $n$ computation steps the runtime behaviour of systems $M$ and $N$ may differ with \emph{probability} at most $p$.
Both metrics $\approx_p$ and $\approx^n_p$ are proved to be preserved by the 
same contexts considered for $\approx$, and hence they reveal to be suitable for compositional
reasonings. In particular, they satisfy a well-known compositional property 
called \emph{non-expansiveness}~\cite{DGJP04,GLT15,GT15}, the analogue of the congruence property of weak bisimulation. 
 Finally, with the help of our case study, 
we will show how $n$-bisimilarity metric can be  very 
helpful  in situations where it is not necessary to observe a system ``ad infinitum'' as it makes much more sense to observe its behaviour 
for bounded computations. 
\paragraph*{Outline.}
In \autoref{sec:calculus}, we give syntax and 
 operational semantics of
 \cname{}. In \autoref{sec:bisimulation}, we provide a bisimulation-based 
probabilistic behavioural semantics for \cname{} and prove its compositionality. 
In \autoref{sec:case-study}, we model  our case study in \cname{}, and prove for it 
 run-time properties as well as system equalities. 
In \autoref{sec:metric}, we define  bisimulation metrics for \cname{}. In 
\autoref{sec:casebis}, we revise our case study by providing a more 
accurate analysis based on the proposed bisimulation metrics. 
In \autoref{sec:conclusioni}, we draw conclusions and discuss related and future work.

%%%%%%%%%%%%%%%%%%%%%%%%%%%%%%%%
%%%%%%                                                                      %%%%%%%
%%%%%%                        A L G E B R A                         %%%%%%%
%%%%%%                                                                      %%%%%%%
%%%%%%%%%%%%%%%%%%%%%%%%%%%%%%%%

\section{The calculus}
\label{sec:calculus}

In this section, we introduce our \emph{Probabilistic Calculus of Cyber-Physical
Systems},  \cname{}. 

Let us start with some preliminary notations. 
We use   $x, x_k \in \mathsf X$ for \emph{state variables\/}
 (associated to physical states of systems), 
 $c,d \in \mathsf C$ for \emph{communication channels\/}, 
 $a, a_k \in \mathsf A$ for \emph{actuator devices\/}, 
 $s,s_k \in \mathsf S$ for \emph{sensors devices\/}.
\emph{Actuator names} are metavariables for actuator devices like
$\mathit{valve}$, $\mathit{light}$, etc. Similarly, \emph{sensor names}
are metavariables for sensor devices, e.g., a sensor
$\mathit{thermometer}$ that measures a state
variable called $\mathit{temperature}$, with a given precision.
\emph{Values}, ranged
over by $v,v' \in \mathsf V$, are built from basic values, such as
Booleans, integers and real numbers; they also include names.
Given a generic set of names $\mathsf N $, we write $\mathbb{R}^{\mathsf N} $ to
denote the set of functions assigning a real value to each name in $\mathsf N$. For $\xi \in \mathbb{R} ^{\mathsf N}$,
$n \in \mathsf N$ and $v \in \mathbb{R} $, we write $\xi [n \mapsto v]$ to
denote the function $\psi \in \mathbb{R} ^{\mathsf N}$ such that
$\psi(m)=\xi(m)$, for any $m \neq n$, and $\psi(n)=v$.
Given   $\xi_1 \in \mathbb{R}^{{\mathsf N}_1} $ and  $\xi_2 \in \mathbb{R}^{{\mathsf N}_2} $
such that ${{\mathsf N}_1} \cap {{\mathsf N}_2}=\emptyset$,  we denote
with $\xi_1 \uplus \xi_2$ the function in
$\mathbb{R}^{{\mathsf N}_1 \cup {\mathsf N}_2} $
such that
$(\xi_1 \uplus \xi_2) (n)=\xi_1(n)$, if $n \in {{\mathsf N}_1} $, and 
$(\xi_1 \uplus \xi_2) (n)=\xi_2(n)$, if $n \in {{\mathsf N}_2} $.

As \cname{} is a probabilistic calculus, we report the necessary mathematical machinery for its formal definition.

\begin{definition}[Probability distribution]
A (discrete) \emph{probability sub-distribution} over a 
set of generic objects $\mathsf O$ is a function $\delta \colon {\mathsf O} \to [0,1]$ with  $\sum_{o \in {\mathsf O}}\delta(o) \in (0 , 1]$.
We write  $\size{\delta}$ as an abbreviation for $\sum_{o \in {\mathsf O}}\delta(o)$.
The \emph{support} of a probability sub-distribution $\delta$ is given by 
$ \support(\delta) = \{ o \in {\mathsf O} \colon \delta(o) > 0 \}$.
We write $\subdistr {\mathsf O}$,  ranged over $\gamma$, $\delta$ and  
$\epsilon$,  
for the set of all \emph{finite-support} probability sub-distributions over
the set $\mathsf O$.
A probability sub-distribution $\delta \in \subdistr {\mathsf O}$ is said to be a \emph{probability distribution} if $\sum_{o\in {\mathsf O}}\delta(o) =1$.
With $\distr {\mathsf O}$ we denote the set of all finite-support probability distributions over $\mathsf O$.
For any $o \in {\mathsf O}$, the \emph{point (Dirac) distribution at $o$\/}, denoted $\dirac{o}$, assigns probability $1$ to $o$ and $0$ to all others elements of $\mathsf O$, so that $\support{(\dirac{o})} = \{ o \}$. 
\end{definition}

Let $I$ be a finite indexing set such that (i) $\delta_i$ is a sub-distribution in ${\mathcal D}_{\mathrm{sub}}({\mathsf O})$ for each $i \in I$, and (ii) 
 $p_i \geq 0$ are probabilities such that $\sum_{i\in I}p_i \in (0,1]$. 
The probability sub-distribution (or convex combination) $\sum_{i \in I}p_i \cdot \delta_i$  is the sub-distribution defined by
$(\sum_{i \in I}p_i \cdot \delta_i)(o) = \sum_{i \in I} p_i \delta_i(o)$ for all 
$o \in {\mathsf O}$.
We  write a sub-distribution as $p_1 \cdot \delta_1 + \ldots + p_k \cdot 
\delta_k$ when the indexing set $I$ is $\{ 1, \ldots , k \}$.

In \cname{}, a cyber-physical system consists of: 
\begin{itemize}
\item a \emph{physical component} (defining physical variables, physical devices, physical evolution, etc.) and 
\item a \emph{cyber (or logical) component} that interacts with the physical devices (sensors and actuators) and communicates via channels with other cyber components.
\end{itemize}

%%Let us start with a formal description of the physical component of a \CPS{}.
Physical components in \cname{} are given by two sub-components: 
\begin{inparaenum}[(i)] 
\item the \emph{physical state},  which is supposed to change at runtime, and
\item the \emph{physical environment}, which contains static information.\footnote{Actually, this information is periodically updated (say, every six months) to take into account possible drifts of the system.} 
\end{inparaenum}
\begin{definition}[Physical state]
\label{def:physical-state}
Let $\mathsf X$ be a set of state variables, $\mathsf{S}$ be a set of sensors,
and $ \mathsf{A}$ be a set of actuators. A \emph{physical state} $S$ is a triple
$\stateCPS
{\statefun{}}
{\sensorfun{}}
{\actuatorfun{}}
$,
where:
\begin{itemize}
\item $\statefun{} \in \mathbb{R}^{\mathsf X} $ is the
\emph{state function},
\item $\sensorfun{} \in \mathbb{R}^{\mathsf S}$ is the \emph{sensor
function}, 
\item $\actuatorfun{} \in \mathbb{R}^{\mathsf A} $ is the
\emph{actuator function}.
\end{itemize}
All functions defining a physical state are \emph{total}. 
\end{definition}
The \emph{state function} $\statefun{}$ returns the current value  associated to each variable in $\mathsf X$.
 The
\emph{sensor function} $\sensorfun{}$ returns the current value
associated to each sensor in $\mathsf S$; similarly,  the
\emph{actuator function} $\actuatorfun{}$ returns the current value
associated to each actuator in $\mathsf A$.

\begin{definition}[Physical environment]
\label{def:physical-env}
Let $\mathsf{X}$ be a set of state variables,
$\mathsf{S}$ be a set of sensors, and
$\mathsf{A}$ be a set of actuators.
A \emph{physical environment} $E$ is a triple  
$\envCPS
{\evolmap{}}
{\measmap{}}
{\invariantfun{}}
$,
where:
\begin{itemize}
\item $\evolmap{} \colon \mathbb{R}^{\mathsf X} \times
\mathbb{R}^{\mathsf A}  \rightarrow \distr{\mathbb{R}^{\mathsf X}}$ is
 the \emph{evolution map}, 
\item $\measmap{} \colon \mathbb{R}^{\mathsf X}  \rightarrow \distr{\mathbb{R}^{\mathsf S} }$ is the \emph{measurement map}, 
\item  $\invariantfun{} \in 2^{\mathbb{R} ^{{\mathsf X} }}$ 
is the \emph{invariant set}. 
\end{itemize}
All the functions defining a physical environment are \emph{total functions}.
\end{definition}

Given a state function and an actuator function, the \emph{evolution map} $\evolmap{}$ returns a \emph{probability distribution over state functions}. 
This function models the \emph{evolution law} of the physical system, where changes made on actuators may reflect on state variables. Since we assume the presence of a known (maximal) uncertainty for our models, the evolution map does not return a specific state function but a probability distribution over state functions. 

Given a state function, the \emph{measurement map} $\measmap{}$  returns a \emph{probability distribution over sensor functions}. Also in this case, since we assume the presence of a known (maximal) measurement error for each sensor, the measurement map returns a probability distribution over sensor functions, rather than a specific sensor function. 

The \emph{invariant set} $\invariantfun{}$ returns the set of
state functions that satisfy the invariant of the system. A
\CPS{} that gets into a  physical state with a state function that does not satisfy the invariant is in \emph{deadlock}.

\smallskip
Let us now formalise  the cyber components of \CPS{s} in our calculus
\cname{}.
Our (logical) processes build on Hennessy and Regan's   
\emph{Timed Process Language} TPL~\cite{HR95} (basically CCS enriched with a discrete notion of time). 
We extend TPL with  three  constructs: one to read values detected at sensors,  one to write values on actuators, and one to express (guarded) probabilistic choice.  The remaining processes of the calculus are the same as those of TPL.

\begin{definition}[Processes]
\emph{Processes} are defined by the grammar:
\begin{displaymath}
\begin{array}{rl}	
P,Q \Bdf & \nil \Bor \tick.C \Bor  P \parallel Q \Bor \timeout {\mathit{chn}.C} {D} 
\Bor \mathit{phy}.C \Bor 
 \ifelse b P Q \Bor   P\backslash c  
\Bor \\
& X \Bor   \fix X P \\[1pt]
C,D \Bdf & \bigoplus_{i\in I}p_i {:} P_i\\[1pt]
\mathit{chn} \Bdf & \OUT{c}{v} \Bor \LIN{c}{x} 
 \\[1pt]
\mathit{phy} \Bdf &  \ \rsens x s \Bor
 \wact v a \, . 
\end{array}
\end{displaymath} 
\end{definition}

We write $\nil$ for the \emph{terminated process}. The process $\tick.C$
models sleeping  for one time unit. We write $P \parallel
Q$ to denote the \emph{parallel composition} of concurrent processes 
$P$ and $Q$. The process $\timeout {\mathit{chn}.C} D$, with $\mathit{chn}\in
\{\OUT{c}{v},\LIN{c}{x}  \}$, denotes
\emph{channel transmission with timeout}. Thus, $\timeout{\OUT c v . C}D$ sends the
value $v$ on channel $c$ and, after that, it continues as $C$; otherwise,
if no communication partner is available within one time unit, 
 it evolves into $D$.
 The process $\timeout{\LIN c x.C}D$ is the obvious counterpart for channel reception.

Processes of the form $\mathit{phy}.C$ denote  activities on 
physical devices (sensors or actuators). Thus, the construct  $\rsens x s.C$
reads the value $v$ detected by the sensor $s$ and, after that, it
continues as $C$, where $x$ is replaced by $v$. The process $\wact v a.C$ writes the value $v$ on the actuator $a$ and then it continues as $C$.

The process $P\backslash c$ is the channel restriction operator of CCS. It is quantified over the set of communication channels, although we often use the shorthand 
$P\backslash{\{ c_1, \cdots , c_n \}}$ to mean $P\backslash{c_1}\backslash{c_2}\cdots\backslash{c_n}$. 
The process $\ifelse b P Q$ is the standard conditional, where $b$ is a decidable guard. For simplicity, as in CCS, we identify process $\ifelse b P Q$ with $P$, if $b$ evaluates to true,  and $\ifelse b P Q$  with $Q$, if $b$ evaluates to false. 
In processes of the form $\tick.D$ and $\timeout {\mathit{chn}.C} D$, the occurrence of $D$ is said to be \emph{time-guarded}. The process $\fix X P$ denotes \emph{time-guarded recursion} as all occurrences of the process variable $X$ may only occur time-guarded in $P$.

The construct $\bigoplus_{i\in I}p_i {:} P_i$ denotes \emph{probabilistic choice}, where $I$ is a  \emph{finite\/}, \emph{non-empty} set of indexes, 
and $p_i \in (0, 1]$, for $i \in I$,  denotes the probability to execute the process $P_i$, with  $\sum_{i \in I}p_i=1$. As  in~\cite{Dengetal2008}, in order to simplify the operational semantics,  \emph{probabilistic choices occur always underneath prefixing}.

In the two constructs $\timeout{\LIN c x. C}D$ and $\rsens x s. C $,
 the variable $x$ is said to be
\emph{bound\/}. Similarly, the process variable $X$ is bound in $\fix X P$.
 %%In the term $\res c P$ the channel $c$ is bound. 
This gives rise to the standard notions of \emph{free/bound (process) variables} and \emph{$\alpha$-conversion}. We identify processes  up to $\alpha$-conversion (similarly, we identify \CPS{s} up to renaming 
of state variables, sensor names, and actuator names).
A term is \emph{closed} if it does not contain free (process) variables, and  
we assume to always work with closed processes: the absence of free variables is
preserved at run-time. As further notation, we write $T{\subst v x}$ for the substitution of
the variable $x$ with the value $v$ in any expression $T$ of our language.
Similarly, $T{\subst P X}$ is the substitution of the process variable $X$
with the process $P$ in $T$.

Everything is in place to provide the definition of cyber-physical systems 
expressed in \cname{}.
\begin{definition}[Cyber-physical system]
Fixed a set of state variables $\mathsf X$, a set of sensors $\mathsf S$, and a set of actuators $\mathsf A$,
a \emph{cyber-physical system} in \cname{} is given by two components:
\begin{itemize}
\item a \emph{physical component} consisting of 
\begin{itemize}
\item a \emph{physical environment} $E$ defined on $\mathsf X$, $\mathsf S$,
and $\mathsf A$, and 

\item a \emph{physical state} $S$ recording the current values associated to
the   state variables in $\mathsf X$,  the sensors in $\mathsf S$,
and the actuators in $\mathsf A$;  
\end{itemize}
\item a \emph{cyber component} %%represented as a concurrent process
$P$ that interacts with the sensors in $\mathsf S$ and the actuators $\mathsf A$, and can communicate, via channels, with other cyber components of the same  or of other \CPS{s}.
\end{itemize}
We write $\confCPS {E;S}  P$ to denote the resulting \CPS{}, and use 
$M$ and $N$ to range over \CPS{s}. Sometimes, when the physical environment $E$ is clearly identified,  we write $\confCPS S  P$ instead of $\confCPS {E;S} P$. 
\CPS{s} of the form  $\confCPS S  P$ are called environment-free \CPS{s}. 
\end{definition}

The reader should notice that the syntax of our \CPS{s} is slightly too permissive as a process might use
sensors and/or actuators  which are not defined in the physical state. 
\begin{definition}[Well-formedness]
\label{def:well-formedness}
Let  $S = \stateCPS
{\statefun{}}
{\sensorfun{}}
{\actuatorfun{}}
$ be a  physical state, $E = \envCPS
{\evolmap{}} 
{\measmap{}}
{\invariantfun{}} $ 
a physical environment, 
and $P$ a process. The \CPS{} $\confCPS {E;S} P$ is said to be 
\emph{well-formed} if: (i)  any sensor mentioned in $P$ is in the domain 
of  the function $\sensorfun{}$; (ii) any actuator  mentioned in $P$ 
is  in the domain of the function $\actuatorfun{}$. 
A  sub-distribution $\gamma \in {\mathcal D}_{\mathrm{sub}}(\cname)$ is said to be well-formed if its support contains only \nolinebreak well-formed \nolinebreak \CPS{s}. 
\end{definition}
Hereafter, we will always work with well-formed \CPS{s}.

As usual in process calculi, we use the  symbol $\equiv$ to denote
 standard \emph{structural congruence}
for  timed processes~\cite{Mil91,MBS11}; its  generalisation to
\CPS{s} is immediate: $\confCPS{E;S}{P} \equiv \confCPS{E;S}{Q}$ if $P \equiv Q$. 
Also the generalisation 
 to sub-distributions in ${\mathcal D}_{\mathrm{sub}}(\cname{})$ is straightforward: given two  sub-distributions $\gamma$ and $\gamma'$ over \CPS{s}, 
we write $\gamma
\equiv \gamma'$ if $\gamma([M]_{\equiv})=\gamma'([M]_{\equiv})$ for all equivalence classes $[M]_{\equiv} \subseteq \cname{}$.

Finally, we assume a number of \emph{notational conventions}.  
We write $\dummyN$ to denote a deadlocked \CPS{}  
which  cannot perform any action. 
This  fictitious \CPS{}  will  be useful when defining behavioural distances 
between \CPS{s} (see  \autoref{def:simulation_metric}).
We write $\mathit{chn}.P$ instead of $\fix{X}\timeout{\mathit{chn}.P}X$, when $X$ does not occur in $P$. 
We write $\OUTCCS c$ (resp.\ $\LINCCS c$) when channel  $c$ is used for pure synchronisation.
For $k\geq 0$, we write $\tick^{k}.P$ as a shorthand for $\tick.\tick. \ldots \tick.P$, where the prefix $\tick$ appears $k$ consecutive times. 
Given a \CPS{} $M = \confCPS {E;S}  P$, a process $Q$ and a channel $c$, we write 
$M\parallel Q$ for $\confCPS {E;S}  { (P\parallel Q) }$, and 
$M \backslash c$ for $\confCPS {E;S} {(P\backslash c)}$. 

In the rest of the paper, 
symbol $\sigma$ ranges over distributions over physical states,  $\pi$ ranges over distributions over processes,  and $\gamma$ ranges over distributions over \CPS{s}.

%%%%%%%%%%%%%%%%%%%%%%%%%%%%%%%%%%%%%%%%%%%%%%%%%%%%%%%%%%%%%%%%%%%
\subsection{Probabilistic labelled transition semantics}
\label{lab_sem}
In this section, we provide the dynamics of \cname{}  
 in terms of a \emph{probabilistic labelled
transition system} (pLTS)~\cite{Seg95}. First, we give a pretty standard probabilistic LTS for  processes, then we lift transition rules 
from processes to \CPS{s} to deal with the probability distributions occurring 
in  physical environments.

\begin{table}[t]
\begin{displaymath}
\begin{array}{l@{\hspace*{8mm}}l}
\Txiom{Outp}
{-}
{ { \timeout{\OUT c v .C}D } \trans{\out c v}   \sem{C}}
&
\Txiom{Inpp}
{-}
{ { \timeout{\LIN c x .C}D } \trans{\inp c v}    {\sem{C{\subst v x}}}  }

\\[13pt]
\Txiom{Write}
{ - }  %%%p \in \{ s, a \} \Q p \textrm{ data-driven}}
{ { \wact v a .C } \trans{\snda a v}   \sem{C}}
&
\Txiom{Read}
{  - } %%%%p \in \{ s , a \} \Q p \textrm{ data-driven}}
{ { \rsens x s .C } \trans{\rcva s x}    \sem{C}}
\\[13pt]
\Txiom{Com}
{ P_1 \trans{\out c v}  { \pi_1}  \Q  P_2 \trans{\inp c v}  {\pi_2} }
{ P_1 \parallel  P_2 \trans{\tau}  {\pi_1 \parallel \pi_2}}
&
\Txiom{Par}
{ P \trans{\lambda}  \pi \Q \lambda \neq  \tick }
{ {P\parallel Q} \trans{\lambda} {\pi \parallel \dirac{Q}}}
\\[13pt]
\Txiom{ChnRes}{P \trans{\lambda} \pi \Q \lambda \not\in \{ {\inp c v}, {\out c v} \}}{P \backslash c \trans{\lambda} {\pi}\backslash c}
&
\Txiom{Rec}
{  {P{\subst {\fix{X}P} X}} \trans{\lambda}  \pi}
{ {\fix{X}P}  \trans{\lambda}  \pi}
\\[13pt]
\Txiom{TimeNil}{-}
{ \nil \trans{\tick}  \dirac{\nil}}
& 
\Txiom{Delay}
{-}
{  { \tick.C} \trans{\tick}  \sem{C}}
\\[13pt]
\Txiom{Timeout}
{-}
{  {\timeout{\mathit{chn}.C}{D} }   \trans{\tick}  \sem{D}}
&
\Txiom{TimePar}
{
  P_1 \trans{\tick}  {\pi_1}  \Q
   P_2 \trans{\tick} {\pi_2}  \Q P_1 \parallel P_2 \ntrans{\tau}
}
{
  {P_1 \parallel P_2}   \trans{\tick}  {\pi_1 \parallel \pi_2}
}
\end{array}
\end{displaymath}
\caption{Probabilistic LTS for processes}
\label{tab:lts_processes} 
\end{table}

In \autoref{tab:lts_processes}, we provide transition rules for processes.
Here, the meta-variable $\lambda$ ranges over labels in the set 
 $\{\tick,
\tau, {\out c v}, {\inp c  v}, \allowbreak \snda a v,\rcva s x \}$. 
These labels denote the passage of time, internal activities, channel transmission, 
channel reception, actuator writing, and sensor reading, respectively. 
As in \cite{Dengetal2008}, the definition of the labelled transition relation
for processes relies on a semantic interpretation of probabilistic processes 
in terms of (discrete) probability distributions over processes. 
\begin{definition}
For any probabilistic choice $\bigoplus_{i \in I} p_i {:} P_i$ 
over a finite index set $I$,  we write $\sem {\bigoplus_{i \in I} p_i {:} P_i}$ 
to denote the probability distribution $\sum_{i \in I}p_i\cdot \dirac{P_i}$. 
\end{definition}
The transition rules in \autoref{tab:lts_processes} use some obvious notation for distributing both parallel composition and channel restriction over a  sub-distribution. 
Given two sub-distributions $\pi_1$ and $\pi_2$ we define the sub-distribution $\pi_1 \parallel \pi_2$ as follows: $({\pi_1} \parallel {\pi_2})(P) = \pi_1(P_1)\cdot\pi_2(P_2)$, if  $P = P_1 \parallel P_2$; $({\pi_1} \parallel {\pi_2})(P) = 0$, otherwise. Given an arbitrary distribution over processes  $\pi = \sum_{i \in I}p_{i}\cdot \dirac{P_i}$, an arbitrary channel $c$, and a value $v$, we define $\pi \backslash c$ as the distribution $ \sum_{i \in I}p_{i}\cdot \dirac{P_i\backslash c}$,
 and  $ \pi \subst{v}{x} $ as the distribution $\sum_{i \in I}p_{i}\cdot \dirac{ P_i  \subst{v}{x} }$. 

Let us comment on the transition rules of \autoref{tab:lts_processes}. 
Rules
\rulename{Outp}, \rulename{Inpp} and \rulename{Com} serve to model channel
communication, on some channel $c$. Rule~\rulename{Write} denotes the 
writing of some data $v$ on an actuator $a$. 
Rule~\rulename{Read} denotes the reading of some value detected at sensor $s$. Rule \rulename{Par} propagates untimed actions over parallel components.
Rules  \rulename{ChnRes} and \rulename{Rec} are the standard rules for 
channel restriction and recursion, respectively. The following four rules 
are standard, and model the passage of one time unit.  The symmetric counterparts of rules \rulename{Com} and \rulename{Par} are obvious and thus omitted from the table.

\begin{table}[t]
\begin{displaymath}
\begin{array}{c}
\Txiom{Out}
{P \trans{\out c v}  \pi  \Q \Q  S \in \invariantfun{}}
{\confCPS S  P   \trans{\out c v}   \confCPS {\dirac{S}}  {\pi}}
\Q\Q
\Txiom{Inp}
{P  \trans{\inp c v}  \pi \Q \Q S \in \invariantfun{}}
{\confCPS S  P    \trans{\inp c v}  \confCPS {\dirac{S}} {\pi}}
\Q\Q 
\Txiom{Tau}{P \trans{\tau} \pi \Q \Q S \in \invariantfun{}}
{ \confCPS S  P \trans{\tau} \confCPS {\dirac{S}}  {\pi}}
\\[16pt]
\Txiom{SensRead}{P \trans{\rcva s z} \pi \Q \Q 
 \sensorfun{}(s) = \sum_{i \in I} p_i \cdot \dirac{v_i}  \Q \Q 
{ \statefun{}  \in \invariantfun{} }}
{\confCPS {{\stateCPS {\statefun{}} {\sensorfun{}} {\actuatorfun{}}}}  P \trans{\tau} \confCPS {\dirac{{\stateCPS {\statefun{}} {\sensorfun{}} {\actuatorfun{}}}}} {\sum_{i \in I}p_i \cdot   \pi \subst{v_i}{z}}} 
\\[16pt]
\Txiom{ActWrite}{P \trans{\snda a v} {\pi}  \Q \Q  {  {\statefun{}  \in \invariantfun{}}}}
{\confCPS {\stateCPS {\statefun{}} {\sensorfun{}} {\actuatorfun{}}}  P \trans{\tau} \confCPS {\dirac{\stateCPS {\statefun{}} {\sensorfun{}} {\actuatorfun{}[a \mapsto v]} }}{\pi}}
\\[16pt]
\Txiom{Time}{ P \trans{\tick} \pi \Q \Q
\confCPS S P \ntrans{\tau}\Q \Q
S \in \invariantfun{} \Q \Q  }
{\confCPS S  P \trans{\tick} \confCPS  {\operatorname{next}_E(S)} {\pi}}
\Q\Q\Q
\Txiom{Deadlock}{S \not \in \invariantfun{}}{\confCPS{S}{P}  
\trans{\tau} \dirac{\dummyN}}
\end{array}
\end{displaymath}
\caption{Probabilistic LTS for a \CPS{} $\confCPS S P$ parametric on an 
environment $E = \envCPS
{\evolmap{}} 
{\measmap{}}
{\invariantfun{}}$}
\label{tab:lts_systems_P} 
\label{tab:lts_systems} 
\end{table}

In \autoref{tab:lts_systems}, we lift the transition rules from processes
to systems, actually to probability distributions overs systems. 
We adopt the following notation for probability distributions:
given a distribution $\sigma$ over physical states and a distribution $\pi$ over processes, we write $\confCPS {\sigma} {\pi}$ to denote the distribution over (environment-free) \CPS{s} defined as $(\confCPS {\sigma} {\pi})(\confCPS{S}{P})= {\sigma}(S) \cdot \pi(P)$.
Moreover, given a physical environment $E$, we write $\confCPS {E; \sigma} {\pi}$ to extend  the distribution $\confCPS {\sigma} {\pi}$ to full \CPS{s} as follows: $(\confCPS {E; \sigma} {\pi})(\confCPS{E;S}{P})= {\sigma}(S) \cdot \pi(P)$. 
 Actions,  ranged over by $\alpha$, are in the set $\ActComp = \{\tau,
{\out c v}, {\inp c v}, \tick \}$. These actions 
denote: non-observable activities ($\tau$); channel transmission (${\out c v}$); channel reception (${\inp c v}$); the passage of  time ($\tick$).

As physical environments contain static information, for simplicity the resulting transition rules are parameterised on a physical environment of the form $E = \envCPS
{\evolmap{}}
{\measmap{}}
{\invariantfun{}} $. Thus, instead of providing the transitions rules for 
a \CPS{} of the form $\confCPS {E;S}{P}$ we give the LTS semantics parametric on $E$ for the environment-free \CPS{} $\confCPS {S}{P}$.  

All rules, except \rulename{Deadlock}, have a common premise requiring that 
the current state function of the system must satisfy the invariant. 
With an abuse of notation, we sometimes write $S   \in \invariantfun{}$  instead of 
 $ \statefun{}  \in \invariantfun{}$ when  
$S=  \stateCPS {\statefun{}} {\sensorfun{}} {\actuatorfun{}}$. 
 Rules \rulename{Out} and
\rulename{Inp} model transmission and reception, with an external system,
on a channel $c$. 
Rule \rulename{Tau} lifts non-observable actions from processes to
systems. 
Rule \rulename{SensRead} models the reading of the
current data detected at sensor $s$. 
 Rule~\rulename{ActWrite} models the writing of a value $v$ on an
 actuator $a$. 
 A similar lifting occurs in rule
\rulename{Time} for timed actions, where $\operatorname{next}_E(S)$ returns a 
probability distribution over possible physical states for the next time slot, 
according to the current physical state $S$ and  physical environment $E$.
 Formally, for 
 $S = \stateCPS
{\statefun{}}
{\sensorfun{}}
{\actuatorfun{}}
$ and $E = \envCPS
{\evolmap{}}
{\measmap{}}
{\invariantfun{}} $, we define: 
 \[
\mathit{next}_E(S)  \; = \;  
\sum_{\substack{ \statefun'{}  \in \support(\evolmap{}(\statefun{}, \actuatorfun{}))  \\ 
\sensorfun'{} \in \support(\measmap{}({\statefun{}'}))} } 
\big( \evolmap{}(\statefun{}, \actuatorfun{}) (\statefun'{})      \cdot
 \measmap{}({\statefun{}')}(\sensorfun'{})  \big)  \cdot
\dirac{ 
 \stateCPS {\statefun'{}} {\sensorfun'{}}
{\actuatorfun{}}  }
 \, . 
\]
Intuitively, the operator $\operatorname{next}_E$ serves to compute the
possible state functions and  sensor functions of the next time slot (actuator changes are governed by the cyber-component). More precisely, the (probability distribution over the) next state function is determined  by applying $\evolmap{}$ to the current state function $\statefun{}$ and the current actuator function $\actuatorfun{}$.
The probability weight of any possible state function $\statefun'{}$ is given by $\evolmap{}(\statefun{}, \actuatorfun{})(\statefun'{})$. 
Then, for a state function $\statefun'{}$, 
the (probability distribution over the) next sensor function is given by applying $\measmap{}$ to $\statefun'{}$. Finally, 
the probability weight of any possible sensor function $\sensorfun'{}$ is
given by $\measmap{}(\statefun'{})(\sensorfun'{})$.

Recapitulating, by an application of rule \rulename{Time} a \CPS{} moves to the next physical state, in the next time slot. Rule \rulename{Deadlock} is straightforward: if the invariant is not satisfied then the \CPS{} deadlocks. 

Finally, notice that in our LTS we defined transitions rules of the 
form $\confCPS {S}{P} \trans{\alpha} \confCPS \sigma \pi$, parametric on some 
physical environment $E$. As physical environments do not change at runtime,
$\confCPS {S}{P} \trans{\alpha} \confCPS \sigma \pi$ entails $\confCPS {E;S}{P} \trans{\alpha} E ; \confCPS \sigma \pi$, thus providing the probabilistic LTS for (full) \CPS{s}.

\begin{remark}
Note that the rules in \autoref{tab:lts_systems_P} define an \emph{image finite} pLTS. This means that for any \CPS{} $M$ and label $\alpha$ there are finitely many distributions reachable from $M$ in one $\alpha$-labelled transition step.
Moreover, all transitions $M \transS[\alpha] \gamma$ are such that $\gamma$ has a finite support.
\end{remark}

Now, having defined the labelled transitions that can be performed by
a \CPS{} of the form $\confCPS{E;S}{P}$, we can easily concatenate these transitions to define 
the possible computation traces of a system. 
A \emph{computation trace}~\cite{BdNL14} for a \CPS{} $\confCPS {{E};S_1} {{P}_1}$ is a sequence of steps of the form
$\confCPS {E; S_1} {{P}_1}  \transS[\alpha_1] \dots \transS[\alpha_{n-1}] \confCPS {E ; S_n} {{P}_n}$ where for any $i$, with $1 \leq i \leq n-1$, we have $\confCPS {E ; S_i} {P_i} \transS[\alpha_i]  E ; \confCPS {\sigma_{i+1}} {\pi_{i+1}}$ for distributions  $\sigma_{i+1}$ and $\pi_{i+1}$ such that 
$S_{i+1} \in \support{(\sigma_{i+1})}$ and $P_{i+1} \in \support{(\pi_{i+1})}$.

Below,  we report a few  desirable time properties~\cite{HR95} which hold in our calculus:
(a) \emph{time determinism\/}, (b) \emph{maximal progress\/}, (c)
\emph{patience\/}, 
 and (d) \emph{well-timedness\/}. 
In its standard formulation, \emph{time determinism}
says that a system reaches at most one new state by executing a 
timed action $\tick$; however, in our setting, this holds
only for the logical components (up to structural congruence) whereas 
the evolution of 
the physical component is intrinsically probabilistic, due to 
the presence of  uncertainty and measurement errors. 
The \emph{maximal progress}
 property usually says that processes communicate as soon as a possibility of communication arises. In our calculus, we generalise this property saying that instantaneous (silent) actions cannot be delayed. 
On the other hand, \emph{patience} says that if no instantaneous actions are possible then time is free to pass.  Finally, \emph{well-timedness}~\cite{MBS11,CHM15} ensures the absence of infinite instantaneous 
traces  which would prevent the passage of time, and hence the 
physical evolution of a \CPS{}.

\begin{theorem}[Time properties] 
\label{prop:time}	
Let $M = \confCPS {E;S} P$. 
\begin{itemize}[noitemsep]
\item[(a)]
\label{prop:timed}
If $M \trans \tick  \gamma $ and 
$M\trans \tick \gamma'$ then $\gamma \equiv \gamma'$.
\item[(b)]
\label{prop:maxprog}
If $M\trans \tau \gamma$ then there is no $\gamma'$ such that $M\trans \tick
\gamma'$. 
\item[(c)]
\label{prop:patience}
If $M \trans \tick \gamma$ for no $\gamma$ then 
either 
$S$ does not satisfy the invariant of $E$ or
there is $\gamma'$ such that $M \trans \tau \gamma'$. 
\item[(d)]
\label{prop:welltime}
There is a $k\in\mathbb{N}$ such that  if
$M  \transS[\alpha_1]\dots \transS[\alpha_n] N$, with $\alpha_i \neq \tick$,   then $n\leq k$. 
\end{itemize}
\end{theorem} 
The proof of \autoref{prop:time} can be found in the Appendix, in \autoref{app_uno}.

%%%%%%%%%%
\begin{comment}
\begin{proof} We report the proof of the last properties. The other proofs
can be found in \autoref{app_uno}.
%the Appendix. 

The proof is by contradiction. 
Suppose there is no $k$ satisfying the statement above.
This implies there is an unbounded derivation:
\[
M= M_1  \trans{\alpha_1}\dots \trans{\alpha_n} M_{n+1} \trans{\alpha_{n+1}} \dots
\]
with $M_i= \confCPS {E_i} {P_i}$ and $\alpha_i \neq \tick$, for any $i$.
By inspection of the rules of \autoref{tab:lts_systems} it follows that if 
$M_i  \trans{\alpha_i} M_{i+1} $, with $\alpha_i \neq \tick$, then  $P_i  \trans{\lambda_i} P_{i+1}  $, for some $\lambda_1$, with $\lambda_i\neq \tick$.
Hence we would have the following unbounded derivation for processes:
\[
P_1  \trans{\lambda_1}\dots \trans{\lambda_n} P_{n+1} \trans{\lambda_{n+1}} \dots
\]
with $\lambda_i \neq \tick$, for any $i$. However, our processes 
are an easy extension of those of TPL~\cite{HR95}, with constructs for 
 sensors reading and 
actuator writing. As these two operators are ephemeral, they cannot introduce
divergence. Thus, as TPL with sensor reading and actuator writing is well-timed, we derive the required contradiction.
\end{proof}
\end{comment}
%%%%%%%%%%

%%%%%%%%%%%%%%%%%%%%%%%%%%%%%%%%%%%%%%%%%%%%%%%%%%%%%%%%%%%%%%%%%%%%

\section{Probabilistic bisimulation}
\label{sec:bisimulation}
 In this section, we are ready to define a bisimulation-based behavioural 
equality for \CPS{s}, relying on our labelled transition semantics. We recall that the only 
\emph{observable activities} in \cname{} are: the passage of time
 and channel communication.  As a consequence, the capability to observe physical events (different from deadlocks) depends on the capability of the cyber components to recognise those events by acting on sensors and actuators, and then signalling them using (unrestricted) channels.

In a probabilistic setting, the definition of weak transition $\TransS[\hat{\alpha}]$, which abstract away non-observable actions, is complicated by the fact that (strong) transitions take \CPS{s} to distributions over \CPS{s}.
Following \cite{Dengetal2008,LMT17}, we need to generalise transitions, so that they take sub-distributions to sub-distributions.

With an abuse of notation, we use  $\gamma$ and $\gamma'$  to range over sub-distributions over \CPS{s}, under the assumption that $\sum_{M \in \nome} \gamma(M) \le 1$.

Let us start with defining the weak transition $M \urg[\hat{\alpha}] \gamma$ for any \CPS{} $M$ and distribution $\gamma$.
If $\alpha = \tau$ then we write $M \urg[\hat{\alpha}] \gamma$ whenever either $M \transS[\alpha] \gamma$ or $\gamma = \dirac{M}$. Otherwise, 
if $\alpha \neq \tau$ then  we write $M \transS[\hat{\alpha}] \gamma$ whenever $M \transS[\alpha] \gamma$. 
The relation $\transS[\hat{\alpha}]$ is extended to model transitions from sub-distributions to sub-distributions.
For a sub-distribution $\gamma =\sum_{i \in I} p_i \cdot \dirac{M_i}$, we write $\gamma \urg[\hat{\alpha}] \gamma'$ if
 there is a non-empty set $J\subseteq I$ such that $M_j \transS[\hat{\alpha}] \gamma_j$ for all $j \in J$,
$M_i \ntransS[\hat{\alpha}]$, for all $i \in I \setminus J$,  and $\gamma' = \sum_{j \in J}p_j \cdot  \gamma_j$.
Note that if $\alpha \neq \tau$ then this definition entails that only some \CPS{s} in the support of $\gamma$ have an $\transS[\hat{\alpha}]$ transition.
Then, we define the weak transition relation $\TransS[\hat{\tau}]$ as the transitive and reflexive closure of $\transS[\hat{\tau}]$, i.e.\ 
$\TransS[\hat{\tau}] = (\transS[\hat{\tau}])^{\ast}$, while for $\alpha \neq \tau$ we let $\TransS[\hat{\alpha}]$ denote $\TransS[\hat{\tau}] \transS[\hat{\alpha}] \TransS[\hat{\tau}]$.
 
In order to define a probabilistic bisimulation, following \cite{DD11} we rely on the notion of 
 \emph{matching}~\cite{Vil08} (also known as \emph{coupling}) for a pair of distributions. Intuitively, the matching for
a pair $(\gamma,\gamma')$ may be understood as a transportation schedule for the shipment of probability mass from $\gamma$ to $\gamma'$.
\begin{definition}[Matching]
\label{def_matching}
A \emph{matching} for a pair of distributions 
$(\gamma,\gamma')$, with $\gamma, \gamma'   \in \distr{\nome}$, is a distribution $\omega$ in the product space $\distr{\nome\times \nome}$ such 
that:
\begin{itemize}
\item
$\sum_{{M'} \in \nome} \omega(M,M')=\gamma(M)$, for all $M \in \nome$, and 
\item 
$\sum_{M \in \nome} \omega(M,M')=\gamma'(M')$, for all $M' \in \nome$. 
\end{itemize}
We write $\Omega(\gamma,\gamma')$ to denote the set of all matchings for $(\gamma,\gamma')$.
\end{definition}

Everything is in place to define  weak probabilistic bisimulation for \cname{}, along the lines of~\cite{ALS00}.

\begin{definition}[Weak probabilistic bisimulation]
\label{def:bisimulation}
A binary symmetric relation $\RR$ over \CPS{s} is a \emph{weak probabilistic bisimulation} if 
 $M \RRr N$ and $M \trans{\alpha} \gamma$ implies that there exist a distribution $\gamma'$ and a matching $\omega \in \Omega(\gamma,\gamma')$ 
such that $N\Trans{\hat{\alpha}}\gamma'$, and  
$M' \RRr N'$ whenever $\omega(M',N') > 0$. We say that $M$ and $N$ are bisimilar, written $M \approx N$, 
if $M \RRr N$ for some weak probabilistic bisimulation $\RR$. 
\end{definition}

A main result of the paper is that  bisimilarity can be used to reason on \CPS{s} in a compositional manner. In particular,  bisimilarity is preserved   by 
parallel composition of \emph{physically-disjoint} \CPS{s}, 
by parallel composition of \emph{pure-logical} processes,
 and by channel restriction; basically, all those contexts that cannot interfere on 
physical devices (sensors and actuators), whereas interferences on logical components (via channel communication) is allowed. 

Intuitively, two \CPS{s} are physically-disjoint if they 
have different plants but they may share logical channels for 
communication purposes. More precisely, physically-disjoint \CPS{s} have disjoint state variables and disjoint physical devices (sensors and actuators).  
As we consider only well-formed \CPS{s} (\autoref{def:well-formedness}), this ensures us that a  \CPS{} cannot 
physically interfere with a parallel \CPS{} by acting on its physical devices. Although, logical interferences on communication channels are allowed.

Formally, let   
$S^i = \stateCPS {\statefun^i{}} {\sensorfun^i{} } {\actuatorfun^i{} }$ and 
$E^i = \envCPS
{\evolmap^i{} }
{\measmap^i{} }   
{\invariantfun^i{} }
$ be physical states and physical environments, respectively, 
associated to  state variables in the set  
  ${\mathsf{X}}_i$, sensors in  the set ${\mathsf{S}}_i$, and actuators in the set ${\mathsf{A}}_i$,  for $i \in \{ 1,2\}$.
For ${\mathsf{X}}_1 \cap {\mathsf{X}}_2=\emptyset$, 
${\mathsf{S}}_1 \cap {\mathsf{S}}_2=\emptyset$  and 
${\mathsf{A}}_1 \cap {\mathsf{A}}_2=\emptyset$,     we define: 

\begin{itemize}
\item   the \emph{disjoint union} of the physical states $S_1$ and $S_2$,
written $S_1 \uplus S_2$,  to be the physical state 
$   \stateCPS
{\statefun{}} 
{\sensorfun{}}
{\actuatorfun{} }
$
such that: 
${\statefun{} }=  \statefun^1{} \uplus \statefun^2{}  $, 
${\sensorfun{} }=\sensorfun ^1{} \uplus \sensorfun^2{}  $, and 
${\actuatorfun{} }=\actuatorfun ^1{} \uplus \actuatorfun^2{}  $; 
\item
 the \emph{disjoint union} of the physical environments $E_1$ and $E_2$,
written $E_1 \uplus E_2$,  to be the physical environment 
$   \envCPS 
{\evolmap{} } 
{\measmap{} }   
{\invariantfun{} }
$ such that:
\begin{center}
\begin{math}
\begin{array}{rcl}
({\evolmapP{}}(\statefun^1{} \uplus \statefun^2{},\actuatorfun ^1{} \uplus \actuatorfun^2{}))(\statefun^1{'} \uplus \statefun^2{'}) & = &{\evolmapP^1{}}(\statefun^1{},\actuatorfun^1{})(\statefun^1{'}) 
\cdot {\evolmapP^2{}}(\statefun^2{},\actuatorfun^2{})(\statefun^2{'}) 
\\
({\measmapP{}}( \statefun^1{} \uplus \statefun^2{} ))( \sensorfun^1{'} \uplus \sensorfun^2{'})& =  &
{\measmap^1{}}(\statefun^1{} )(\sensorfun^1{'} )
\cdot
{\measmap^2{}}(\statefun^2{} )(\sensorfun^2{'} )
 \\
 \statefun^1{} \uplus \statefun^2{}  \in \invariantfun{} & \text{iff} &
   \statefun^1{}   \in \invariantfun^1{}  \text{ and }
 \statefun^2{}   \in \invariantfun^2{} \, . 
\end{array}
\end{math}
\end{center}
\end{itemize}

\begin{definition}[Physically-disjoint \CPS{s}]
Let $M_i = \confCPS {E_i;S_i}{P_i}$, for $i \in \{ 1 , 2 \}$. We say that 
$M_1$ and $M_2$ are physically-disjoint if
$S_1$ and $S_2$ have disjoint sets of state variables, sensors and actuators. In this case,  we write $M_1 \uplus M_2$ to denote the \CPS{}  defined as $\confCPS {(E_1\uplus E_2); (S_1 \uplus S_2)}{(P_1\parallel P_2)}$. For any 
$M \in \cname{}$, the special system $\dummyN$ is physically-disjoint with 
$M$, and $M \uplus \dummyN = \dummyN \uplus M = \dummyN$. 
\end{definition}

A \emph{pure-logical process} is a process which may interfere on communication channels but it never interferes on physical devices as it never accesses  sensors and/or actuators. Basically, a pure-logical process is  a (possibly probabilistic) TPL process~\cite{HR95}.
 Thus, in a system $M \parallel Q$, where $M$ is an arbitrary 
\CPS{}, a pure-logical process $Q$ cannot interfere with the physical evolution of $M$. 
Although, process $Q$ can definitely interact with $M$ via communication 
channels, and hence affect its observable behaviour. 
\begin{definition}[Pure-logical processes]
A process $P$ is called \emph{pure-logical} if it never acts on 
sensors and/or actuators.
\end{definition}

Now, we can finally prove the compositionality of  probabilistic bisimilarity $\approx$. 
\begin{theorem}[Congruence results]
\label{thm:congruence}
Let $M$ and $N$ be two arbitrary \CPS{s} in \cname{}. 
 \begin{enumerate}
\item
\label{thm:congruence1}
 $M \approx N$ implies $M \uplus O \approx N \uplus O$, for any 
physically-disjoint \CPS{} $O$; 
\item
\label{thm:congruence2}
 $M \approx N$ implies $M \parallel P \approx N \parallel P$, for any 
pure-logical process $P$; 
\item 
\label{thm:congruence3}
$M \approx N$ implies $M \backslash c \; \approx \; N \backslash c$, for 
any channel $c$.
\end{enumerate} 
\end{theorem}
%

%%\mm{\autoref{thm:congruence} is a special case of \autoref{thm:congruenceP},
The proof can be found in the Appendix, at the end of \autoref{Sec:A.5}. 
%%%%%\remarkM{In realtà quando il lettore va in Appendice resta un po' interdetto perchè il rinvio al Theorem 3 e' messo un po' a cazzo.}

The reader may wonder whether the bisimilarity $\approx$ is preserved by more permissive contexts.  The answer is no. 
Suppose to allow in the second item of \autoref{thm:congruence} a process $P$ that can also read on sensors. In this case, even if $M$ and $N$ are bisimilar, the parallel process $P$ might read a different value in the two systems  at the very same sensor $s$ (due to the sensor error) and transmit these different values on a free channel, breaking the congruence. Activities on actuators  may also lead to different behaviours of the compound systems:  bisimilar \CPS{s} may have physical components that are not exactly aligned. 
A similar reasoning applies when composing \CPS{s} with non physically-disjoint ones:  interference on physical devices may break the congruence.

However, in the next section we will see that the congruence results of 
\autoref{thm:congruence} will  be very useful when reasoning on complex systems.

\section{Case study}
\label{sec:case-study}
In this section, we provide a case study  to illustrate
how \cname{} can be used to specify and reason on \CPS{s} in a compositional 
manner. In particular, we model an engine %%%, called $\mathit{Eng}$,
 whose temperature is maintained within a specific range by means of a cooling system.

As regards the \emph{physical environment} we adopt discrete uniform distributions over suitable intervals to model both the evolution map and the measurement map.\footnote{Other forms of finite-support discrete probability distributions could be treated as well.} 
In our model, we assume a \emph{granularity}  $g\in \mathbb{N}^{+}$ representing the precision $10^{-g}$ of the model in estimating  physical values. 
Thus, for an arbitrary real interval $[v,w]$  we write  $\dint{[v,w]}{g}$
to denote the \emph{finite} set of reals $\{k \in[v,w] \colon k= v+ h \cdot 10^{-g}, 
\textrm{ with } h \in \mathbb{N}\}$. 

Given a granularity $g\in \mathbb{N}^{+}$, the physical state $\state_g$ of the engine  is characterised by: (i) a state variable $\mathit{temp}$ containing the current temperature of the engine; (ii)   a sensor
$s_{\mathrm{t}}$ (such as a thermometer or a thermocouple) measuring the temperature of the engine,
%%  with an error $\epsilon =0.1$;
(iii) an actuator $\mathit{cool}$ to turn on/off the cooling system. 
The physical environment of the engine, $\env_g$, is constituted by:  
%%(iii) an uncertainty $\delta=0.4$ associated to the only variable $\mathit{temp}$; 
(i) a simple evolution law $\evolmap$ that increases (resp.\ decreases) the value of $\mathit{temp}$, when the cooling system is inactive (resp.\ active), by a value determined according to a discrete distribution of probability, taking into account an uncertainty in the model that may reach the threshold $\delta=0.4$, and granularity $g$  over reals;
 (ii)  a measurement map $\measmap{}$ returning the value detected by the sensor $s_t$  determined by a discrete probability distribution  based on a
 measurement error that may reach the threshold  $\mathit{err} = 0.1$, and granularity $g$;
(ii) an invariant set saying that the system gets faulty when the temperature of the engine gets out of the range $[0, 30]$.

Formally,  
$\state_g = \stateCPS 
{\statefun{}} 
{\sensorfun{}}
{\actuatorfun{}} 
$ and 
$\env_g = \envCPS 
{\evolmap{}}
{\measmap{}}   
{\invariantfun{}}$ 
with:
\begin{itemize}
\item[(i)] 
$\statefun{} \in \mathbb{R} ^{\{\mathit{temp}\} }$ and 
$\statefun{}(\mathit{temp})=0$;
\item[(ii)] 
$\sensorfun{} \in  \mathbb{R}^{\{s_{\mathrm{t}} \}}$ and 
$\sensorfun{}(\mathit{temp})=0$;
\item[(iii)]
$\actuatorfun{} \in \mathbb{R} ^{\{\mathit{cool}\} } $ and
$\actuatorfun{}(\mathit{cool})=\off$; for the sake of simplicity, we can
assume $\actuatorfun{}$ to be a mapping $\{ \mathit{cool} \} \rightarrow
\{ \on , \off\}$ such that $\actuatorfun{}(\mathit{cool})= \off$ if
$\actuatorfun{}(\mathit{cool}) \geq 0$, and $\actuatorfun{}(\mathit{cool})= \on$ if
$\actuatorfun{}(\mathit{cool}) < 0$. 
\end{itemize}

Furthermore, 
\begin{itemize}
\item[(i)] 
$\evolmap{}( \statefun'{}, \actuatorfun'{}) = 
\sum_{v \in [v_1,v_2]_g}  \frac{1}{|[v_1, v_2]_g|} \cdot \dirac{[\mathit{temp} \mapsto  \xi'_{\operatorname{x}}(temp) + v]}$,
for any $\xi'_{\operatorname{x}}\in \mathbb{R} ^{\{temp\}}$ and $\actuatorfun'{} \in \mathbb{R} ^{\{cool\}}$,
where $[v_1,v_2] =[1{-}\delta \, , \, 1{+}\delta]$,  if $\actuatorfun'{}(\mathit{cool}) = \off$ (inactive cooling), and $[v_1, v_2] =[-1{-}\delta \, , \, -1{+}\delta]$,
if $\actuatorfun'{}(\mathit{cool}) = \on$ (active cooling);
\item[(ii)] 
$\measmap{}(\statefun'{}) = 
\sum_{v \in   [ -\mathit{err}, +\mathit{err}]_g}  \frac{1}{| [ -\mathit{err}, +\mathit{err}]_g|} \cdot \dirac{[\mathit{s_t} \mapsto  \xi'_{\operatorname{x}}(temp) + v]}$,
for any $\xi'_{\operatorname{x}}\in \mathbb{R} ^{\{temp\}}$; 
%
%%%\item[(iii)] 
%%%$\invariantfun{} = \{ x \in \mathbb{R} \colon 0 \leq x \leq 30   \}$. 
%
\item[(iii)] 
$\invariantfun{} = \{ [ \mathit{temp} \mapsto x] \colon x \in \mathbb{R} \: \text{ and } \: 0 \leq x \leq 30   \}$. 
\end{itemize}

The \emph{cyber component} of the engine consists of a process $\mathit{Ctrl}$ 
which models the controller activity. Intuitively, process $\mathit{Ctrl}$ senses the temperature of the engine at each time interval. When the sensed temperature is above $10$, the controller activates the coolant. The cooling activity is maintained for $5$ consecutive time units. After that time, if the temperature does not drop below $10$ then the controller transmits its $\mathit{ID}$ on a specific channel for signalling a  $\mathit{warning}$, it keeps cooling  for another $5$ time units, and then checks again the sensed temperature; otherwise, if the 
sensed temperature is not above the threshold $10$, the controller turns off the cooling and moves to the next time interval.
 Formally,
\begin{displaymath}
\begin{array}{rcl}
\mathit{Ctrl} & \; = \; & \fix{X} \rsens x {s_{\operatorname{t}}} . 
\ifelse {x>10}
 { \mathit{Cooling}}
 {  \tick.X } \\[1pt]
\mathit{Cooling} & \;  = \;  &   \wact{\on}{\mathit{cool}}. \fix{Y} 
 \tick^5 .   \rsens x {s_{\operatorname{t}}} . \\
&&
\ifelse {x>10} {\OUT{\mathit{warning}}{\mathrm{ID}}.Y}
{\wact{\off}{\mathit{cool}}.\tick.X } \enspace . 
 \end{array}
\end{displaymath}

The \emph{whole engine} is defined as: 
\begin{math}
\mathit{Eng}_g \: = \:  \confCPS {\env_g;\state_g} { \mathit{Ctrl} } \,,
\end{math}
where $\env_g$ and $\state_g$ are the physical environment and the 
physical state defined before.

Our operational semantics allows us to formally prove a number of 
\emph{run-time properties\/} of our engine. 
For instance,  the following proposition says that our engine 
never reaches a warning state and  never deadlocks.  
\begin{proposition} 
\label{prop:sys}
Let $\mathit{Eng}_g$ be the \CPS{} defined before. 
Given any computation $\mathit{Eng}_g \transS[\alpha_1] \ldots \transS[\alpha_n] M$, then  $\alpha_i \in \{ \tau , \tick  \}$, for $1 \leq i \leq n$, and
 there is a distribution $\gamma$ such that $M \trans{\alpha} \gamma$, for some 
$\alpha \in \{ \tau , \tick  \}$. 
\end{proposition}

Actually, knowing that in each of the $5$ time slots of cooling, the temperature will drop of a value laying in the
interval $[1{-}\delta, 1 {+} \delta]_g$, we can be quite precise on the temperature reached by the engine before and after the cooling activity. 
Formally:
\begin{proposition}
\label{prop:X}
 Let $\mathit{Eng}_g \trans{\alpha_1} \ldots \trans{\alpha_n} M$ be an 
arbitrary computation of the engine, for some \CPS{} $M$: 
\begin{itemize}[noitemsep]
\item if $M$ turns  the cooling on  then the value of
the state variable $\mathit{temp}$ in $M$ ranges over $(9.9  ,  11.5]$;
\item if $M$ turns  the cooling off then the value of
the 
variable $\mathit{temp}$ in $M$ ranges over $(2.9  ,   8.5]$. 
\end{itemize}
\end{proposition}
The proofs of both propositions can be found in the Appendix,  in \autoref{app:sec:case-study}. 

\begin{figure}[t]
\centering
\includegraphics[width=7cm,keepaspectratio=true,angle=0]{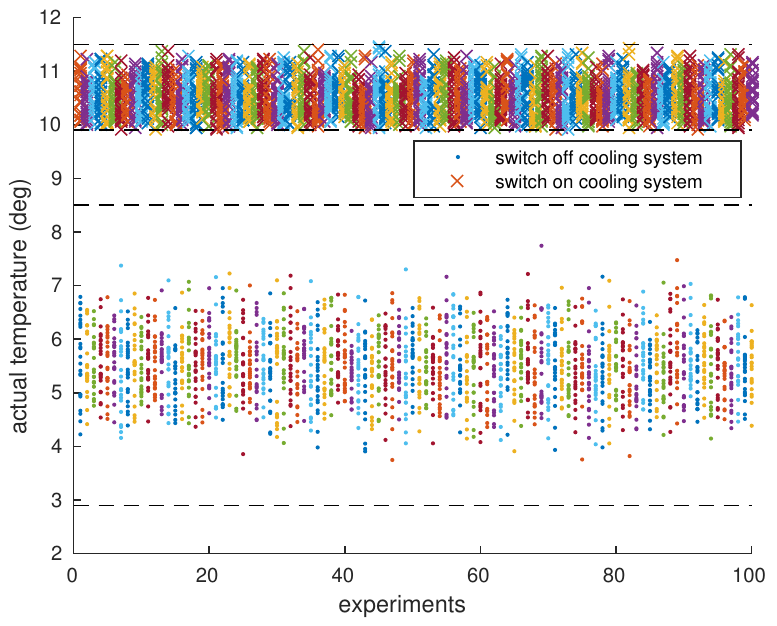}
\Q\Q\Q
\includegraphics[width=7cm,keepaspectratio=true,angle=0]{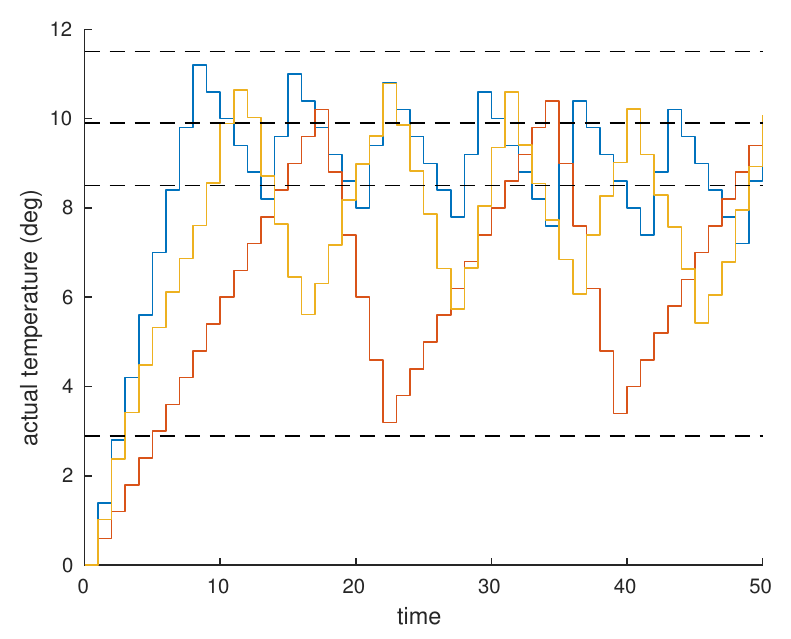}
\caption{Simulations in MATLAB of the engine $\mathit{Eng}$}
\label{f:HS traj}
\end{figure}

The result formally proved in  \autoref{prop:X} finds a correspondence 
in the left graphic of \autoref{f:HS traj}. In that graphic, we collect 
 a campaign of 100 simulations of our engine 
in MATLAB\footnote{MATLAB chooses a value in a real interval
by means of a discrete uniform distribution depending on the granularity  imposed by the finite number of bits used for the representation of floats. }, lasting 250 time units each, showing that the value of the
state variable $\mathit{temp}$ when the cooling system is turned on
(resp., off) lays in the interval $(9.9, 11.5]$ (resp., $(2.9,8.5]$);
these bounds are represented by the dashed horizontal lines. Obviously, when dealing with complex systems even several thousands of simulations do not ensure
the absence of incorrect states, as formally proved in \autoref{prop:sys}  and \autoref{prop:X}.

The right graphic of the same figure shows three  possible evolutions in time of
the state variable $\mathit{temp}$: (i)  the first one (in red),  in which  
the temperature of the engine always grows of $1-\delta = 0.6$ degrees per time step, when the cooling is off, and always decrease of $1+\delta=1.4$ degrees per time unit, when the cooling is on; (ii) the second  
one (in blue), in which the temperature always grows of $1+\delta=1.4$ degrees per time unit, when the cooling is off,  and  always decrease of $1-\delta=0.6$ degrees per time unit, when the cooling is on; (iii) and a third one (in yellow), in which, depending whether the cooling is off or on, at each time step the temperature grows or decreases of an arbitrary
offset laying in the interval $[1-\delta , 1+\delta]$. 

Now, the reader may wonder whether it is possible to design a variant 
of our engine which meets the same specification with better 
performances. 
For instance, an engine consuming less coolant. 
Let us consider a variant of the engine described before: 
\begin{center}
\begin{math}
\widetilde{\mathit{Eng}_g} \; = \; \confCPS {{\widetilde{\env_g}};\state_g} { \mathit{Ctrl} } \, . 
\end{math}
\end{center}
Here,   $\widetilde{\env_g}$ is the same as $\env_g$ except for the  evolution map, as  we set   $[v_1,v_2] =[-0.8{-}\delta \, , \, -0.8{+}\delta]$ 
if $\actuatorfun'{}(\mathit{cool}) = \on$ (active cooling).
 This means that in $\widetilde{\mathit{Eng}_g}$ we reduce the power  of the cooling system by $20\%$. In  \autoref{fig:consumption}, we report the results of our simulations  in MATLAB over $10000$ runs lasting $10000$ time units each. From this graph,  $\widetilde{\mathit{Eng}_g}$ saves in average more than $10\%$ of coolant with respect to $\mathit{Eng_g}$. So, the new question is: are these two engines behavioural equivalent? 
Do they meet the same specification?

Our  bisimilarity  provides us with 
a precise answer to these  questions: the two variants of the engine are bisimilar. 
\begin{proposition}
\label{prop:performances} 
 $\mathit{Eng}_g \approx \widetilde{\mathit{Eng}_g}\,$, for any $g\in \mathbb{N}^{+}$.
\end{proposition}

The proof can be found in the Appendix, in \autoref{app:sec:case-studybis}. 

At this point, one may wonder whether it is possible to improve the performances 
of our engine even more. For instance,  by  reducing the power of the cooling system by a further $10\%$, by setting  $[v_1 , v_2] =[-0.7{-}\delta \, , \, -0.7{+}\delta]$ if $\actuatorfun'{}(\mathit{cool}) = \on$ (active cooling). We can formally prove that
this is not possible. 
\begin{proposition}
\label{prop:stop} Let $\widehat{\mathit{Eng}_g}$ be the  same as  $\mathit{Eng}_g$, except for the evolution map,  in which the real interval \( [v_1, v_2] \) 
is given by $[-0.7{-}\delta
\, , \, -0.7{+}\delta] $
if $\actuatorfun'{}(\mathit{cool}) = \on$.
Then, $\mathit{Eng}_g \not\approx 
\widehat{\mathit{Eng}_g}\, $, for any $g\in \mathbb{N}^{+}$.
\end{proposition}
The proof can be found in the Appendix, in \autoref{app:sec:case-study}.

Finally, we show  how we can use the 
compositionality of our 
behavioural semantics  (\autoref{thm:congruence}) to deal with bigger \CPS{s}. 
Suppose that $\mathit{Eng}_g$ denotes the  model in our calculus of an airplane
engine. In this case, we could model a very simple 
\emph{airplane control system} that checks whether
the left engine ($\mathit{Eng}_g^{\mathrm{L}}$) and the right engine
($\mathit{Eng}_g^{\mathrm{R}}$) are signalling warnings. 
The whole \CPS{} is defined as follows:
\begin{displaymath}
\mathit{Airplane}_g  \; = \;  \big( ( \mathit{Eng}_g^{\mathrm{L}} \uplus
\mathit{Eng}_g^{\mathrm{R}}  )  \parallel \mathit{Check} \big) \backslash\{warning \}  
\end{displaymath}%
where {\small $\mathit{Eng}_g^{\mathrm{L}} = \mathit{Eng}_g
\subst{\mathrm L}{\mathrm{ID}}
\subst{\mathit{temp{\_}l}}{\mathit{temp}}
\subst{\mathit{cool{\_}l}}{\mathit{cool}} 
\subst{s_{\mathrm{t}{\_}l}}{s_{\mathrm{t}}}$}, and 
{\small 
$\mathit{Eng}_g^{\mathrm{R}} = \mathit{Eng}_g
\subst{\mathrm R}{\mathrm{ID}}
\subst{\mathit{temp{\_}r}}{\mathit{temp}}
\subst{\mathit{cool{\_}r}}{\mathit{cool}} 
\subst{s_{\mathrm{t}{\_}r}}{s_{\mathrm{t}}}$}, 
and process $\mathit{Check}$ is defined as follows:  
\begin{center}
\begin{math}
\begin{array}{rcl}
\mathit{Check} & = &  \fix{X}
\timeout{\LIN{\mathit{warning}}{x }.
\ifelse {x = {\mathrm{L}}} {\mathit{Check}^{\mathrm L}_1}
{\mathit{Check}^{\mathrm R}_1}}{X}\\[4pt]
\mathit{Check}^{\mathit{id}}_i & = & 
\timeout{\LIN{\mathit{warning}}{y }.
\ifelse {y \neq {\mathit{id}}} {\OUTCCS{\mathit{alarm}}.\tick.X}
{\tick. \mathit{Check}^{\mathit{id}}_{i+1}}}
{\mathit{Check}^{\mathit{id}}_{i+1}}\\[4pt]
\mathit{Check}^{\mathit{id}}_5  & = & \lfloor {\LIN{\mathit{warning}}{z }.
\ifelse {z \neq {\mathit{id}}} {\OUTCCS{\mathit{alarm}}.\tick.X}
{ \OUT{\mathit{failure}}{\mathit{id}}.\tick. X}} \rfloor \\
&&
 {\OUT{\mathit{failure}}{\mathit{id}}.X}
\end{array}
\end{math}
\end{center}
for $1 \leq i \leq 5$. 
Intuitively, if one of the two engines is in a warning state then the 
process $\mathit{Check}^{\mathit{id}}_i$, for ${\mathit{id}} \in \{ \mathrm{L}, 
\mathrm{R}\}$,   checks whether also the second engine moves into a warning state,  in the following 
 $5$ time intervals (i.e.\ during the cooling cycle). If both engines get in a 
 warning state then  an $\mathit{alarm}$ is 
sent, otherwise, if only one  engine is facing a  warning then
 the airplane control system yields a \emph{failure} signalling which engine 
is not working properly.

So, since we know that $\mathit{Eng}_g \approx \widetilde{\mathit{Eng}_g}\,$,  for any $g\in \mathbb{N}^{+}$,
 the final question becomes the following: can we safely equip our airplane with the  more performant engines, $\widetilde{\mathit{Eng}_g^{\mathrm{L}} }$ and $\widetilde{\mathit{Eng}_g^{\mathrm{R}} }$, in which 
$[v_1, v_2] =[-0.8{-}\delta \, , \, -0.8{+}\delta]$,
if $\actuatorfun'{}(\mathit{cool}) = \on$,  without affecting the
whole observable 
behaviour of the airplane?
The answer is  ``yes'', and this result can be formally  proved 
by relying on \autoref{prop:performances} and \autoref{thm:congruence}. 
\begin{proposition}
\label{prop:air}
Let \begin{math}
\widetilde{\mathit{Airplane}_g}  \: = \:   \big( (\widetilde{\mathit{Eng}_g^{\mathrm{L}} } 
\uplus 
 \widetilde{\mathit{Eng}_g^{\mathrm R}}) \parallel \mathit{Check} \big)
\backslash \{warning\} 
\end{math}. Then,  
$\mathit{Airplane}_g \approx \widetilde{\mathit{Airplane}_g} \, $.
\end{proposition}
%
%The proof of \autoref{prop:air} can be found in the Appendix, in \autoref{app:sec:case-study}.

\begin{figure}[t]
\centering
\includegraphics[width=8cm,keepaspectratio=true,angle=0]{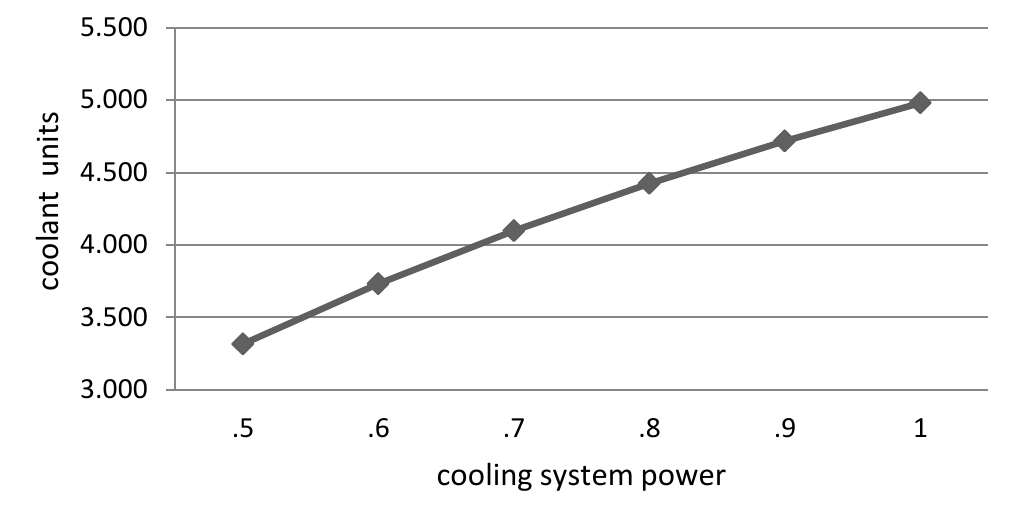}
\caption{Simulations in MATLAB of coolant consumption}
\label{fig:consumption}
\end{figure}

We end this section with an observation. Although, 
  the engine $\widehat{\mathit{Eng}_g}$  is
not behavioural equivalent to the original engine $\mathit{Eng}_g$, 
an airplane maker might be interested in knowing  an 
estimation of the deviation of its behaviour with respect 
the behaviour of the original engine. If this deviation would be very small then 
aeronautical engineers might consider to adopt in their airplanes 
the engine $\widehat{\mathit{Eng}_g}$ instead $\widetilde{\mathit{Eng_g}}$ to save even more coolant. So, the new question is: 
how big is the deviation, in terms of behaviour, of the 
engine $\widehat{\mathit{Eng}_g}$  with respect to
 the original engine $\mathit{Eng}_g$? 

The rest  of the paper is devoted to develop general quantitative 
techniques to estimate the deviation of the probabilistic behaviour of a \CPS{} with respect to another. 
%
%% PARTE NUOVA
%

%%%%%%%%%%%%%%%%

%%%%%%%%%%%%%%%%%%%%%%%%%%%%%%%%%%%%%%%%%%%%%%%%%%%%%%%%%%%%%%%%%%%%

\section{Bisimulation metrics}
\label{sec:metric}

In this section, we provide a weak behavioural distance to compare the probabilistic behaviour of \CPS{s} up to a given approximation. To this end, we
 adapt the notion of \emph{weak bisimilarity metric}~\cite{DJGP02} to \cname{}. 
Intuitively, we will write $M \approx_{p} N$ if the weak bisimilarity between $M$ and $N$ holds with a \emph{distance} $p$,  with $p \in [0,1]$.
Thus, $\approx_{0}$ will coincide with the weak 
probabilistic bisimilarity of  \autoref{def:bisimulation}, 
whereas $\bigcup_{p \in [0,1]} \approx_{p}$ will correspond to the cartesian product  $\cname{} \times \cname{}$.

Weak bisimilarity metric is  defined as a \emph{pseudometric} measuring the tolerance of the probabilistic weak bisimilarity. 
\begin{definition}[Pseudometric]
\label{def:metric}
A function $d \colon \nome \times \nome \to [0,1]$ is said to 
be a \emph{1-bounded pseudometric} if
\begin{itemize} 
	\item $d(M,M)= 0$, for all $M \in \nome$, 
         \item $d(M,M')= d(M',M)$, for all $M,M' \in \nome$, 
	\item $d(M,M') \le d(M,M'') + d(M'',M')$, for all $M,M',M''\in \nome$.
\end{itemize}
\end{definition}
 
Weak bisimilarity metric provides the quantitative analogous of the weak bisimulation game: two \CPS{s} $M$ and $N$ at  distance $p$ can mimic each other transitions and evolve to distributions $\gamma$ and $\gamma'$, respectively, 
 placed at some distance $q$, with $q \leq p$. 
%%not greater than the distance between $M$ and $M'$.
This requires to lift pseudometrics from \CPS{s} to distributions over \CPS{s}.
%We need to lift this definitions to (sub)distributions.
To this end, as in \cite{LMT17}, we rely on the notions of 
\emph{matching}~\cite{Vil08} and \emph{Kantorovich lifting}~\cite{K42}.\footnote{The original formulation of weak bisimulation metric~\cite{DJGP02} is technically different but %%substantially 
equivalent to our definition~\cite{DENG200973}.}

In  \autoref{def_matching}, we  already provided the definition of matching. 
Let us define the Kantorovich lifting.
\begin{definition}[Kantorovich lifting] \label{def:KantorovichLifting}
\label{def:Kantorovich}
Let $d\colon \nome \times \nome  \to [0,1]$ be a pseudometric. The \emph{Kantorovich lifting} of $d$ is the function
$\Kantorovich(d)\colon \distr{\nome} \times \distr{\nome} \to [0,1]$ defined 
as: 
\begin{displaymath}
\Kantorovich(d)(\gamma,\gamma') \, =  \, \min_{\omega \in \Omega(\gamma,\gamma')} \sum_{M,M' \in \nome}\omega(M,M') \cdot d(M,M') 
\end{displaymath}
for all $\gamma, \gamma' \in \distr{\nome}$.
\end{definition}

Note that since we are considering only distributions with finite support, the minimum over the set of matchings $\Omega(\gamma,\gamma')$ is well defined. 

\begin{definition}[Weak bisimulation metric]
\label{def:simulation_metric}
We say that a pseudometric $d \colon \nome \times \nome \to [0,1]$ is a \emph{weak bisimulation metric} if for all $M,N \in \nome$, with $d(M,N)<1$, whenever $M \transS[\alpha] \gamma$ there is a sub-distribution $\gamma'$ such that $N \TransS[\hat{\alpha}] \gamma'$  and $\Kantorovich(d)(\gamma \, , \,   \gamma' + (1{-}\size{\gamma'}) \dirac{\dummyN}) \le d(M,N)$. 
\end{definition}
Note that in the previous definition, if $\size{\gamma'} < 1$ then, with probability $1-\size{\gamma'}$, there is no way to simulate the behaviour of any \CPS{} with a valid invariant in the support of $\gamma$ (the special \CPS{} $\dummyN$ does not perform any action).

A crucial result is the existence of the minimal weak bisimulation metric~\cite{DJGP02}, called \emph{weak bisimilarity metric}, and denoted with $\metric$.
We remark that in \cite{DJGP02} it is shown that the kernel of $\metric$ 
coincides with the definition of  weak probabilistic bisimilarity.
\begin{proposition}
\label{prop:kernel}
For all  $M, N \in \cname{}$ we have $\metric(M,N) = 0$ if and only if $M \approx N$.
\end{proposition}

Now, we have all ingredients to define our notion of behavioural  distance between \CPS{s}. 
\begin{definition}[Distance between \CPS{s}]
\label{def:distance}
Let $M,N\in \cname{}$ and $p \in [0,1]$. 
We say that \emph{$M$ and $N$ have distance $p$}, written $M \approx_p N$, if 
and only if $\metric(M,N) = p$.
\end{definition}

In the next section, we will use a more refined notion of distance that considers only the first $n \in \mathbb{N}$ computation steps,  when comparing 
two \CPS{s}.

Such definition requires the  introduction of a complete lattice  ${([0,1]^{\nome \times \nome},\sqsubseteq)}$  of functions of type $ \nome \times \nome \to [0,1]$
ordered by $d_1 \sqsubseteq d_2$ iff $d_1(M, N) \le d_2(M,N)$ for all $M,N \in \nome$, 
where for each set $D \subseteq [0,1]^{ \nome\times  \nome}$ the supremum and infimum are defined as
$\sup(D)(M,N) = \sup_{d \in D}d(M,N)$ and 
$\inf(D)(M,N) = \inf_{d \in D}d(M,N)$, 
for all $M,N \in \nome$.
Notice that the infimum of the lattice is the constant function zero, which we denote by $\zeroF$.

We also need a functional $\Bisimulation$ defined over the lattice mentioned above such that $\Bisimulation(d)(M,N)$ returns the minimum possible value for
 $d(M,N)$ in order to ensure that $d$ is a weak bisimulation metric.
\begin{definition}[Bisimulation metric functional] 
\label{def:metric_sim_functional}
Let $\Bisimulation \colon [0,1]^{\nome\times \nome } \to [0,1]^{ \nome  \times \nome}$ be the functional  such that for any $d \in [0,1]^{\nome\times \nome }$ 
 and $M,N \in \nome$, $\Bisimulation(d)(M,N)$   is 
given by: 
\begin{displaymath}
{\scriptsize
\sup_{\{\alpha \, \colon \,   M \transS[\alpha] \: \vee \: N \transS[\alpha]\}} 
\max \left
\{ \max_{M \transS[\alpha]\gamma_1} \min_{N \TransS[\hat{\alpha}] \gamma_2} \Kantorovich(d)(\gamma_1,\gamma_2 + (1-\size{\gamma_2})\dirac{\dummyN}), 
\max_{N \transS[\alpha] \gamma_2}\min_{M \TransS[\hat{\alpha}]\gamma_1} \Kantorovich(d)(\gamma_1+ (1-\size{\gamma_1})\dirac{\dummyN},\gamma_2)  \right
\}
}
\end{displaymath}
 where $\max \emptyset = 0$ and $\min \emptyset = 1$.
\end{definition}

Notice that \autoref{def:metric_sim_functional} and \autoref{def:simulation_metric} are strictly related as  weak bisimulation metrics are pseudometrics that are prefixed points of $\Bisimulation$. 
Notice also that all $\max$ and $\min$ in \autoref{def:metric_sim_functional} are well defined since our pLTS is image finite and \CPS{s} enjoy the well timedness property.

Since $\Kantorovich$ is monotone \cite{Pan09} it follows that $\Bisimulation$ is a monotone function on $([0,1]^{\nome  \times \nome },\sqsubseteq)$.
Furthermore, since this structure is a lattice, by Knaster-Tarski theorem it follows that $\Bisimulation$ has a least prefixed point (which is also the least fixed point).
Later we will show that this least prefixed point coincides with $\metric$.

Now,  we exploit the functional $\Bisimulation$ to introduce a notion of 
\emph{$n$-weak bisimilarity metric}, denoted $\metric^n$, which intuitively quantifies the tolerance of the weak bisimulation in $n$ steps.
The idea is that  $\metric^0$ coincides with the constant function $\zeroF$ assigning distance $0$ to all pairs of \CPS{s}, whereas 
 $\metric^n(M,N)$, for $n>0$, is defined as  $\metric^n(M,N) = \Bisimulation(\metric^{n-1})(M,N)$.
Thus, the $n$-weak bisimilarity metric between $M$ and $N$ is defined in terms of the $(n{-}1)$-weak bisimilarity metric between the distributions reached (in one step) by $M$ and $N$, respectively.

%\remarkM{Se qualcuno ha gia' definito questa up to, anche in versione strong, va menzionato.}
\begin{definition}[$n$-weak bisimilarity metric]
Let $n \in \mathbb{N}$. 
The function $\Bisimulation^n(\zeroF)$, abbreviated as $\metric^n$, is called  \emph{$n$-weak bisimilarity metric}.
\end{definition}

\begin{proposition}
\label{prop:up-to-k-metric}
For all $n \ge 0$, $\metric^n$ is a 1-bounded pseudometric.
\end{proposition}
The proof of this proposition can be found in Appendix, in \autoref{Sec:A.5}.

Finally, we are ready to define our notion of $n$-distance between two \CPS{s}. 
\begin{definition}[$n$-distance between \CPS{s}]
\label{def:distance-n}
Let $M, N \in \cname{}$ and  $p \in [0,1]$. We say that \emph{$M$ and $N$ have $n$-distance $p$}, written $M \approx^n_p N$, if and only if $\metric^n(M,N) = p$.
\end{definition}

Since our pLTS is %finite branching, and hence 
image-finite, 
%\remarkM{Image-finite e finitely branching mi pare di ricordare che siano definizioni leggermente diverse.}
and all transitions lead to distributions with finite support, it is possible to prove that $\Bisimulation$ is continuous~\cite{vB12}.
Since $\Bisimulation$ is also monotone, we can deduce that the closure ordinal of $\Bisimulation$ is $\omega$ (see Section 3 of \cite{vB12}). 
As a consequence, the $n$-weak bisimilarity metrics converge to the weak bisimilarity metric when $n$ grows indefinitely. %% \to \infty$. 
Formally, 
\begin{proposition}
\label{prop:metric_as_a_limit}
%%The functions $\metric$ and $\metric^k$ are related by 
$\metric = \lim_{n \to \infty}\metric^n$.
\end{proposition}

%%%%%%%%%%%
\begin{comment}
\begin{proposition}
\label{prop:metric_is_a_metric}
The function $\metric$ is the minimal bisimulation metric.
\end{proposition}
\end{comment}
%%%%%%%%%

Last but but not least,  the distances introduced in 
 \autoref{def:distance} and 
\autoref{def:distance-n} 
 allow us to compare \CPS{s} in a compositional manner. 
In particular, these distances are preserved by parallel composition of 
physical-disjoint \CPS{s}, 
by parallel composition of pure-logical processes, and by channel restriction. 
\begin{theorem}[Compositionality of distances]
	Let $M$ and $N$ be two arbitrary \CPS{s} in \cname{}. 
	\label{thm:congruenceP}
	\begin{enumerate}
		\item
		\label{thm:congruence1P}
		$M \approx_p N$ implies $M \uplus O \approx_q N \uplus O$, with $q \le p$, for any 
		physically-disjoint \CPS{} $O$; 
		\item
		\label{thm:congruence2P}
		$M \approx_p N$ implies $M \parallel P \approx_q N \parallel P$, with $q \le p$, for any  
		pure-logical process $P$; 
		\item 
		\label{thm:congruence3P}
		$M \approx_p N$ implies $M \backslash c \; \approx_q \; M \backslash c$, with $q \le p$, for any channel $c$; 
		\item
		\label{thm:congruence1Pk}
		$M \approx^n_p N$ implies $M \uplus O \approx^n_q N \uplus O$, with $q \le p$, for any
		physically-disjoint \CPS{} $O$ and any $n \ge 0$; 
		\item
		\label{thm:congruence2Pk}
		$M \approx^n_p N$ implies $M \parallel P \approx^n_q N \parallel P$, with $q \le p$, for any  
		pure-logical process $P$ and any  $n \ge 0$; 
		\item 
		\label{thm:congruence3Pk}
		$M \approx^n_p N$ implies $M \backslash c \; \approx^n_q \; M \backslash c$, with $q \le p$, for any channel $c$ and $n \ge 0$.
	\end{enumerate} 
\end{theorem}
The proof of Theorem~\ref{thm:congruenceP} can be found in the Appendix, 
in Section~\ref{Sec:A.5}. 

Now, suppose that $M \approx_p N$, $M' \approx_{p'} N'$, with  $M$ (resp.\ $N$)
and $M'$ (resp.\ $N'$)  physically-disjoint.  
By Theorem~\ref{thm:congruenceP}.\ref{thm:congruence1P},  we can infer  both
$M  \uplus  M' \approx_{q} N \uplus M'$ and $N \uplus M' \approx_{q'} N \uplus N'$, with $q \le p$ and $q' \le p'$.
Then, by triangular property of the pseudometric $\metric$ we get $M  \uplus  M' \approx_{q''} N \uplus N'$,  for some $q'' \le  q + q' \le p+p'$.
Similarly, by applying
Theorem~\ref{thm:congruenceP}.\ref{thm:congruence1Pk} we can infer that
$M \approx^n_p N$ and $M' \approx^n_{p'} N'$ entail $M  \uplus  M' \approx^n_{q} N  \uplus  N'$,  for some $q \le p + p'$. This says that our metrics enjoy  a well-known  compositional property called \emph{non-expansiveness}~\cite{DGJP04,GLT16,GT18}. 

In the next section, the compositional properties of \autoref{thm:congruenceP} will be very useful when reasoning on our case study.

%%%%%%%%%%%%%%%%%%%%%%%%%%%%%%%%%%%%%%%%%%%%%%%

\section{Case study, reloaded}
\label{sec:casebis}

In \autoref{sec:case-study}, we proved that the original version
of the proposed engine,  $\mathit{Eng_g}$, and its  variant  $\widetilde{\mathit{Eng_g}}$ (saving up to $10\%$ of coolant) are behavioural equivalent (i.e., bisimilar). Then, by relying on the 
 compositionality of our probabilistic bisimilarity (\autoref{thm:congruence}), 
we proved that the two compound systems, $\mathit{Airplane}_g$ and  $\widetilde{\mathit{Airplane}_g}$, mounting  engines $\mathit{Eng_g}$ and $\widetilde{\mathit{Eng_g}}$, respectively,  are bisimilar as well. 

Actually, both results can be proved in terms of 
weak probabilistic metric with distance $0$, as this
specific metric coincides  with the probabilistic bisimilarity (\autoref{prop:kernel}). 
\begin{proposition}
\label{prop:case-propbis}
\label{prop:air2}
Let  $g \in \mathbb{N}^{+}$. Then, 
\begin{itemize}[noitemsep] 
\item  $\mathit{Eng}_g \: \approx_0  \: \widetilde{\mathit{Eng}_g}$ 
\item  $\mathit{Airplane}_g \: \approx_0  \: \widetilde{\mathit{Airplane}_g}
\, .$
\end{itemize}
\end{proposition}

Then, in  \autoref{sec:case-study} we moved our attention to 
a  more performant engine,  $\widehat{\mathit{Eng_g}}$,  saving 
almost  $20\%$ of coolant with respect to the original engine $\mathit{Eng_g}$.
In our behavioural analysis we rejected this new variant  as it 
may exhibit a different probabilistic behaviour 
when compared to  $\mathit{Eng_g}$. More precisely, 
the two systems   $\mathit{Eng_g}$ and $\widehat{\mathit{Eng_g}}$ \emph{are not} bisimilar (\autoref{prop:air}). 

However, in many complex probabilistic systems, such as \CPS{s}, probabilistic  bisimilarity might reveal to be too strong as the natural behavioural equivalence to take systems apart. Thus, in  \autoref{sec:case-study} we advocated for some appropriate notion of behavioural distance to estimate the effective difference, in terms of behaviour, of these two versions of the  engine. 

In the current section, we apply the bisimulation metrics defined in \autoref{sec:metric}
to estimate the distance between $\mathit{Eng}_g$ and $\widehat{\mathit{Eng}_g}$, by varying 
the granularity $g \in \mathbb{N}^{+}$. 
In particular, we apply the notion of $n$-weak bisimilarity metric.  

\begin{proposition}
\label{prop:case-prop}
Let $g \in \mathbb{N}^{+}$ and $n \in \mathbb{N}$. Then, 
for $p_g= \frac{\mid [0.3, 0.4)_g \mid}{\mid [0.3, 1.1]_g \mid }$ and
$q_g= \frac{\mid (1.3, 1.4] _g \mid}{\mid [0.6, 1.4]_g \mid }$, we have: 
\[\metric^n(\mathit{Eng_g}, \widehat{\mathit{Eng_g}}) \; \le \;   1- \left(1- q_g(p_g)^5\right)^n \enspace. \]  
\end{proposition}
Note that if the cooling system of $\widehat{\mathit{Eng_g}}$ is off and it is not going to be activated in the current time slot, then the sensed 
temperature is below than or equal to $10$, and  the real temperature  is below than or equal to $10.1$ degrees (we recall that $\mathit{err}=0.1$). 
Assume that the temperature is exactly $10.1$.
If in the current time slot the temperature increases of a value  $ v\in (1.3, 1.4]$ then it will  reach a value in the interval  $(11.4, 11.5]$  (we recall that $\delta=0.4$).
This happens with a probability bounded by $q_g$.
In this case, the cooling system will be turned on, and 
the temperature will drop,  in each of the following $5$  time slots, of some value laying in the interval $ [0.7{-}\delta \, , \, 
0.7{+}\delta]=[0.3,1.1]$. 
However, if  in each of those $5$ slots of cooling the temperature 
is decreased of a value laying in 
$ [0.3, 0,4)$, then the cooling activity might not be enough to 
avoid (observable) warnings, and the two engines $\mathit{Eng_g}$ and $ \widehat{\mathit{Eng_g}}$ will be distinguished. 
Thus, $p_g$ is given by the number of possible ``bad decreases'', $|[0.3, 0.4)_g|$, divided by 
the number of all possible decreases, $\mid [0.3, 1.1]_g \mid$; whereas
$q_g$ is given by the number of possible ``bad increases'', 
$\mid (1.3, 1.4] _g \mid$, divided by 
the number of all possible increases $\mid [0.6, 1.4]_g \mid$.

Notice that $p_g$ and $q_g$ refer to real intervals which are basically 
shifted. Thus,  we have that 
\linebreak
 $ \mid [0.3, 0.4)_{g} \mid \, = \, \mid (1.3, 1.4] _g \mid \, =   10^{g -1} $  and  $  \mid [0.3, 1.1]_{g} \mid \, =  \, \mid [0.6, 1.4]_g \mid \, = \, 8 \cdot 
 10^{g-1} +1$. 
As a consequence,  
$p_g= q_g = \frac{10^{g-1}}{8 \cdot 10^{g -1} +1}  =  \frac{1 }{8+ 10^{-g+1}}$.
 Obviously,  the finer is the granularity $g$ the closer is the value of 
 $p_g$ and $q_g$ to $\frac{1}{8}$. Formally, 
\begin{equation}
\label{eq:lim}
\lim_{ g \rightarrow \infty}  \metric^n(\mathit{Eng_g}, \widehat{\mathit{Eng_g}})
\;  \leq \;  1-\big(1- \frac{1}{8^6}  \big)^n  \enspace . 
\end{equation}

 Thus, for instance, assuming a granularity $g=6$, after $n= 3000$ computation 
 steps the distance between the two systems is less than $0.012$.
 Intuitively, this means that if we limit our analysis to 
$3000$ computation steps the behaviours of two engines may differ 
 with probability at most $0.012$.
By an easy inspection in the (common) logics of the two engines, it is 
easy to see that any two subsequent $\tick$-actions are  separated by  at most $2$ untimed actions. 
Thus, $3000$ computation steps means around $1000$ time slots.
Considering time slots lasting $20$ seconds each, this means
more than five hours. Thus,  an 
utilisation of $\widehat{\mathit{Eng_g}}$  might be feasible in airplanes used for  short-range flights, where the engine is actually used for 
a limited amount of  time.   Actually, aeronautical engineers might consider
perfectly acceptable the risk of mounting the engine $\widehat{\mathit{Eng_g}}$ instead of $\mathit{Eng_g}$,  when compared to the reliability of the other components of the airplane.  

However, since an airplane mounts two engines,   engineers need to estimate
the difference in terms of  behaviour on the whole airplane resulting 
by the adoption of different versions of the engine. This is exactly the point where we can rely on  \autoref{thm:congruenceP} to support compositional reasoning. 

 The following result follows from 
\autoref{eq:lim}, \autoref{prop:case-prop} and \autoref{thm:congruenceP}. 
\begin{proposition}
\label{prop:air3}
Let $g \in \mathbb{N}^{+}$ and $n \in \mathbb{N}$. 
Let 
\( 
\widehat{\mathit{Airplane}_g}  \; = \;  \big( ( \widehat{\mathit{Eng}_g^{\mathrm{L}} } 
\uplus 
(  \widehat{\mathit{Eng}_g^{\mathrm R}} ) \parallel \mathit{Check} \big)
\backslash \{warning\} \, . \) Then, 
\begin{enumerate}
\item
\label{prop:air3.1}
$\metric^n( \mathit{Airplane}_g  \: , \: \widehat{\mathit{Airplane}_g}  ) \: \leq \: 2p  $,
 where  $p= 1- \big(1- q_g(p_g)^5\big)^n $ 
\item  
\(
\lim_{ g \rightarrow \infty}   \metric^n( \mathit{Airplane}_g  \: , \: \widehat{\mathit{Airplane}_g}  ) \: \leq \: 2\big(1- \left(1- \frac{1}{8^6}  \big)^n\right) . 
\)
\end{enumerate}
\end{proposition}
Thus, for $g=6$,  
the probability that the two airplanes mounting different engines
exhibit a different behaviour within  $n=3000$ computation steps
 is at most $0.024$; a distance 
which may be considered still 
acceptable in specific contexts. Notice that 
 in the (common) logics of the two airplanes, it is 
easy to see that two $\tick$-actions are  separated by at most $5$
untimed actions (two for each engine plus one to signal a possible alarm). 
Thus, $3000$ computation steps means around $600$ time slots, i.e.,  more than 
three hours for time slots lasting $20$ second each. 

Finally, the reader should notice that the bound
of the distance between the two airplanes is given by the summation 
of the bounds of the distances between the two corresponding engines. This 
is perfectly in line with the fact that our bisimulation metrics enjoy the 
\emph{non-expansiveness} property.

The proofs of the previous propositions can be found in the Appendix, in \autoref{app:sec:case-studybis}.

\enlargethispage{.5\baselineskip}

%%%%%%%%%%%%%%%%%%%%%%
%

\section{Conclusions, related and future work}
\label{sec:conclusioni}

We have proposed a hybrid probabilistic process calculus, called  \cname{}, for specifying and reasoning on cyber-physical systems.
Our calculus allows us to model a \CPS{} by specifying its 
\emph{physical plant\/}, containing information on state variables, sensors, 
 actuators, evolution law, etc., and its \emph{logics}, i.e., controllers, IDSs, supervisors, etc.
Physical and logical components  interact through sensors and actuators, whereas interactions within the logics or between logics of different \CPS{s} rely
on channel-based communication. In 
\cname{}, the representation of the evolution map takes into account the uncertainty of the physical model, whereas the representation of the measurement map consider  measurement errors in sensor reading. As a consequence, the two maps returns discrete  probability distributions over state functions  and sensor functions, respectively.

 \cname{} is equipped with a probabilistic labelled transition semantics  which satisfies  classical time properties: \emph{time determinism}, \emph{patience}, \emph{maximal progress}, and \emph{well-timedness}.
As behavioural semantics we  adopt a natural notion of \emph{weak probabilistic bisimilarity} which is proved to be preserved by appropriate system contexts that are suitable for \emph{compositional reasoning\/}. 
Then, 
%% 
%At the end of the case study 
we  argue that probabilistic bisimilarity is only partially satisfactory to reason on \CPS{s} as it can only establish whether two \CPS{s} behave exactly in the same way. 
%%Any tiny variation of the probabilistic behaviour of one of the two systems under consideration will break the equality without any further information on the distance between their behaviours. 
To this end, we generalise our probabilistic bisimilarity to provide a notion of \emph{weak bisimulation metric}  along the lines of~\cite{DJGP02}. 
We also define a notion of weak bisimulation metric in $n$ steps, which 
reveals to be very effective whenever it is not necessary to observe the system ``ad infinitum''  but it is enough to observe its behaviour restricted to  bounded computations. Again, both bisimulation metrics are proved to be suitable for compositional reasonings. The paper provides a case study, taken from an engineering application, and use it to illustrate our definitions and  our compositional probabilistic behavioural theory for \cname{}.

\paragraph*{Related work.}
A number of approaches have been proposed for modelling hybrid systems using formal methods. 
For instance, \emph{hybrid automata}~\cite{ACHH1993} combine
finite state transition systems  (to model the cyber component) and continuous variables and dynamic  (to represent the physical component). 
A number of \emph{hybrid process algebras}~\cite{CuRe05,BergMid05,vanBeek06,RouSong03,HYPE} have been proposed for reasoning about physical systems and provide techniques for analysing and verifying protocols for hybrid automata. Among these approaches, \cname{} shares some similarities with the $\phi$-calculus~\cite{RouSong03}, a hybrid extension of the $\pi$-calculus~\cite{Mil91} equipped with a
weak bisimilarity that is  not compositional. 
Galpin et al.~\cite{HYPE} proposed a process algebra, called  HYPE,
in which the continuous part of the system is represented by  appropriate variables whose changes are 
determined by active influences (i.e., commands on actuators).
The authors define a strong bisimulation that extends the \emph{ic-bisimulation} of~\cite{BergMid05}. Unlike ic-bisimulation, the
bisimulation in HYPE is preserved by a notion of parallel 
composition that is slightly more
permissive than ours. However, bisimilar systems in HYPE must always have the 
same influence. Thus, in HYPE we cannot compare 
\CPS{s} sending different commands on actuators at the same time, 
as we do (for instance) in  \autoref{prop:performances}. 

In order to enrich hybrid models with probabilistic or stochastic behaviour,  a number of different approaches have been proposed in the last years \cite{Spro2000,Hu2000,Buj04,Abate08,Franzle2011,Hahn2013,Wang17}. 
Most of these approaches consist in introducing 
either probabilities in the transitions relation,  or probabilistic choice, or stochastic differential equations.  
For instance, in \emph{Stochastic Hybrid CSP} (SHCSP)  \cite{Wang17} probabilistic choice replaces non-deterministic choice, stochastic differential equations replace  differential equations, and communication interrupts are generalised by communication interrupts with weights. 

\enlargethispage{.2\baselineskip}
The formal analysis of probabilistic and stochastic systems follows the two classic mainstreams: (i)  \emph{model checking}
 (e.g.,  \cite{Abate08}) and \emph{reachability} 
(e.g., \cite{Spro2000,Abate08}), when the focus is on  a single system;  
\emph{behavioural equivalences} (e.g., \cite{LS91,Seg95,ALS00,BKHH02,BHK04,Bouj05}) 
when the goal is to compare the behaviour of two systems (very often, specification and implementation of the same system).
Al already said in the Introduction, probabilistic behavioural equivalences may be too strong in certain probabilistic and stochastic models in which  many interesting systems are only approximately behavioural equivalent. 
This led to several notions of \emph{behavioural distance} that can be grouped in two main families:   quantitative counterparts of  trace equivalence \cite{CK14,DAK12,DHKP16,Wn10}, and   quantitative counterparts of bisimulation equivalence \cite{DJGP02,DGJP04,BW05,DCPP06}.
We refer to \cite{A13,BA17} for a  comparison between these two approaches.
In the present paper,  we have adopted a bisimulation-based definition because, unlike trace semantics, bisimulation is sensitive to system deadlock, a phenomenon that has a great impact in   \CPS{s}.

%% Vigo, lo trattiamo anche troppo bene :)))
Vigo et al.~\cite{VNN13} proposed a  calculus for wireless-based 
cyber-physical systems 
endowed with a theory to study cryptographic primitives, together with explicit notions of communication failure and unwanted communication. 
 The calculus does not provide any notion of  behavioural equivalence. 
It also lacks a clear distinction between physical and logical 
components. 
%%Compared to~\cite{VNN13}, paper~\cite{WuZh15} 
%%introduces a static network topology and enrich the theory with an harmony theorem.

Lanese et al.~\cite{LBdF13} proposed an untimed calculus of mobile IoT devices interacting with the physical environment by means of sensors and  actuators. 
The calculus does not allow any representation of the physical environment, 
and it is equipped with an end-user bisimilarity in which end-users may: (i) provide values to sensors, (ii) check actuators, and (iii) observe the mobility of smart devices. End-user bisimilarity is not preserved by parallel composition. Compositionality is recovered by strengthening its discriminating power.

Lanotte and Merro~\cite{LM16} extended and generalised the work of~\cite{LBdF13}
in a timed setting by providing a %%fully abstract 
bisimulation-based 
semantic theory that is suitable for compositional reasoning. 
As in~\cite{LBdF13}, the  physical environment is not represented.

Bodei et al.~\cite{BDFG16,BDFG17} have proposed a new untimed 
process calculus, I{\small o}T-L{\small Y}S{\small A}, supporting  a control flow analysis that safely approximates the abstract behaviour of IoT systems. Essentially, they track how data spread from sensors to the logics of the network, 
and how physical data are manipulated. 
Intra-node generative communications in IoT-L{\small Y}S{\small A} are implemented through a shared store \`a la Linda~\cite{Linda}. In this manner physical data are made available to software entities that analyse them and trigger the relevant actuators to perform the desired behaviour. The calculus adopt asynchronous multi-party communication among nodes taking care of node proximity (the topology 
is static). The dynamics of the calculus is given in terms of a reduction 
relation. No behavioural equivalences are defined.

Finally, the paper at hand extends the conference paper \cite{LaMe17} in 
the following aspects:
\begin{inparaenum}[(i)]
\item the calculus has become a probabilistic calculus, both in its logical and its physical components; the logics has been enriched with probabilistic choice, 
whereas  discrete (finite-support) probability distributions have replaced  continuous non-deterministic uncertainties in the evolution  and continuous non-deterministic error-prone measurements; 
\item standard bisimulation has been replaced with probabilistic bisimulation and then with bisimulation metrics;
\item as a consequence, the case study has been revisited using our bisimulation metrics to estimate the deviation in terms of behaviour of the systems under investigation. 
\end{inparaenum}

\paragraph*{Future work.}
We believe that our paper can  lay and streamline 
\emph{theoretical foundations} for the development of formal and 
automated tools to verify \CPS{s} before their practical implementation.
 To that end, we will consider applying, possibly after proper enhancements,
 existing tools and frameworks for automated verification, such as 
Maude~\cite{Maude},  PRISM~\cite{PMC}, SMC UPPAAL~\cite{SMC-Uppaal}
and Ariadne~\cite{Ariadne}, resorting to the development of a dedicated tool if existing ones prove not up to the task.
 Finally, in~\cite{LMMV16}, we are currently working on a non-probabilistic version of \cname{} extended with security features to provide a formal study of a variety of  \emph{cyber-physical attacks}  targeting  physical devices. In this case, the final goal is to develop formal and automated tools to analyse security properties of \CPS{s}.

\enlargethispage{1.1\baselineskip}
As possible future work, a non-trivial challenge would be to extend the present work in order to deal with 
\emph{continuous probability distributions}.
In our setting,  this would  mean, for instance,  that the evolution map $\evolmap$ should return a continuous distribution over state functions, and that the function $\operatorname{next}_E(S)$ should return a continuous distributions over physical states. However, this would immediately give rise to a serious technical problem: the definition of probabilistic \emph{weak labelled transitions}, and hence the definition of \emph{weak behavioural equivalences and distances}. 
To  better illustrate the problem, suppose to adopt continuous probability distributions in our calculus, and suppose a cyber-physical system $M$ such that $M \trans{\tick}  \gamma$, for some continuous probability distribution $\gamma$ over \CPS{s}. 
Suppose $\gamma$ is a \emph{uniform distribution} such that $\support (\gamma) = \{M_r \colon r \in [0,1]\}$, with $M_r \neq M_{r'}$, for any $r \neq r'$.
Independently on the specific definition of the \CPS{s} $M_r$, as the logics of any \CPS{} is intrinsically discrete,   the cyber-component of any $M_r$  will drive the whole system to a \emph{discrete distribution}. As an example, 
assume a cyber-physical system $N$ such that 
for all reals  $r \in [0,0.5]$ there is a  $\tau$-transition $M_r \trans{\tau} \dirac{N}$;  whereas for all reals 
$r \in (0.5,1]$ there is a  $\tau$-transition $M_r \trans{\tau} \dirac{M_r}$. 
%%This is clearly a  discrete probability distribution over \CPS{s},  actually  a point distribution. 
In such a situation,  it is far from obvious to determine what should be the distribution $\gamma_m$ reached by the original \CPS{} $M$ after a  weak $\tick$-transition, $M \Trans{\tick} \gamma_m$.  
In fact, $\gamma_m$ can be neither a discrete nor a continuous distribution. 
This because $\gamma_m$ should map $N$ to a probability weight $0.5$ (as in a discrete distribution), and then it should distribute  the remaining mass probability as a uniform (sub-)distribution to all $M_r$ with $r\in (0.5,1]$, such that $\int_{0.5}^1  \gamma_m(M_t) dt = 0.5$ (as in a continuous distribution).

A possible solution to capture weak transitions when working with continuous probability distributions is to approximate them via discrete ones by adopting  the approach proposed for labelled Markov processes in \cite{Desh03,DGJP04}.
In these papers, 
%On the other hand
%\remarkR{Ho sostituito la \cite{DGJP04} con   \cite{Desh03} che secondo me e' meglio} 
%in \cite{Desh03} 
Desharnais et al.\ propose approximation techniques for continuous-state labelled Markov processes $\mathcal{S}$ in terms of finite-state Markov chains $\mathcal {S}(n,\epsilon)$, parametric in a natural number  $n$ and a rational number $\epsilon >0$.
Here, $n$ is the maximal number of possible consecutive transitions from the start state of $\mathcal {S}(n,\epsilon)$ (the idea being that this Markov chain is the $n$-steps unfolding of the original Markov process $\mathcal{S}$), whereas 
the rational number $\epsilon>0$ measures the accuracy of probabilities in $\mathcal {S}(n,\epsilon)$ when approximating the transitions
%%probabilities
 of the original process $\mathcal{S}$.
In their Theorem~4.4 \cite{Desh03} the authors prove that if a state $s$ of $\mathcal S$ satisfies a formula in the logic characterising probabilistic bisimulation then there is some approximation $\mathcal{S}(n, \epsilon)$
satisfying exactly the same formula. 
Furthermore, the same authors show that one can always reconstruct the original process from the approximations.  More precisely, a Markov process bisimilar to the original one can always be derived from the countable approximates $\mathcal{S}(n,  {2^{-n}})$, for some $n \in \mathbb{N}$ (in the current paper we adopted a  granularity $\epsilon=10^{-n}$). 
Actually, they do not 
reconstruct the original state space, but they reconstruct all the transition probability information, i.e., the dynamical aspects of the
process (see Theorem~4.5 of \cite{Desh03}).

\paragraph*{Acknowledgements.} We thank the anonymous reviewers for their insightful and careful reviews.

\bibliography{main}
\bibliographystyle{plain}

%%%%%
%%%%%  APPENDIX
%%%%%

\appendix

\section{Proofs}

%%%%%
%%%%%  SECTION A.1
%%%%%

\subsection{Proofs of  \autoref{sec:calculus}}
\label{app_uno}

%We recall that the cyber-components of our \CPS{s} are basically TPL-processes~\cite{HR95} extended with constructs to read values detected at sensors, to write on actuators, and to express guarded probabilistic choice. 
\autoref{prop:time} states that \CPS{s} enjoy time determinism, maximal progress, patience and well-timedness. 
We start with showing that processes enjoy the same properties.

\begin{lemma}[Processes time properties]
%%~\cite{HR95,MBS11}
\label{lem:time2}
Assume a process $P$.
\begin{itemize}

\item[(a)] 
\label{lem:time21}
If $P \trans \tick \pi$ and 
$P\trans \tick  \pi'$, then   
$\pi \equiv  \pi'$. 

\item[(b)] 
\label{lem:time22}
If $P\trans \tau \pi$  then there is no $\pi'$ such that  $P\trans \tick \pi'$. 

\item[(c)] 
\label{lem:time23}
If $P \trans \tick \pi'$  for no  $\pi'$ 
 then  
there is $\pi$ such that $P \trans \lambda \pi$ for some $\lambda \in \{\tau, {\snda a v} ,{\rcva s x}\}$. 

\item[(d)]
\label{lem:time24}
There is a $k\in\mathbb{N}$ such that if
$P  \trans{\lambda_1}\dots \trans{\lambda_n} P'$, with $\lambda_i \neq \tick$, then $n\leq k$. 

\end{itemize}
\end{lemma}

\begin{proof}
We show the four properties separately.
\begin{itemize}

\item[(a)]
The proof is by induction on the depth $d$ of the derivation tree allowing us to derive $P\trans \tick \pi$.

\underline{Base case $d=1$}. The transition $P\trans \tick \pi$ is derived by applying one of the rules \rulename{TimeNil},  
\rulename{Delay} and \rulename{Timeout}, and the thesis is immediate.

\underline{Inductive step $d>1$}. The transition $P\trans \tick \pi$ is derived by applying one of the rules \rulename{TimePar},  
\rulename{ChnRes} and \rulename{Rec}. We consider the case \rulename{TimePar}, the others are similar. 
Since $P\trans \tick \pi$ is derived by rule \rulename{TimePar}, process $P$ must be of the form $P \equiv P_1 \parallel P_2$ for suitable processes $P_1$ and $P_2$.
Therefore also the rule $P\trans \tick \pi'$ is derived through rule \rulename{TimePar}.
We have

\[
\Txiombis
{
  P_1 \trans{\tick}  {\pi_1}  \Q
  P_2 \trans{\tick} {\pi_2}  \Q P_1 \parallel P_2 \ntrans{\tau}
}
{
  {P_1 \parallel P_2}   \trans{\tick}  {\pi_1 \parallel \pi_2}
}
\Q\Q
\Txiombis
{
  P_1 \trans{\tick}  {\pi_1'}  \Q
  P_2 \trans{\tick} {\pi_2'}  \Q P_1 \parallel P_2 \ntrans{\tau}
}
{
  {P_1 \parallel P_2}   \trans{\tick}  {\pi_1' \parallel \pi_2'}
}
\]
with $\pi =\pi_1 \parallel \pi_2$ and $\pi' =\pi_1' \parallel \pi_2'$.

By the inductive hypothesis we have that $\pi_1 \equiv \pi_1'$  and  $\pi_2 \equiv \pi_2'$, which gives
$\pi_1 \parallel \pi_2 \equiv \pi_1'\parallel \pi_2'$ and concludes the proof. 

\item[(b)]
The proof is by induction on the depth $d$ of the derivation tree allowing us to derive $P\trans \tau \pi$.

\underline{Base case $d=1$}.
There is no rule in
\autoref{tab:lts_processes} allowing us to derive  transition $P\trans \tau \pi$ with depth $1$, hence the thesis follows trivially.

\underline{Inductive step $d>1$}. The transition $P\trans \tau \pi$ is derived by applying one of the rules 
\rulename{Com}, \rulename{Par}, \rulename{ChnRes} and \rulename{Rec}. 
We consider the case \rulename{Com}. 
Since $P\trans \tau \pi$ is derived by rule \rulename{Com}, process $P$ must be of the form $P \equiv P_1 \parallel P_2$ for suitable processes $P_1$ and $P_2$.
To show the thesis that no transition from $P_1 \parallel P_2$ labelled $\tick$ can be derived, it is enough to note that the only 
rule  in
\autoref{tab:lts_processes}.
which may be applied to infer any $\tick$-labelled transition from $P_1 \parallel P_2$ is rule \rulename{TimePar}, which cannot be applied since it has $P_1 \parallel P_2  \ntrans{\tau}$ among its premises.

The other cases follow directly by induction.

\item[(c)]
First  of al we notice that, if $P=\fix X Q$, then, since $P$ is bounded and has time-guarded recursion,
by applying repetitively  the structural congruence $\fix X Q \equiv   {Q{\subst {\fix{X}Q} X}}$, 
  we find a process  $P'  \equiv P$  such that
 $P' \neq \fix Y R$, for any $Y$ and $R$.  
Since $P'  \equiv P$ implies $P' \trans{\lambda}$ iff $P \trans{\lambda}$, for any $\lambda$,
we can prove the thesis  by structural induction on $P$ where $P$ is not of the form $P=\fix X Q$.

The base cases $P = \nil$, $P = \tick.C$ and $P =\timeout {\mathit{chn}.C} {D}$ are immediate since in all these cases a transition labelled $\tick$ from $P$ can be derived.
The base case $P=  \mathit{phy}.C$ holds since we can apply either rule \rulename{Write} to derive a transition from $P$ labelled 
${\snda a v}$, or rule \rulename{Read} to derive a transition labelled ${\rcva s x}$.

The inductive steps are $P = P_1 \parallel P_2$,  $P= \ifelse b{ P_1} {P_2} $ and  $P= Q\backslash c$.
Consider the case $P = P_1 \parallel P_2$.
If no transition from $P_1 \parallel P_2$  labelled $\tick$ can be derived, then rule \rulename{TimePar} cannot be applied.
Then, at least one of the premises $P_1 \trans{\tick} \pi_1$, $P_2 \trans{\tick} \pi_2$ and $P_1 \parallel P_2 \ntrans{\tau}$ does not hold.
If $P_1 \trans{\tick} \pi_1$ does not hold, then by the inductive hypothesis we have  $P_1 \trans \lambda \pi_1$ for some $\lambda \in \{\tau, {\snda a v} ,{\rcva s x}\}$, and
by rule \rulename{Par} we infer $P_1 \parallel P_2 \trans \lambda \pi_1 \parallel \dirac{P_2}$, which gives the thesis.
If $P_2 \trans{\tick} \pi_2$ does not hold,  then by the inductive hypothesis we have  $P_2 \trans \lambda \pi_2$ for some $\lambda \in \{\tau, {\snda a v} ,{\rcva s x}\}$, and
by the rule symmetric to \rulename{Par} we infer $P_1 \parallel P_2 \trans \lambda \dirac{P_1} \parallel \pi_2$, which gives the thesis.
If $P_1 \parallel P_2 \ntrans{\tau}$ does not hold then there is some transition  $P_1 \parallel P_2 \trans{\tau} \pi$, which gives the thesis.
The cases $P= \ifelse b{ P_1} {P_2} $ and $P= Q\backslash c$ are similar.

\item[(d)]
The well-timedness property  is 
straightforward from time--guardedness recursion. 

\end{itemize}
\end{proof}

The challenge in the proof of  \autoref{prop:time} is to lift the results of \autoref{lem:time2} to the \CPS{s} of \cname. \\

\noindent
\textbf{Proof of \autoref{prop:time}}
\begin{enumerate}

\item[(a)]
We note that transitions labelled $\tick$ can be derived only by rule \rulename{Time}.
Therefore, from the hypothesis $M \trans{\tick} \gamma$ and $M \trans{\tick} \gamma'$ with $M = E; \confCPS{S}{P}$, we infer that 
there are process distributions $\pi$ and $\pi'$ such that 
 \[\Txiombis
  {P  \trans{\tick} {\pi} \Q\Q
 \confCPS{S}{P} \ntrans{\tau} \Q\Q
S \in \invariantfun{}}
{\confCPS{S}{P} \trans{\tick} \confCPS {\operatorname{next}_E(S)} {\pi}}
\Q\q \textrm{and} \Q\q \Txiombis
  { P  \trans{\tick} {\pi'} \Q\Q
 \confCPS{S}{P} \ntrans{\tau} \Q\Q
S \in \invariantfun{}}
{\confCPS{S}{P} \trans{\tick} \confCPS {\operatorname{next}_E(S)} {\pi'}}
\]
where $\gamma = E ; \confCPS {\operatorname{next}_E(S)} {\pi}$ and $\gamma' =  E ; \confCPS {\operatorname{next}_E(S)} {\pi'}$.
By the property of time determinism for processes in \autoref{lem:time2} we infer that $P  \trans{\tick} {\pi}$ and $P  \trans{\tick} {\pi'}$ imply $\pi \equiv \pi'$, hence $\gamma \equiv \gamma'$, which completes the proof.

\item[(b)]
From the hypothesis $M \trans{\tau} \gamma$ with $M = E; \confCPS{S}{P}$, we infer that $\gamma = E ; \confCPS{\sigma}{\pi}$ for distributions $\sigma$ and $\pi$ such that
$\confCPS{S}{P} \trans{\tau} \confCPS{\sigma}{\pi}$ is derived from the rules in
\autoref{tab:lts_systems}.
To show the thesis that no transition from $M$ labelled $\tick$ can be derived, it is enough to show that no transition from $\confCPS{S}{P}$ labelled $\tick$ can be derived from the rules in
\autoref{tab:lts_systems}.
This follows by the fact that the only rule which may be applied to infer any $\tick$-labelled transition from $\confCPS{S}{P}$ is rule \rulename{Time}, which cannot be applied since it has $\confCPS{S}{P}  \ntrans{\tau}$ among its premises.

\item[(c)]
From the hypothesis that $M \trans{\tick} \gamma$ with $M = E; \confCPS{S}{P}$ cannot be inferred for any distribution $\gamma$, we infer that 
$\confCPS{S}{P} \trans{\tick} \confCPS{\sigma}{\pi}$ cannot be derived for any $\sigma$ and $\pi$ from the rules in \autoref{tab:lts_systems}.
Therefore, at least one of the premises $P  \trans{\tick} {\pi}$, $\confCPS{S}{P}   \ntrans{\tau}$ and $S \in \invariantfun{}$ of rule  \rulename{Time} does not hold.
If premise $P  \transS[\tick] {\pi}$ does not hold for any  $\pi$, then by the property of patience for processes in \autoref{lem:time22} we have $P \trans \lambda \pi'$ for some  $\pi'$ and $\lambda \in \{\tau, {\snda a v} ,{\rcva s x} \}$.
Let us consider the case $\lambda=\tau$. 
From $P \trans \tau \pi'$, either $S \in \invariantfun{}$ is not valid, or we can apply rule \rulename{Tau} to infer the transition $\confCPS{S}{P} \transS[\tau] \confCPS{\dirac{S}}{\pi'}$, which gives $M \transS[\tau] E; \confCPS{\dirac{S}}{\pi'}$.
In both cases the thesis holds.
The cases $\lambda \in \{ {\snda a v} ,{\rcva s x} \}$ can be proved similarly by using rules \rulename{ActWrite} and \rulename{SensRead}, respectively.
If premise $P \transS[\tick] \pi$ holds for some $\pi$ then either premises  $S \in \invariantfun{}$ or premise $\confCPS{S}{P} \ntrans{\tau}$ does not hold. In the former case the thesis follows.
In the latter case we have a $\tau$-labelled transition from $M$ and the thesis holds as well.

\item[(d)]
The proof is by contradiction.  
Suppose there is no $k$ satisfying the statement of the thesis.
Hence there exists an unbounded derivation 
\[
E;\confCPS{S}{P} = E;\confCPS{S_1}{P_1} \trans{\alpha_1} \dots \trans{\alpha_n}  E;\confCPS{S_n}{P_n}  \trans{\alpha_{n+1}} \dots
\]
with $\alpha_i \neq \tick$ for $i \ge 1$,
namely there exist distributions $\confCPS{\sigma_i}{\pi_i}$ for $i \ge 1$ with $\confCPS{\sigma_i}{P_i} \transS[\alpha_i] \confCPS{\sigma_{i+1}}{\pi_{i+1}}$, $S_{i+1} \in \support(\sigma_{i+1})$ and $P_{i+1} \in \support(\pi_{i+1})$.
This contradicts the property of well-timedness for processes in \autoref{lem:time2}.
\qed

\end{enumerate}

\begin{comment}

%%%%%
%%%%%  SECTION A.2 commentata il 19 gennaio 2018
%%%%%

\subsection{Proofs of \autoref{sec:bisimulation}}
\label{Sec-A.2}

\noindent
\textbf{Proof of \autoref{thm:congruence}} \hspace{0.2 cm}
%This is a special case of \autoref{thm:congruenceP}.
%In detail, c
Consider  \autoref{thm:congruence}.1.
We have that 
\[
M \approx N 
 \; \Longrightarrow  \; \; 
\metric(M,N) = 0 
 \; \; \Longrightarrow  \;\;
\metric(M \uplus O,N\uplus O) = 0
 \;\; \Longrightarrow  \;\;
M \uplus O \approx N \uplus O
\]
by applying, respectively,  \autoref{prop:kernel}, \autoref{thm:congruenceP}.1, and \autoref{prop:kernel} again.
\autoref{thm:congruence}.2 and \autoref{thm:congruence}.3 are analogous.
\qed

\end{comment}

%%%%%
%%%%%  SECTION A.3
%%%%%

\subsection{Proofs of  \autoref{sec:case-study}} 
\label{app:sec:case-study}

In order to prove \autoref{prop:sys} and \autoref{prop:X} we use the following lemma that formalises the invariant properties binding the state variable $\mathit{temp}$ with the activity of the cooling system.
Intuitively,  when the cooling system is inactive then the value of the state variable $\mathit{temp}$ lays in the interval $[0, 11+\mathit{err}+\delta]$. 
Furthermore, if the coolant is not active and the variable $\mathit{temp}$ lays in the interval $(10+\mathit{err}, 11+\mathit{err}+\delta]$ then the cooling will be turned on in the next time slot. 
Finally, if the cooling system is active then there is some $k=1 \dots 5$ such that the system was activated $k$ time units ago, it was kept active so far and the 
state variable $\mathit{temp}$ lays in the real interval  
 $(10-\mathit{err}-k{*}(1{+}\delta) , 11+\mathit{err}+\delta-k{*}(1{-}\delta)]$.

\begin{lemma} 
\label{lem:sys} 
Let $ \mathit{Eng}_g $ be the system defined in \autoref{sec:case-study}.
Let
\begin{displaymath}
 \mathit{Eng_g}  = {M_1}  \trans{t_1}\trans\tick 
 {M_2} \trans{t_2}\trans\tick  \dots 
\trans{t_{n-1}}\trans\tick  {M_n} 
\end{displaymath}
such that the traces $t_j$ contain no $\tick$-actions, for any $j \in  1 \ldots n{-}1$,  and for any $i \in  1 \ldots n $ we have ${M_i}= \env_g ; \confCPS {S_i}{P_i} $ with 
$\state_i = \stateCPS 
{\statefun^i{}} 
{\sensorfun^i{}}
{\actuatorfun^i{}} 
$ 
and 
$\env_g = \envCPS 
{\evolmap{}}
{\measmap{}}   
{\invariantfun{}}$. 
Then, for any $i \in 1 \ldots n{-}1 $ we have the following:
\begin{enumerate}

\item 
\label{one}
 if   $ \actuatorfun^i{}(\mathit{cool})= \off $ then
 $\statefun^i{}(\mathit{temp})  \in [0, 11+\mathit{err}+\delta ]$; 

\item 
\label{two}
  if   $ \actuatorfun^i{}(\mathit{cool})= \off $ and 
$\statefun^i{}(\mathit{temp})\in (10+\mathit{err}, 11+\mathit{err}+\delta]$ then, in the next time slot,  $\actuatorfun^{i{+}1}{}(\mathit{cool})=\on$; 

\item 
\label{three}
if  $ \actuatorfun^i{}(\mathit{cool})=\on$ then   $\statefun^i{}(\mathit{temp}) \in ( 10-\mathit{err}-k {*}(1{+}\delta) , 11+\mathit{err}+\delta -k{*}(1{-}\delta)] $, 
for some  $k  \in 1 \dots 5$  such that $\actuatorfun^{i-k}{}(\mathit{cool})=\off $ and 
$\actuatorfun^{i-j}{}(\mathit{cool}) =\on $, for all $j \in 0 \ldots k{-}1$. 

\end{enumerate}
\end{lemma}

\begin{proof}
Let  us denote with  $v_i$  the values  of
the state variable $\mathit{temp}$ in the systems 
 ${M_i}$, i.e., $\statefun^i{}  (\mathit{temp})=v_i $.
Moreover  we will say that  the coolant is active (resp., is not active) in  ${M_i}$ if $\actuatorfun^i{}(\mathit{cool})=\on$
(resp., $\actuatorfun^i{}(\mathit{cool})=\off$).

The proof is  by mathematical induction on $n$, i.e., the 
number of $\tick$-actions of our traces. 

The case base $n=1$ follows directly from the definition of $\mathit{Eng}_g$. 
Let prove the inductive case. 
We assume that the three statements holds for $n-1$ and we prove that they  
also hold for $n$.
\begin{enumerate}[noitemsep]

\item 
Let us assume that the cooling  is not active  in ${M_{n}}$, 
then we prove that $v_n \in [0, 11+\mathit{err} +\delta ]$. 
We consider separately the cases in which  the coolant is active or not in ${M_{n-1}}$.

\begin{itemize}[noitemsep]

\item 
Suppose the coolant is not active  in ${M_{n{-}1}}$ (and inactive in  ${M_{n}}$).

By the inductive hypothesis we have $v_{n-1} \in [0, 11+\mathit{err}  +\delta ]$. 
Since we know that in ${M_n}$ the cooling is not active, it follows that 
$v_{n-1} \in [0, 10+\mathit{err} ]$, the reason being that 
$v_{n-1} \in (10+\mathit{err}  , 11+\epsilon +\delta ]$ and the inductive hypothesis would imply that the coolant is active in ${M_{n}}$.
Furthermore, in ${M_{n}}$  the temperature
will increase of a value laying in the interval $[1-\delta,1+\delta]_g=[0.6,1.4]_g$. Thus $v_{n}$ will be  in 
$ [0.6, 11+\mathit{err}  +\delta ]\subseteq[0, 11+\mathit{err}  +\delta ]$.

\item 
Suppose the coolant is active  in ${M_{n{-}1}}$ (and  inactive in  ${M_{n}}$).

By the inductive hypothesis we have  $v_{n-1} \in (
10-\mathit{err}  -k *(1+\delta) , 11+\mathit{err}  +\delta -k*(1-\delta)] $ for some $k \in 1\dots 5$ such that the coolant is
not active in ${M_{n{-}1{-}k}}$ and is active in all ${M_{n{-}k}},
\ldots, {M_{n-1}}$.

%\remarkS{Questa frase in blu risponde al ref.}
%\remarkR{ok}
The case $k \in \{1,\ldots,4\}$  is not admissible, the reason being that $k \in \{1,\ldots,4\}$ together with the fact that the coolant is inactive in  ${M_{n}}$ would imply that the coolant bas been kept active for less than 5 steps, which cannot happen.

Hence it must be $k=5$. Since $\delta=0.4$, $\mathit{err} =0.1$  and $k=5$, it holds that $v_{n-1 }\in (10-0.1 -5*1.4, 11+ 0.1 +0.4 -5*0.6]=(2.8, 8.6] $. Moreover, since
the coolant is active for $5$ $\tick$ actions, the controller of 
${M_{n{-}1}}$ checks the
temperature.  
However, since $v_{n-1} \in (2.8, 8.6]$ then the coolant is turned off. 
Thus, in the next time slot,  the temperature
will increase of a value in $[1-\delta,1+\delta]_g=[0.6,1.4]_g$. As
a consequence in ${M_{n}}$  we will have $v_{n}
\in [2.8+0,6, 8.6+1.4]=[3.4,10] \subseteq [0, 11+\mathit{err}  +\delta ]$.

\end{itemize}

\item 
Let us assume that the coolant  is not active  in ${M_{n}}$ and $v_n \in (10+\mathit{err} , 11+\mathit{err}  +\delta ]$, then  we prove that  the coolant is active  in ${M_{n+1}}$.
Since the coolant is not active in
${M_{n}}$ then it will check the
temperature before the next time slot. 
Since $v_n \in (10+\mathit{err}  , 11+\mathit{err}  +\delta ]$ and $\mathit{err} =0.1$, then the
process $\mathit{Ctrl}$ will sense a temperature greater than $10$ and 
the coolant will be turned on. Thus the coolant will be active in
${M_{n{+}1}}$. 

\item 
Let us assume that the coolant is active in
${M_{n}}$, then  we prove that  $v_{n} \in ( 10-\mathit{err}  -k *(1+\delta), 11+\mathit{err}  +\delta -k*(1-\delta)] $ for some $k \in 1\ldots 5$
 and  the coolant is not active in ${M_{n{-}k}}$ and  active
in all ${M_{n-k+1}}, \ldots, {M_{n}}$.

We separate the case in which the coolant is active in ${M_{n{-}1}}$ from that in which is not active. 

\begin{itemize}[noitemsep]

\item 
Suppose the coolant is not active  in ${M_{n{-}1}}$ (and active in ${M_{n}}$).

In this case $k=1$ as the coolant is not
active in ${M_{n-1}}$ and it is active in ${M_{n}}$. 
Since  $k=1$, we have to prove $v_n \in (10-\mathit{err}  -(1+\delta), 11+\mathit{err}  +\delta-(1-\delta)]$.

However, since the coolant is not
active in ${M_{n-1}}$ and is active in ${M_{n}}$ it means that the coolant has been switched on in ${M_{n-1}}$ because the sensed temperature  was above $10$ (this may happen
 only if $v_{n-1} > 10-\mathit{err} $).
By inductive hypothesis, since  the coolant is not active  in ${M_{n-1}}$, we have that
$v_{n-1} \in [0, 11+\mathit{err}  +\delta ]$.
Therefore, from $v_{n-1} > 10-\mathit{err}$  and 
$v_{n-1} \in [0, 11+\mathit{err}  +\delta ]$ it follows that  $v_{n-1} \in (10-\mathit{err}  , 11+\mathit{err} +\delta ]$. 
Furthermore, 
since the coolant is active in ${M_{n}}$, the temperature will
decrease of a value in $[1-\delta,1+\delta]_g$ and therefore
$v_n \in (10-\mathit{err}  -(1+\delta), 11+\mathit{err}  +\delta-(1-\delta)]$  which concludes this case of the proof.

\item 
Suppose the coolant is active in ${M_{n{-}1}}$ (and active in ${M_{n}}$ as well).

By inductive hypothesis there is $h \in 1\ldots5$ such that  $v_{n-1} \in (
10-\mathit{err}  -h *(1+\delta) , 11+\mathit{err}  +\delta -h*(1-\delta)] $ and  the coolant is
not active in ${M_{n{-}1{-}h}}$ and is active in ${M_{n{-}h}},
\ldots, {M_{n{-}1}}$.

The case   $h=5$ is not admissible. In fact, since $\delta=0.4$ and $\mathit{err} =0.1$,
 if $h=5$ then  
 $v_{n-1 }\in (10-0.1 -5*1.4, 11+ 0.1 +\delta -5*0.6]=(2.8, 8.6] $. Furthermore, since
the coolant is already active since $5$ $\tick$ actions, the controller of 
 ${M_{n{-}1}}$ is supposed to check the
temperature. As  $v_{n-1 }\in (2.8, 8.6] $ the coolant 
should be turned off. 
In contradiction  with the the fact that  the coolant is   active  in ${M_{n }}$.

Hence it must be 
$h \in 1 \ldots 4$. Let us prove that for $k=h+1$ we obtain  our result. 
Namely we have to prove  that, for $k=h+1$,  (i)   $v_{n} \in ( 10-\mathit{err}  -k *(1+\delta), 11+\mathit{err}  +\delta -k*(1-\delta)] $,   and (ii)
 the coolant is not active in ${M_{n{-}k}}$ and  active
in all ${M_{n-k+1}}, \dots, {M_{n}}$.

Let us prove the  statement (i). By inductive hypotheses, it holds that  
$v_{n-1} \in ( 10-\mathit{err}  -h *(1+\delta) , 11+\mathit{err}  +\delta -h*(1-\delta)] $.
Since the coolant is active in  ${M_{n}}$ then the temperature will decrease. 
Hence,  
$v_{n }   \in ( 10-\mathit{err}  -(h+1) *(1+\delta) , 11+\mathit{err}  +\delta -(h+1)*(1-\delta)]   $. 
Therefore, since  $k=h+1$, we have that 
$v_{n} \in ( 10-\mathit{err}  -k *(1+\delta) , 11+\mathit{err}  +\delta -k*(1-\delta)] $.

Let us prove the statement (ii).
By inductive hypothesis  the coolant is
inactive in ${M_{n-1-h}}$ and  it is active in all ${M_{n-h}},
\ldots, {M_{n-1}}$. 
Now, since the coolant is active in  ${M_{n}}$, for  $k=h+1$, we have that  the coolant is 
not active in ${M_{n-k}}$ and is active in all ${M_{n-k+1}},
\ldots, {M_{n}}$ which concludes this case of the proof.

\end{itemize}

\end{enumerate}

\end{proof}

\noindent 
\textbf{Proof of \autoref{prop:sys}} \hspace{0.2 cm}
By the first two items of \autoref{lem:sys} and since $\delta=0.4$ and $\mathit{err} =0.1$, we infer that the value of the state variable $\mathit{temp}$ is always in the real interval $[0, 11.5]$. 
As a consequence, the invariant of the system is never violated and the system never deadlocks.
Then, the last item of \autoref{lem:sys} ensures that after $5$ $\tick$-actions happening when the coolant is active, the state variable $\mathit{temp}$ is always in the real interval  $( 10-0.1-5 *1.4 , 11+0.1+0.4-5*0.6]=(2.9, 8.5]$. 
Hence the process $\mathit{Ctrl}$ will never transmit on the channel $\mathit{warning}$.
\qed
\\

\noindent 
\textbf{Proof of \autoref{prop:X}}  \hspace{0.2 cm}
Let us prove the two statements separately. 
\begin{itemize}
\item   
If process $\mathit{Ctrl}$ senses a
temperature above $10$  (and hence $\mathit{Eng}$ turns on the cooling) 
then the value of the state variable
$\mathit{temp}$ is greater than $10-\mathit{err} $. 
By  \autoref{lem:sys}
 the value of the state variable $\mathit{temp}$ is always less or equal than
$11+\mathit{err}  +\delta $. 
Therefore, if $\mathit{Ctrl}$ senses a temperature above $10$,
then the value of the state variable $\mathit{temp}$ is in $(10-\mathit{err}  ,11+\mathit{err}  +\delta ] = (9.9,11.5]$.

\item
By   \autoref{lem:sys}  (third item) the coolant can
 be active for no more than $5$ time slots.
Hence,  by  \autoref{lem:sys}, when 
$\mathit{Eng}$ turns off the cooling system 
the state variable $\mathit{temp}$ ranges over  $( 10-\mathit{err}  -5 *(1+\delta) , 11+\mathit{err}  +\delta-5*(1-\delta)] = (2.9,8.5]$.
\end{itemize}
\qed
\\

\noindent 
\textbf{Proof of \autoref{prop:stop}}  \hspace{0.2 cm}
It is is enough to prove that there exists an execution trace of the engine $\widehat{\mathit{Eng}}_g$  containing an  output along channel $\mathit{warning}$.
Then the  result follows by an application of \autoref{prop:sys}. 

We prove the thesis for $g=1$. Indeed a trace of $\widehat{\mathit{Eng}}_g$ with $g=1$ is  a trace of 
$\widehat{\mathit{Eng}}_{g'}$ with $g' \ge g$.

We can easily build up a trace for $\widehat{\mathit{Eng}_g}$ with $g=1$ in which, after $10$ $\tick$-actions, in the $11$-th time slot,  the value of the state variable  $\mathit{temp}$  is  $10.1$. 
In fact, it is enough to increase the temperature of $1 $ degree
for the first $9$ rounds and an increase   of $1.1 $ degrees
in the  $10$-th time slot. Notice that these are admissible values, since both $1$ and $1.1$ are in $  [  1-\delta,1+\delta ]_g= [  0.6,1.4]_g$ with $g=1$.
Being $10.1 $ the value  of the state variable $\mathit{temp}$, there is an execution trace  
in which the  sensed temperature is $10$ (recall that $\mathit{err}=0.1$ and $-0.1\in [-0.1,0.1]_g$ with $g=1$) and hence 
the cooling system is not activated. 
However, 
in the following time slot, i.e.\ the $12$-th time slot,
 the temperature may reach  the value
$10.1  + 1+\delta=11.5$, imposing the activation of the cooling system. 
After $5$ time units of cooling, in the $17$-th time slot, 
the  variable $\mathit{temp}$ could be    
$11.5 -5 \ast (0.7-\delta)=11.5-1.5=10$.
The sensed temperature would be  in 
the real interval $[9.9 ,10.1 ]_g $ with $g=1$. Thus, there 
is an execution trace in which the sensed temperature is $ 10.1 $. 
As a consequence,  the warning will be emitted, in the $17$-th time slot.
\qed
\\

%\noindent 
%\textbf{Proof of \autoref{prop:air}}  \hspace{0.2 cm}
%The thesis follows directly by \autoref{prop:air2} and \autoref{prop:kernel}
%\qed
\begin{comment}
\begin{proof}
By \autoref{prop:performances} we derive 
$\mathit{Eng} \approx \widetilde{\mathit{Eng}}$.
By simple 
$\alpha$-conversion it follows that $\mathit{Eng}_{\mathrm L} \approx \widetilde{\mathit{Eng}_{\mathrm L}}$ and 
$\mathit{Eng}_{\mathrm R} \approx \widetilde{\mathit{Eng}_{\mathrm R}}$, respectively. By \autoref{thm:congruence}(1) 
 (and transitivity of $\approx$) it 
follows that $\mathit{Eng}_{\mathrm L} \uplus \mathit{Eng}_{\mathrm R}
\approx \widetilde{\mathit{Eng}_{\mathrm L}} \uplus \widetilde{\mathit{Eng}_{\mathrm R}}$. By \autoref{thm:congruence}(2) it follows that 
 $(\mathit{Eng}_{\mathrm L} \uplus \mathit{Eng}_{\mathrm R}) \parallel \mathit{Check}
\approx (\widetilde{\mathit{Eng}_{\mathrm L}} \uplus \widetilde{\mathit{Eng}_{\mathrm R}}) \parallel \mathit{Check}$. By \autoref{thm:congruence}(3) we
obtain $\mathit{Airplane} \approx \widetilde{\mathit{Airplane}}$.
\end{proof}
\end{comment}

%%%%
%%%%  SECTION A.5
%%%%

\subsection{Proofs of \autoref{sec:metric}}  
\label{Sec:A.5}
To prove that all $\metric^n$ are 1-bounded pseudometrics (\autoref{prop:up-to-k-metric}), we need some preliminary results.
First we show that the Kantorovich functional $\Kantorovich$ maps pseudometrics to pseudometrics.

\begin{proposition}
\label{prop_kant_metric}
If $d \colon \nome \times \nome \to [0,1]$ is a 1-bounded pseudometric, then also $\Kantorovich(d) \colon \distr{\nome} \times \distr{\nome}$ is a 1-bounded pseudometric.
\end{proposition}

\begin{proof}
To show $\Kantorovich(d)(\gamma,\gamma) = 0$ for all $\gamma \in \distr{\nome}$ it is enough to take the matching $\omega \in \Omega(\gamma,\gamma)$ defined by $\omega(M,M) = \gamma(M)$, for all $M \in \nome$, and $\omega(M,N) = 0$, for all $M,N \in \nome$ with $M \neq N$. 
In fact, we have $\Kantorovich(d)(\gamma,\gamma) \le
\sum_{M,N \in \nome}\omega(M,N) \cdot d(M,N) = \sum_{M \in \nome} \gamma(M) \cdot d(M,M) =0$.

The symmetry $\Kantorovich(d)(\gamma,\gamma') = \Kantorovich(d)(\gamma',\gamma)$ for all $\gamma,\gamma' \in \distr{\nome}$ follows directly by the fact that if we take two functions $\omega,\omega' \colon \nome \times \nome \to [0,1]$ such that $\omega(M,N) = \omega'(N,M)$ for all $M,N \in \nome$, then $\omega \in \Omega(\gamma,\gamma')$ if and only if $\omega' \in \Omega(\gamma',\gamma)$.

To prove the triangle inequality $\Kantorovich(d)(\gamma_1,\gamma_2) \le \Kantorovich(d)(\gamma_1,\gamma_3) + \Kantorovich(d)(\gamma_3,\gamma_2)$ for all $\gamma_1,\gamma_2,\gamma_3 \in \distr{\nome}$, first we consider
the function $\omega \colon \nome \times \nome \to [0,1]$ defined for all $M_1,M_2 \in \nome$ as $\omega(M_1,M_2) = \sum_{M_3 \in \nome \mid \gamma_3(M_3) \neq 0} \frac{\omega_1(M_1,M_3) \cdot \omega_2(M_3,M_2)}{\gamma_3(M_3)}$, where the function $\omega_1 \in \Omega(\gamma_1,\gamma_3)$ is one of the optimal matchings realising $\Kantorovich(d)(\gamma_1,\gamma_3)$ and $\omega_2 \in \Omega(\gamma_3,\gamma_2)$ one of the optimal matchings realising $\Kantorovich(d)(\gamma_3,\gamma_2)$.
Then, we prove that
\begin{inparaenum}[(i)]
\item \label{Kant_triang_uno}
$\omega$ is a matching in $\Omega(\gamma_1,\gamma_2)$, and
\item \label{Kant_triang_due}
$\sum_{M_1,M_2 \in \nome} \omega(M_1,M_2) \cdot d(M_1,M_2) \le \Kantorovich(d)(\gamma_1,\gamma_3)  + \Kantorovich(d)(\gamma_3,\gamma_2)$, which immediately implies $\Kantorovich(d)(\gamma_1,\gamma_2) \le \Kantorovich(d)(\gamma_1,\gamma_3)  + \Kantorovich(d)(\gamma_3,\gamma_2)$.
\end{inparaenum}
To show (\ref{Kant_triang_uno}) we prove that the left marginal of $\omega$ is $\gamma_1$ by
\[
\begin{array}{rlr}
& 
\sum_{M_2 \in \nome} \omega(M_1,M_2)
\\
= \quad & \sum_{M_2 \in \nome}  \sum_{M_3 \in \nome \mid \gamma_3(M_3) \neq 0} \frac{\omega_1(M_1,M_3) \cdot \omega_2(M_3,M_2)}{\gamma_3(M_3)}
\\[1.7 ex]
= \quad & \sum_{M_3 \in \nome \mid \gamma_3(M_3) \neq 0} \frac{\omega_1(M_1,M_3) \cdot \gamma_3(M_3)}{\gamma_3(M_3)} & \text{(by $\omega_2 \in \Omega(\gamma_3,\gamma_2)$)}
\\[1.7 ex]
= \quad & \sum_{M_3 \in \nome \mid \gamma_3(M_3) \neq 0} \omega_1(M_1,M_3) 
\\[1.7 ex]
= \quad & \gamma_1(M_1) & \text{(by $\omega_1 \in \Omega(\gamma_1,\gamma_3)$)}
\end{array}
\]
and we observe that the proof that the right marginal of $\omega$ is $\gamma_2$ is analogous.
Then, we show (\ref{Kant_triang_due}) by
\[
\begin{array}{rlr}
&  \sum_{M_1,M_2 \in \nome} \omega(M_1,M_2) \cdot d(M_1,M_2)
\\ 
=  \quad &\sum_{M_1,M_2 \in \nome} \sum_{M_3 \in \nome \mid \gamma_3(M_3) \neq 0} \frac{\omega_1(M_1,M_3) \cdot \omega_2(M_3,M_2)}{\gamma_3(M_3)} \cdot d(M_1,M_2)
\\[1.7 ex]
\le   \quad &  \sum_{M_1,M_2 \in \nome,M_3 \in \nome \mid \gamma_3(M_3) \neq 0} \frac{\omega_1(M_1,M_3) \cdot \omega_2(M_3,M_2)}{\gamma_3(M_3)} \cdot d(M_1,M_3)  \; + 
\\[1.7 ex]
 \quad  &  \sum_{M_1,M_2 \in \nome, M_3 \in \nome \mid \gamma_3(M_3) \neq 0} \frac{\omega_1(M_1,M_3) \cdot \omega_2(M_3,M_2)}{\gamma_3(M_3)} \cdot d(M_3,M_2) 
\\[1.7 ex]
=  \quad &  \sum_{M_1,M_3 \in \nome} \frac{\omega_1(M_1,M_3) \cdot \gamma_3(M_3)}{\gamma_3(M_3)} \cdot  d(M_1,M_3)  +
   \sum_{M_2,M_3 \in \nome} \frac{\gamma_3(M_3) \cdot \omega_2(M_3,M_2)}{\gamma_3(M_3)}  \cdot d(M_3,M_2)
\\[1.7 ex]
=  \quad &  \sum_{M_1,M_3 \in \nome} \omega_1(M_1,M_3) \cdot d(M_1,M_3)  +
   \sum_{M_2,M_3 \in \nome} \omega_2(M_3,M_2)  \cdot d(M_3,M_2)
\\[1.7 ex]
=  \quad &  \Kantorovich(d)(\gamma_1,\gamma_3)  + \Kantorovich(d)(\gamma_3,\gamma_2) 
\end{array}
\]
where the inequality follows from the triangular property of $d$ and the third last equality follows by $\omega_2 \in \Omega(\gamma_3,\gamma_2)$ and $\omega_1 \in \Omega(\gamma_1,\gamma_2)$.
\end{proof}

Now we show that, given any weak bisimulation metric $d$ with $d(M,N) < 1$, then $N$ can mimic weak transitions $M \TransS[\hat{\alpha}]$ besides those of the form $M \transS[\alpha]$.

\begin{lemma} 
\label{lemma_sim_weak_transitions}
Assume a weak bisimulation metric $d$ and $M,N \in \nome$ with $d(M,N) <1$.
If $M \TransS[\hat \alpha] \gamma_M$, then there is a transition $N \TransS[\hat \alpha] \gamma_N$ such that $ \Kantorovich(d)(\gamma_M + (1-\size{\gamma_M}) \dirac{\dummyN}, \gamma_N + (1-\size{\gamma_N}) \dirac{\dummyN}) \le d(M,N)$.
\end{lemma}

\begin{proof}
We proceed by induction on the length $n$ of $M \TransS[\hat \alpha] \gamma_M$.

\underline{Base case $n=1$}. We have two sub-cases: 
The first is $\alpha = \tau$ and $\gamma_M = \dirac{M}$, the second is $M \transS[\alpha] \gamma_M$.  
In the first case, by definition of $\TransS[\widehat \tau]$ we have $N \TransS[\widehat \tau] \dirac{N}$ and
the thesis holds for $\gamma_N = \dirac{N}$ by observing that $\Kantorovich(d)(\dirac{M} + (1-\size{\dirac{M}})\dirac{\dummyN}),\dirac{N} + (1-\size{\dirac{N}})\dirac{\dummyN}) = \Kantorovich(d)(\dirac{M},\dirac{N}) = d(M,N)$.
In the second case, the thesis follows directly by the definition of weak simulation metric.

\underline{Inductive step $n>1$}.
%\remarkM{Rev3: Leggere l'ultima sua nota. }
%\remarkS{Ho seguito le sue indicazioni.}
The derivation $M \TransS[\hat \alpha] \gamma_M$ is obtained by $M \TransS[\hat \beta_1] \rho_M$ and $\rho_M \transS[\hat{\beta}_2] \gamma_M$, for some distribution $\rho_M \in\distr{\nome}$.
The length of the derivation $M \TransS[\hat \beta_1] \rho_M$ is $n-1$ and hence, by the inductive hypothesis, 
there is a transition $N \TransS[\hat \beta_1] \rho_N$ such that $\Kantorovich(d)(\rho_M + (1-\size{\rho_M}) \dirac{\dummyN},\rho_N+ (1-\size{\rho_N}) \dirac{\dummyN}) \le d(M,N)$.
The sub-distributions $\rho_M$ and $\rho_N$ are of the form $\rho_M = \sum_{i \in I}p_i \cdot \dirac{M_i}$ and $\rho_N = \sum_{j \in J}q_j \cdot \dirac{N_j}$.
We have two sub-cases: The first is $\beta_1=\tau$ and $\beta_2=\alpha$, the other $\beta_1=\alpha$ and $\beta_2=\tau$.

We consider the case $\beta_1=\tau$ and $\beta_2=\alpha$, the other is analogous.
In this case we have $\size{\rho_M} = \size{\rho_N} = 1$ and $\Kantorovich(d)(\rho_M ,\rho_N) \le d(M,N)$.
The transition $\rho_M \transS[\hat{\beta}_2] \gamma_M$ is derived from a $\beta_2$-transition by some of the \CPS{s} $M_i$, namely $I$ is partitioned into sets $I_1 \cup I_2$ such that for all $i \in I_1$ we have $M_i \transS[\beta_2] \gamma_i$ for suitable distributions $\gamma_i$, for each $i \in I_2$ we have $M_i \ntransS[\beta_2]$, and $\rho_M = \sum_{i \in I_1} p_i \cdot \gamma_i$.
Analogously, $J$ is partitioned into sets $J_1 \cup J_2$ such that for all $j \in J_1$ we have $N_j \TransS[\hat{\beta}_2] \gamma_j$ for suitable distributions $\gamma_j$ and for each $j \in J_2$ we have $N_j \nTransS[\hat{\beta}_2]$. 
This gives $\rho_N \TransS[\hat{\beta}_2] \gamma_N$ with $\gamma_N = \sum_{j \in J_1} q_j \cdot \gamma_j$.
Since we had $N \TransS[\hat{\beta}_1] \rho_N$, we can conclude $N \TransS[\hat{\alpha}] \gamma_N$.
In the following we prove that the transitions $N_j \TransS[\hat{\beta}_2] \gamma_j$ can be chosen so that $\Kantorovich(d)(\gamma_M + (1-\size{\gamma_M}) \dummyN,\gamma_N + (1-\size{\gamma_N}) \dummyN) \le d(M,N)$, which concludes the proof.

Let $\omega$ be one of the optimal matchings realising $\Kantorovich(d)(\rho_M ,\rho_N)$.
We can rewrite the distributions $\rho_M$ and $\rho_N$ as 
$\rho_M = \sum_{i \in I, j \in J} \omega(M_i,N_j) \cdot  \dirac{M_i}$ and 
$\rho_N = \sum_{i \in I, j \in J} \omega(M_i,N_j)  \cdot \dirac{N_j}$.
For all $i \in I_1$ and $j \in J$, define $\gamma_{i,j} = \gamma_i$.
We can rewrite $\gamma_M$ as $\gamma_M = \sum_{i \in I_1,j\in J} \omega(M_i,N_j) \cdot \gamma_{i,j}$.
Analogously, for each $j \in J_1$ and $i \in I$ we note that the transition $q_j \cdot \dirac{N_j} \TransS[\hat{\beta}_2] \gamma_{j}$ can always be splitted into $\sum_{i \in I} \omega(M_i,N_j) \cdot \dirac{N_j} \TransS[\hat{\beta}_2] \sum_{i \in I} \omega(M_i,N_j) \cdot \gamma'_{i,j}$
so that we can rewrite $\gamma_j$ as $\gamma_j = \sum_{i \in I}\omega(M_i,N_j)\cdot  \gamma'_{i,j}$ and $\gamma_N$ as
$\gamma_N = \sum_{i\in I, j \in J_1} \omega(M_i,N_j) \cdot \gamma'_{i,j}$. 
Then we note that for all $i \in I_1$ and $j \in J_1$ with $d(M_i,N_j) < 1$, the transition $N_j \TransS[\hat{\beta}_2] \gamma'_{i,j}$ can be chosen so that  
$\Kantorovich(d)(\gamma_{i,j},\gamma'_{i,j} + (1-\size{\gamma'_{i,j}})\dummyN) \le d(M_i,N_j)$.

For all $i \in I_1$ and $j \in J_1$ with $d(M_i,N_j) <1$, let $\omega_{i,j}$ be one of the optimal matchings realising $\Kantorovich(d)(\gamma_{i,j}, \gamma_j + (1-\size{\gamma_j}) \dummyN)$.
Define $\omega' \colon \nome \times \nome \to [0,1]$ as the function such that 
\[ 
\omega'(M',N') = 
\begin{cases} 
\sum_{i \in I_1, j \in J_1}\omega(M_i,N_j) \cdot \omega_{i,j}(M',N') 
& \text{ if } M' \neq \dummyN \neq N'\\
\sum_{i \in I_1, j \in J_1}\omega(M_i,N_j) \cdot \omega_{i,j}(M',N') + \sum_{i \in I_1, j \in J_2}\omega(M_i,N_j) \cdot \gamma_{i,j}(M')
& \text{ if } M' \neq \dummyN = N' \\
\sum_{i \in I_1, j \in J_1}\omega(M_i,N_j) \cdot \omega_{i,j}(M',N') + \sum_{i \in I_2, j \in J_1}\omega(M_i,N_j) \cdot \gamma'_{i,j}(N')
& \text{ if } M' = \dummyN \neq N' \\
\sum_{i \in I_1, j \in J_1}\omega(M_i,N_j) \cdot \omega_{i,j}(M',N') +  \sum_{i \in I_1, j \in J_2}\omega(M_i,N_j) \cdot \gamma_{i,j}(M')  \\
 + \sum_{i \in I_2, j \in J_1}\omega(M_i,N_j) \cdot \gamma'_{i,j}(N')+ \sum_{i \in I_2, j \in J_2}\omega(M_i,N_j) 
& \text{ if } M' = \dummyN = N'. 
\end{cases}
\] 

To infer the proof obligation $\Kantorovich(d)(\gamma_M + (1-\size{\gamma_M}) \dirac{\dummyN},\gamma_N + (1-\size{\gamma_N}) \dirac{\dummyN}) \le d(M,N)$ we show that
\begin{inparaenum}[(i)]
\item \label{matching} $\omega'$ is a matching in $\Omega(\gamma_M + (1-\size{\gamma_M}) \dirac{\dummyN},\gamma_N + (1-\size{\gamma_N}) \dirac{\dummyN})$, and
\item \label{metric_condition} $\sum_{M',N' \in \nome} \omega'(M',N') \cdot d(M',N') \le d(M,N)$.
\end{inparaenum}

To show (\ref{matching}) we prove that the left marginal of $\omega'$ is $\gamma_M + (1-\size{\gamma_M}) \dirac{\dummyN}$.
The proof that the right marginal is $\gamma_N + (1-\size{\gamma_N}) \dirac{\dummyN})$ is analogous.
For any \CPS{} $M' \neq \dummyN$, we have
\[
\begin{array}{rcl}
& & \sum_{N' \in \nome}\omega'(M',N') 
\\[0.5 ex]
= & \quad & 
\sum_{N' \neq \dummyN} \sum_{i \in I_1, j \in J_1}\omega(M_i,N_j) \cdot \omega_{i,j}(M',N') 
+
\sum_{i \in I_1, j \in J_1}\omega(M_i,N_j) \cdot \omega_{i,j}(M',\dummyN)
\\
& \quad & 
+ 
\sum_{i \in I_1, j \in J_2}\omega(M_i,N_j) \cdot \gamma_{i,j}(M') 
\\[0.5 ex]
= & \quad & 
\sum_{i \in I_1, j \in J_1}\omega(M_i,N_j) \sum_{N' \in \nome} \omega_{i,j}(M',N') 
+ 
\sum_{i \in I_1, j \in J_2}\omega(M_i,N_j) \cdot \gamma_{i,j}(M')
\\[0.5 ex]
 = & \quad & 
%= (since $\omega_{i,j}$ is a matching in $\Omega(\gamma_{i,j},\gamma'_{i,j})$)
\sum_{i \in I_1, j \in J_1}\omega(M_i,N_j) \cdot \gamma_{i,j}(M') 
+ 
\sum_{i \in I_1, j \in J_2}\omega(M_i,N_j) \cdot \gamma_{i,j}(M') 
\\[0.5 ex]
= & \quad & 
\sum_{i \in I_1, j \in J}\omega(M_i,N_j) \cdot \gamma_{i,j}(M') 
\\[0.5 ex]
= & \quad & 
(\gamma_M+ (1-\size{\gamma_M}) \dummyN)(M')
\end{array}
\]
with the third equality by the fact that $\omega_{i,j}$ is a matching in $\Omega(\gamma_{i,j},\gamma'_{i,j})$.

%Consider now the \CPS{} 
Consider now the \CPS{} $\dummyN$. In this case we have that
\[
\begin{array}{rcl}
& & 
\sum_{N' \in \nome}\omega'(\dummyN,N')
\\[0.5 ex] 
= & \quad & 
\sum_{N' \neq \dummyN}\sum_{i \in I_1, j \in J_1}\omega(M_i,N_j) \cdot \omega_{i,j}(\dummyN,N')
+ \sum_{N' \neq \dummyN} \sum_{i \in I_2, j \in J_1}\omega(M_i,N_j) \cdot \gamma'_{i,j}(N')
\\
& \quad & 
+ \sum_{i \in I_1, j \in J_1}\omega(M_i,N_j) \cdot \omega_{i,j}(\dummyN,\dummyN)
+ \sum_{i \in I_1, j \in J_2}\omega(M_i,N_j) \cdot \gamma_{i,j}(\dummyN)
\\
& \quad & 
+ \sum_{i \in I_2, j \in J_1}\omega(M_i,N_j) \cdot \gamma'_{i,j}(\dummyN)
+ \sum_{i \in I_2, j \in J_2}\omega(M_i,N_j)
%Then, this summation is equal to 
\\[0.5 ex] 
= & \quad & 
\sum_{N' \in \nome}\sum_{i \in I_1, j \in J_1}\omega(M_i,N_j) \cdot \omega_{i,j}(\dummyN,N') +
\sum_{N' \in \nome} \sum_{i \in I_2, j \in J_1}\omega(M_i,N_j) \cdot \gamma'_{i,j}(N') 
\\
& \quad & 
+ \sum_{i \in I_1, j \in J_2}\omega(M_i,N_j) \cdot \gamma_{i,j}(\dummyN) +
\sum_{i \in I_2, j \in J_2}\omega(M_i,N_j)
\\[0.5 ex]
%Since $\omega_{i,j}$ is a matching in $\Omega(\gamma_{i,j},\gamma'_{i,j})$, the first summand is
%$\sum_{i \in I_1, j \in J_1}\omega(M_i,N_j) \cdot \gamma_{i,j}(\dummyN)$.
%Then, since $\gamma'_{i,j}$ is a distribution, the second summand is 
%$\sum_{i \in I_2, j \in J_1}\omega(M_i,N_j)$.
= & \quad & 
\sum_{i \in I_1, j \in J_1}\omega(M_i,N_j) \cdot \gamma_{i,j}(\dummyN)+
\sum_{i \in I_2, j \in J_1}\omega(M_i,N_j)
\\ 
& \quad & 
+ \sum_{i \in I_1, j \in J_2}\omega(M_i,N_j) \cdot \gamma_{i,j}(\dummyN)
+ \sum_{i \in I_2, j \in J_2}\omega(M_i,N_j)
\\[0.5 ex]
= & \quad &
\sum_{i \in I_1, j \in J}\omega(M_i,N_j) \cdot \gamma_{i,j}(\dummyN) +
\sum_{i \in I_2, j \in J}\omega(M_i,N_j)
\\[0.5 ex]
= & \quad &
(\gamma_M + (1-\size{\gamma_M}) \dirac{\dummyN})(\dummyN) 
\end{array}
\]
where the third equality follows by observing that, being $\omega_{i,j}$ a matching in $\Omega(\gamma_{i,j},\gamma'_{i,j})$, then we have 
$\sum_{N' \in \nome}\sum_{i \in I_1, j \in J_1}\omega(M_i,N_j) \cdot \omega_{i,j}(\dummyN,N')= \sum_{i \in I_1, j \in J_1}\omega(M_i,N_j) \cdot \gamma_{i,j}(\dummyN)$, and being
$\gamma'_{i,j}$ a distribution, then 
$\sum_{N' \in \nome} \sum_{i \in I_2, j \in J_1}\omega(M_i,N_j) \cdot \gamma'_{i,j}(N') = \sum_{i \in I_2, j \in J_1}\omega(M_i,N_j)$, and the last equality follows by
%where in the last step we exploit 
$\sum_{i \in I_1, j \in J}\omega(M_i,N_j) = \sum_{i \in I_1} p_i = \size{\gamma_M}$.

To prove (\ref{metric_condition}), by looking at the definition of $\omega'$ above we get that $\sum_{M',N' \in \nome} \omega'(M',N') \cdot d(M',N')$ is the summation of the following values:
\begin{itemize}
\item
$\sum_{M' \neq \dummyN \neq N'} \sum_{i \in I_1, j \in J_1}\omega(M_i,N_j) \cdot \omega_{i,j}(M',N') \cdot d(M',N')$
\item
$\sum_{ M' \neq \dummyN}\sum_{i \in I_1, j \in J_1}\omega(M_i,N_j) \cdot \omega_{i,j}(M',\dummyN) \cdot d(M',\dummyN) + \sum_{i \in I_1, j \in J_2}\omega(M_i,N_j) \cdot \gamma_{i,j}(M') \cdot d(M',\dummyN)$
\item
$\sum_{N' \neq \dummyN}\sum_{i \in I_1, j \in J_1}\omega(M_i,N_j) \cdot  \omega_{i,j}(\dummyN,N') \cdot d(\dummyN,N') + \sum_{i \in I_2, j \in J_1}\omega(M_i,N_j) \cdot \gamma'_{i,j}(N') \cdot d(\dummyN,N')$
\item
$\sum_{i \in I_1, j \in J_1}\omega(M_i,N_j) \cdot \omega_{i,j}(\dummyN,\dummyN) \cdot d(\dummyN,\dummyN)+  \sum_{i \in I_1, j \in J_2}\omega(M_i,N_j) \cdot \gamma_{i,j}(\dummyN) \cdot d(\dummyN,\dummyN) \\
+ \sum_{i \in I_2, j \in J_1}\omega(M_i,N_j) \cdot \gamma'_{i,j}(\dummyN)\cdot d(\dummyN,\dummyN)+ \sum_{i \in I_2, j \in J_2}\omega(M_i,N_j)\cdot d(\dummyN,\dummyN)$.
\end{itemize}
By moving the first summand of the second, third and fourth items to the first item, we rewrite this summation as the summation of the following values:
\begin{itemize}
\item
$\sum_{M' , N' \in \nome} \sum_{i \in I_1, j \in J_1} \omega(M_i,N_j) \cdot \omega_{i,j}(M',N') \cdot d(M',N')$
\item
$\sum_{i \in I_1, j \in J_2}\omega(M_i,N_j) \cdot \gamma_{i,j}(M') \cdot d(M',\dummyN)$
\item
$\sum_{i \in I_2, j \in J_1}\omega(M_i,N_j) \cdot \gamma'_{i,j}(N') \cdot d(\dummyN,N')$
\item
$\sum_{i \in I_1, j \in J_2}\omega(M_i,N_j) \cdot \gamma_{i,j}(\dummyN) \cdot d(\dummyN,\dummyN)
+ \sum_{i \in I_2, j \in J_1}\omega(M_i,N_j) \cdot \gamma'_{i,j}(N')\cdot d(\dummyN,\dummyN)+ \sum_{i \in I_2, j \in J_2}\omega(M_i,N_j)\cdot d(\dummyN,\dummyN)$.
\end{itemize}
By the definition of $\omega_{i,j}$ the first item is $\sum_{i \in I_1, j \in J_1} \omega(M_i,N_j)\cdot  \Kantorovich(d)(\gamma_{i,j},\gamma'_{i,j})$.
If $d(M_i,N_j) < 1$, we chosen $\gamma'_{i,j}$ such that $\Kantorovich(d)(\gamma_{i,j},\gamma'_{i,j}) \le d(M_i,N_j)$.
If $d(M_i,N_j) = 1$, then $\Kantorovich(d)(\gamma_{i,j},\gamma'_{i,j}) \le d(M_i,N_j)$ is immediate.
Henceforth we are sure that in all cases the first item is less or equal $\sum_{i \in I_1, j \in J_1} \omega(M_i,N_j) \cdot d(M_i,N_j)$.
The second item is clearly less or equal than $\sum_{i \in I_1, j \in J_2}\omega(M_i,N_j)$.
The third item is clearly less or equal than $\sum_{i \in I_2, j \in J_1}\omega(M_i,N_j)$.
%The third item is 0, since $d(\dummyN,N') = 0$ for all $N' \in \nome$.
Finally, the last item is 0 since $d(\dummyN,\dummyN) = 0$.
Summarising, we have $\sum_{M',N' \in \nome} \omega'(M',N') \cdot d(M',N') \le \sum_{i \in I_1, j \in J_1} \omega(M_i,N_j) \cdot d(M_i,N_j) + \sum_{i \in I_1, j \in J_2}\omega(M_i,N_j)
+ \sum_{i \in I_2, j \in J_1}\omega(M_i,N_j)$.
Since $\Kantorovich(d)(\rho_M ,\rho_N )$ is the summation of the following values:
\begin{itemize}
\item
$\sum_{i \in I_1,j\in J_1} \omega(M_i,N_j) \cdot d(M_i,N_j)$
\item
$\sum_{i \in I_1,j\in J_2} \omega(M_i,N_j) \cdot d(M_i,N_j) = \sum_{i \in I_1,j\in J_2} \omega(M_i,N_j)$ ($M_i \trans{\beta_2}$ and $N_j \not\!\!\TransS[\hat{\beta}_2]$ give $d(M_i,N_j) =1$) 
\item
$\sum_{i \in I_2,j\in J_1} \omega(M_i,N_j) \cdot d(M_i,N_j)  = \sum_{i \in I_2,j\in J_1} \omega(M_i,N_j)$ ($N_j \trans{\beta_2}$ and $M_i \not\!\!\TransS[\hat{\beta}_2]$ give $d(M_i,N_j) =1$) 
\item
$\sum_{i \in I_2,j\in J_2} \omega(M_i,N_j) \cdot d(M_i,N_j)$.
\end{itemize}
it follows $\sum_{i \in I_1, j \in J_1} \omega(M_i,N_j) \cdot d(M_i,N_j) + \sum_{i \in I_1, j \in J_2}\omega(M_i,N_j) + \sum_{i \in I_2, j \in J_1}\omega(M_i,N_j) \le \Kantorovich(d)(\rho_M ,\rho_N )$.
Since we had $\Kantorovich(d)(\rho_M ,\rho_N) \le d(M,N)$ we can conclude $\sum_{M',N' \in \nome} \omega'(M',N') \cdot d(M',N') \le d(M,N)$, as required.
\end{proof}

We are now ready to prove that all $\metric^n$ are pseudometrics.
\\

\noindent
\textbf{Proof of  \autoref{prop:up-to-k-metric}} \hspace{0.2 cm}
We have to prove that $\metric^{n}(M,M) = 0$,  $\metric^{n}(M,N) = \metric^n(N,M)$ and $\metric^{n}(M,N) \le \metric^{n}(M,O)+\metric^{n}(O,N)$ for all $M,N,O \in \nome$.
We reason by induction over $n$. 
The base case $n=0$ is immediate since $\metric^0(M,N) = 0$ for all $M,N \in \nome$.
We consider the inductive step $n+1$.

Let us start with proving $\metric^{n+1}(M,M) = 0$.
We have to show that for each transition $M \transS[\alpha] \gamma$ there is a transition $M \TransS[\hat{\alpha}] \rho$ with 
$ \Kantorovich(\metric^{n})(\gamma, \rho + (1-\size{\rho})\dirac{\dummyN}) = 0$.
We choose $\rho = \gamma$ and the transition $M \transS[\alpha] \gamma$.
We obtain $ \Kantorovich(\metric^{n})(\gamma, \rho + (1-\size{\rho})\dirac{\dummyN})$ = $ \Kantorovich(\metric^{n})(\gamma,\gamma)$ = $0$, with the last equality by the inductive hypothesis and  \autoref{prop_kant_metric}. 

The symmetry $\metric^{n+1}(M,N) = \metric^{n+1}(N,M)$ follows by 
$ \metric^{n+1}(M,N) = \Bisimulation(\metric^{n})(M,N) = \Bisimulation(\metric^{n})(N,M) = \metric^{n+1}(N,M)$, where the second equality follows immediately by the definition of $\Bisimulation$.

Finally we prove the triangular property $\metric^{n+1}(M,N) \le \metric^{n+1}(M,O) +\metric^{n+1}(O,N)$.
This result is immediate if $\metric^{n+1}(M,O) =1$ or $\metric^{n+1}(O,N) =1$.
Otherwise, it is enough to prove that any $M \transS[\alpha] \gamma_M$ is mimicked by some transition $N \TransS[\hat{\alpha}] \gamma_N$ with $ \Kantorovich(\metric^{n})(\gamma_M,\gamma_N + (1-\size{\gamma_N})\dirac{\dummyN}) \le \metric^{n+1}(M,O) + \metric^{n+1}(O,N)$.
From $M \transS[\alpha] \gamma_M$ and $\metric^{n+1}(M,O) < 1$ we immediately infer that there is a transition $O \TransS[\hat{\alpha}] \gamma_O$ with $ \Kantorovich(\metric^{n})(\gamma_M ,\gamma_O + (1-\size{\gamma_O})\dirac{\dummyN}) \le \metric^{n+1}(M,O)$.
By  \autoref{lemma_sim_weak_transitions}, from $O \TransS[\hat{\alpha}] \gamma_O$ and $\metric^{n+1}(O,N) <1$ there is a transition $N \TransS[\hat{\alpha}] \gamma_N$ such that $ \Kantorovich(\metric^{n})(\gamma_O + (1-\size{\gamma_O})\dirac{\dummyN},\gamma_N + (1-\size{\gamma_N})\dirac{\dummyN}) \le \metric^{n+1}(O,N)$.
By the inductive hypothesis and  \autoref{prop_kant_metric} we get that  $\Kantorovich(\metric^{n})$ is a pseudometric, hence it satisfies the triangle inequality, namely $\Kantorovich(\metric^{n})(\gamma_M,\gamma_N+ (1-\size{\gamma_N})\dirac{\dummyN}) \le \Kantorovich(\metric^{n})(\gamma_M, \gamma_O + (1-\size{\gamma_O})\dirac{\dummyN}) + \Kantorovich(\metric^{n})(\gamma_O+ (1-\size{\gamma_O}) \dirac{\dummyN},\gamma_N + (1-\size{\gamma_N})\dirac{\dummyN})$.
Therefore we can conclude the proof by $\Kantorovich(\metric^{n})(\gamma_M,\gamma_N+ (1-\size{\gamma_N})\dirac{\dummyN}) \le  \Kantorovich(\metric^{n})(\gamma_M, \gamma_O + (1-\size{\gamma_O})\dirac{\dummyN}) +  \Kantorovich(\metric^{n})(\gamma_O+ (1-\size{\gamma_O}) \dirac{\dummyN},\gamma_N + (1-\size{\gamma_N})\dirac{\dummyN})
 \le \metric^{n+1}(M,O) + \metric^{n+1}(O,N)$.
\qed
\\

%\noindent
%\textbf{Proof of  \autoref{prop:metric_as_a_limit}} \hspace{0.2 cm}
%Follows by the same arguments in \cite{vB12}.
%\qed
%\\

In order to prove the compositionality or our weak bisimilarity metrics, i.e.\ \autoref{thm:congruenceP}, we divide its statement in six different propositions.
To prove that $\approx_p$ preserves the compositionality we need a number of technical lemmas.

Given a distribution $\gamma$ over \CPS{s} and a \CPS{} $O$, we denote with $\gamma \uplus O$ the distribution defined by $(\gamma \uplus O)(M \uplus O) = \gamma(M)$ for all \CPS{s} $M$.

\autoref{lem:aux1P} serves to propagate untimed actions on parallel \CPS{s}.

\begin{lemma} 
\label{lem:aux1P}
Assume two physically disjoint \CPS{}s $M_1$ and $M_2$ such that $M_2 =  E_2; \confCPS {\state_2}  {P_2}$ and 
$E_2 = \envCPS
{\evolmap^2{} }
{\measmap^2{} }   
{\invariantfun^2{} }$.
If $M_1  \trans{\alpha} \gamma$, with $\alpha \neq \tick$, and $\state_2 \in \invariantfun{}^2$ 
then $M_1 \uplus M_2 \trans{\alpha} \gamma \uplus   M_2 $.  
\end{lemma}

\begin{proof}
If $M_1$ is the \CPS{} $\dummyN$ then also $M_1 \uplus M_2$ is $\dummyN$  and the thesis is immediate. 
Consider the case $M_1 \neq \dummyN$.
Let us assume that $M_1 = E_1; \confCPS {\state_1}  {P_1}$ with $E_1 = \envCPS
{\evolmap^1{} }
{\measmap^1{} }   
{\invariantfun^1{} }$ 
and $\state_1 = \stateCPS {\statefun^1{}} {\sensorfun^1{} } {\actuatorfun^1{} }$.
Moreover, assume that $\state_2 = \stateCPS {\statefun^2{}} {\sensorfun^2{} } {\actuatorfun^2{} }$.
We consider the case in which $M_1 \trans{\alpha} \gamma$ is derived by rule \rulename{SensRead}.
The other cases where the transition is derived by the other rules in \autoref{tab:lts_systems_P} can be proved in a similar manner.
In this case, we have $\alpha=\tau$ and there are a sensor $s$, probability values $p_i$ and real values $v_i$ with $i\in I$ and a distribution $\pi$ such that 
the rule \rulename{SensRead} instances as  
\[
\Txiombis{P_1 \trans{\rcva s z} \pi \Q \Q 
 \sensorfun^1{}(s) = \sum_{i \in I} p_i \cdot \dirac{v_i}  \Q \Q 
\statefun^1{}  \in \invariantfun^1{} 
}
{\confCPS {{\stateCPS {\statefun^1{}} {\sensorfun^1{}} {\actuatorfun^1{}}}}  P_1 \trans{\tau} \confCPS {\dirac{{\stateCPS {\statefun^1{}} {\sensorfun^1{}} {\actuatorfun^1{}}}}} {\sum_{i \in I}p_i \cdot   \pi \subst{v_i}{z}}} 
\]
and $\gamma= E_1 ; \confCPS {\dirac{{\stateCPS {\statefun^1{}} {\sensorfun^1{}} {\actuatorfun^1{}}}}} {\sum_{i \in I}p_i \cdot  \pi \subst{v_i}{z}}$.

Now we argue that we can apply rule \rulename{SensRead} to infer a transition by $M_1 \uplus M_2$. 
Recall that $M_1 \uplus M_2$ is the \CPS{} $(E_1 \uplus E_2) ; \confCPS { \stateCPS {\statefun^1{} \uplus \statefun^2{}} {\sensorfun^1{} \uplus \sensorfun^2{}} {\actuatorfun^1{} \uplus \actuatorfun^2{}}}{P_1 \parallel P_2}$.
Let $E_1 \uplus E_2 =  \envCPS
{\evolmap{} }
{\measmap{} }   
{\invariantfun{} }$.
From $P_1  \trans{\rcva s z} \pi$, by rule \rulename{Par} in \autoref{tab:lts_processes} we can derive the transition
$P_1 \parallel P_2 \trans{\rcva s z} \pi \parallel \dirac{ P_2}$, which is one of the premises of  rule \rulename{SensRead} necessary to infer a transition by $\confCPS { \stateCPS {\statefun^1{} \uplus \statefun^2{}} {\sensorfun^1{} \uplus \sensorfun^2{}} {\actuatorfun^1{} \uplus \actuatorfun^2{}}}{P_1 \parallel P_2}$.
Then, the premise $\statefun^1{} \uplus \statefun^2 \in \invariantfun{} $ of  \rulename{SensRead} follows by $\statefun^1{}  \in \invariantfun^1{}$, the hypothesis $\statefun^2{}  \in \invariantfun{}^2$ and the property 
$\statefun^1{} \uplus \statefun^2 \in \invariantfun{}$ iff  $\statefun^1{}  \in \invariantfun^1{}$ and $\statefun^2{}  \in \invariantfun{}^2$.
Finally, the premise $(\sensorfun^1{} \uplus \sensorfun^2{})(s) = \sum_{i \in I} p_i \cdot \dirac{v_i}$  follows by $(\sensorfun^1{} \uplus \sensorfun^2{})(s) = \sensorfun^1{}(s)$ and
$\sensorfun^1{}(s) = \sum_{i \in I} p_i \cdot \dirac{v_i}$. 
Therefore we have
\[
\Txiombis{P_1 \parallel P_2 \trans{\rcva s z} \pi \parallel \dirac{ P_2}\Q \Q 
(\sensorfun^1{} \uplus \sensorfun^2{})(s) = \sum_{i \in I} p_i \cdot \dirac{v_i} \Q \Q 
\statefun^1{} \uplus \statefun^2 \in \invariantfun{} 
}
{\confCPS { \stateCPS {\statefun^1{} \uplus \statefun^2{}} {\sensorfun^1{} \uplus \sensorfun^2{}} {\actuatorfun^1{} \uplus \actuatorfun^2{}}}{P_1 \parallel P_2}\trans{\tau} \confCPS {\dirac{ \stateCPS {\statefun^1{} \uplus \statefun^2{}} {\sensorfun^1{} \uplus \sensorfun^2{}} {\actuatorfun^1{} \uplus \actuatorfun^2{}}}} {\sum_{i \in I}p_i \cdot  (\pi \parallel \dirac{ P_2}) \subst{v_i}{z}}}
\]
with $(E_1\uplus E_2) ; \confCPS {\dirac{ \stateCPS {\statefun^1{} \uplus \statefun^2{}} {\sensorfun^1{} \uplus \sensorfun^2{}} {\actuatorfun^1{} \uplus \actuatorfun^2{}}}} {\sum_{i \in I}p_i \cdot (\pi\parallel \dirac{P_2}) \subst{v_i}{z}} = 
\gamma \uplus M_2$. 
\end{proof}

\autoref{lem:aux1P} can be generalised to weak transitions.

\begin{lemma}
\label{lem:aux1Pbis}
Assume two physically disjoint \CPS{}s $M_1$ and $M_2$ such that $M_2 =  E_2; \confCPS {\state_2}  {P_2}$ and 
$E_2 = \envCPS
{\evolmap^2{} }
{\measmap^2{} }   
{\invariantfun^2{} }$.
If $M_1  \Trans{\widehat{\alpha}} \gamma$, with $\alpha \neq \tick$, and $\state_2 \in \invariantfun{}^2$ 
then $M_1 \uplus M_2 \Trans{\widehat{\alpha}} \gamma \uplus   M_2 $. 
\end{lemma}

\begin{proof}
By induction over the length $n$ of $\TransS[\widehat{\alpha}]$.
The base case $n=1$ is given by \autoref{lem:aux1P}.
Consider the inductive step $n+1$.
We have $M_1 \TransS[\widehat{\alpha_1}] \gamma' \transS[\widehat{\alpha_2}] \gamma$ with either $\alpha_1 = \alpha$ and $\alpha_2 = \tau$, or $\alpha_1 = \tau$ and $\alpha_2 = \alpha$.
Since the length of $\TransS[\widehat{\alpha_1}]$ is $n$, we can apply the inductive hypothesis and infer 
$M_1 \uplus M_2  \TransS[\widehat{\alpha_1}]  \gamma' \uplus M_2$.
Assume $\gamma' = \sum_{i \in I}p_i \cdot \dirac{M_i}$, for suitable probability values $p_i$ and \CPS{} $M_i$. 
By definition, $\gamma' \transS[\widehat{\alpha_2}] \gamma$ implies that there exists a subset $J \subseteq I$ with
$M_j \transS[\widehat{\alpha_2}] \gamma_j$ for all $j \in J$, $M_i \ntransS[\alpha_2]$ for $i \in I \setminus J$ and $\gamma = \sum_{j \in J} p_j \cdot \dirac{M_j}$.
We can prove now that for any $j \in J$ we have $M_j \uplus M_2 \transS[\widehat{\alpha_2}] \gamma_j\ \uplus M_2$.
We distinguish two cases.
The first case is $M_j \transS[\alpha_2] \gamma_j$. 
By \autoref{lem:aux1P} we get $M_j \uplus M_2 \transS[\alpha_2] \gamma_j\ \uplus M_2$, and, therefore, $M_j \uplus M_2 \transS[\widehat{\alpha_2}] \gamma_j\ \uplus M_2$.
The second case is $\alpha_2 = \tau$ and $\gamma_j = \dirac{M_j}$.
We immediately have $M_j \uplus M_2 \transS[\widehat{\tau}] \gamma_j\ \uplus M_2$.
Hence $\sum_{j\in J} M_j \uplus M_2 \transS[\widehat{\alpha_2}] \sum_{j \in J} \gamma_j  \uplus M_2$, namely  $\gamma' \uplus M_2 \transS[\widehat{\alpha_2}] \gamma \uplus M_2$.
Then, from $M \uplus M_2  \TransS[\widehat{\alpha_1}]  \gamma' \uplus M_2$ and $\gamma' \uplus M_2 \transS[\widehat{\alpha_2}] \gamma \uplus M_2$ we get $M \uplus M_2  \TransS[\widehat{\alpha}] \gamma \uplus M_2$, which completes the proof.
\end{proof}

Next lemma says that the invariants of \CPS{s} in distance $< 1$ must agree.

\begin{lemma}
\label{lem:bis-invP}
Assume two \CPS{}s $M_1$ and $M_2$ such that $M_i =  E_i; \confCPS {\state_i}  {P_i}$ and 
$E_i = \envCPS
{\evolmap^i{} }
{\measmap^i{} }   
{\invariantfun^i{} }$, for $i=1,2$.
If
$\metric(M_1, M_2) < 1$ 
then $S_1 \in \invariantfun{}^1 $ iff $ S_2 \in \invariantfun{}^2$.
\end{lemma}

\begin{proof}
The proof is by contradiction. 
Assume that $\metric(M_1, M_2) < 1$, $S_1 \in \invariantfun{}^1 $ and $ S_2 \not \in \invariantfun{}^2$.
We show that $M_1\TransS[\widehat{\tick}]$ and $M_2 \nTransS[\widehat{\tick}]$, which contradicts $\metric(M_1, M_2) < 1$.
By the well timedness property for \CPS{s} (\autoref{prop:time}, last item), there exists a natural $n$ such that all derivations $M_1 \trans{\tau} N_1 \trans{\tau}\dots \trans{\tau} N_k$ are such that $k \le n$, then we have $N_k \ntrans \tau$. 
Since $N_k \ntrans \tau$, by the maximal progress property for \CPS{s} (\autoref{prop:time}, second item) it follows that $N_k \trans \tick \gamma$, for some $\gamma$.
We conclude $M_1\TransS[\widehat{\tick}]$. 
Since $ S_2 \not \in \invariantfun{}^2$, the \CPS{} $M_2$ can  perform only the step  $M_2\trans \tau \dummyN$ and $\dummyN$ can not perform  any 
action, and hence, $M_2 \nTransS[\widehat{\tick}]$. 
\end{proof}

Here comes one of the main technical result: the bisimilarity metric is preserved by the parallel composition of physically disjoint \CPS{s}. 

\begin{proposition}
\label{lem:cong1P}
$\metric(M \uplus O ,N \uplus O) \le \metric(M , N) $, for any physically disjoint \CPS{} $O$. 
\end{proposition}

\begin{proof}
The case $\metric(M , N) = 1$ is immediate, therefore we assume $\metric(M , N) < 1$. 
Let us define the function $d \colon \nome \times \nome \to [0,1]$ by $d(M \uplus O , N \uplus O) = \metric(M,N)$ for all $M,N,O \in \nome$.
To prove the thesis it is enough to show that $d$ is a weak bisimulation metric. 
In fact, since $\metric$ is the minimal weak bisimulation metric, we infer $\metric \sqsubseteq d$, thus giving $\metric(M \uplus O , N \uplus O) \le d(M \uplus O , N \uplus O) = \metric(M,N)$.
To prove that $d$ is a weak bisimulation metric, we show that any transition $M \uplus  O \trans{\alpha} \gamma$ is simulated by some transition $N \uplus  O \TransS[\widehat{\alpha}] \gamma'$ with $\Kantorovich(d)(\gamma,\gamma' + (1-\size{\gamma'})\dummyN ) \le d(M \uplus  O,N \uplus  O)$.
The cases where one of the \CPS{s} $M$, $N$ and $O$ are $\dummyN$ is immediate.
Hence, assume that $M$, $N$ and $O$ are not $\dummyN$.
Let us assume that $M_1 = E_1; \confCPS {\state_1}  {P_1}$ with $E_1 = \envCPS
{\evolmap^1{} }
{\measmap^1{} }   
{\invariantfun^1{} }$ 
and $\state_1 = \stateCPS {\statefun^1{}} {\sensorfun^1{} } {\actuatorfun^1{} }$.
Moreover, assume that $O = E_2; \confCPS {\state_2}  {P_2}$ with $E_2 = \envCPS
{\evolmap^2{} }
{\measmap^2{} }   
{\invariantfun^2{} }$ 
and $\state_2 = \stateCPS {\statefun^2{}} {\sensorfun^2{} } {\actuatorfun^2{} }$.
Finally  $E_1 \uplus E_2 =  \envCPS 
{\evolmap{} } 
{\measmap{} }   
{\invariantfun{} }
$.

We proceed by case analysis on how $M \uplus O \trans{\alpha} \gamma$ is derived.
The cases are the following:
\begin{itemize}

\item
The transition $M \uplus O \trans{\tau} \gamma$ is derived by rule \rulename{SensRead} in \autoref{tab:lts_systems_P}, instantiated as
\[
\Txiombis{P_1 \parallel P_2 \trans{\rcva s z} \pi \Q \Q 
(\sensorfun^1{} \uplus \sensorfun^2{})(s) = \sum_{i \in I} p_i \cdot \dirac{v_i} \Q \Q 
\statefun^1{} \uplus \statefun^2 \in \invariantfun{} 
}
{ \confCPS { S_1 \uplus S_2}{P_1 \parallel P_2}\trans{\tau} \confCPS {\dirac{ S_1 \uplus S_2}} {\sum_{i \in I}p_i \cdot  {\pi \subst{v_i}{z}}}} 
\]
with $\gamma = (E_1\uplus E_2) ; \confCPS {\dirac{ S_1 \uplus S_2} }{\sum_{i \in I}p_i \cdot   {\pi \subst{v_i}{z}}} $. 
 
\item
The transition $M \uplus  O \trans{\tau} \gamma$ is derived by rule \rulename{ActWrite} in \autoref{tab:lts_systems_P}
instantiated as
\[
\Txiombis{P_1 \parallel P_2 \trans{\snda a v} {\pi}  \Q \Q  { \statefun^1{} \uplus \statefun^2 \in \invariantfun{} }}
{\confCPS {\stateCPS {\statefun^1{} \uplus \statefun^2{}} {\sensorfun^1{}\uplus \sensorfun^2{}} {\actuatorfun^1{} \uplus \actuatorfun^2{}}}  P_1 \parallel P_2 \trans{\tau}
 \confCPS {\dirac{\stateCPS {\statefun^1{} \uplus \statefun^1{}} {\sensorfun^1{}\uplus \sensorfun^2{}} {\actuatorfun^1{}\uplus\actuatorfun^1{} [a \mapsto v]} }}{\pi}}
\]

\item
The transition $M \uplus O \trans{\tau} \gamma$ is derived by rule \rulename{Tau} in \autoref{tab:lts_systems_P}, instantiated as
\[
\Txiombis{P_1 \parallel P_2 \trans{\tau} \pi \Q\Q (S_1 \uplus S_2) \in \invariantfun{}}
{ \confCPS{S_1 \uplus S_2} {P_1 \parallel P_2}\trans{\tau} \confCPS {\dirac{S_1 \uplus S_2}}  {\pi} }
\]
with $\gamma = \confCPS {(E_1 \uplus E_2); \dirac{S_1 \uplus S_2}}  {\pi}$. 

\item
The transition $M \uplus O \trans{\tick} \gamma$ is derived by rule \rulename{Time} in \autoref{tab:lts_systems_P}, instantiated as 
 \[  \Txiombis
  {P_1\parallel P_2  \trans{\tick} {\pi} \Q\Q
 \confCPS{S_1 \uplus S_2} {P_1 \parallel P_2} \ntrans{\tau} \Q\Q
(S_1 \uplus S_2) \in \invariantfun{}}
{\confCPS{S_1 \uplus S_2} {P_1 \parallel P_2}  \trans{\tick} 
\confCPS  {\operatorname{next}_{(E_1 \uplus E_2 )} (S_1 \uplus S_2 )} {\pi} } 
\]
with $ \gamma=\confCPS {   (E_1 \uplus E_2 ) : \operatorname{next}_{(E_1 \uplus E_2 )} (S_1 \uplus S_2 )} {\pi}  $.  

\item
The transition $M \uplus O \trans{\inp c v} \gamma$ is derived by rule \rulename{Inp} in \autoref{tab:lts_systems_P}, instantiated as
\[
\Txiombis
{P_1 \parallel P_2 \trans{\inp c v} \pi \Q (S_1 \uplus S_2) \in \invariantfun{}}
{\confCPS{S_1 \uplus S_2} {P_1 \parallel P_2}  \trans{\inp c v}   \confCPS {\dirac{S_1 \uplus S_2}} {\pi}}  
\]
with $\gamma=\confCPS {(E_1 \uplus E_2); \dirac{S_1 \uplus S_2} } {\pi}$.

\item
The transition $M \uplus O \trans{\out c v} \gamma$ is derived by rule \rulename{Out} in \autoref{tab:lts_systems_P}
instantiated as
\[
\Txiombis
{P_1 \parallel P_2  \trans{\out c v}  \pi  \Q \Q  S_1 \uplus S_2 \in \invariantfun{}}
{\confCPS {S_1 \uplus S_2}{  P_1 \parallel P_2}    \trans{\out c v}   \confCPS {\dirac{S_1 \uplus S_2}}  {\pi}}.
\]

\end{itemize}

We show only the first case, the other are analogous.
We recall that, by definition of operator $\uplus$, the physical environments $E_1$ and $E_2$ have different physical devices. 
Thus, there are two cases:
\begin{itemize}

\item 
$s$ is a sensor of $E_1$.
In this case, the transition $P_1 \parallel P_2 \trans{\rcva s z} \pi$ derives by rule \rulename{Par} in \autoref{tab:lts_processes} from 
$P_1 \trans{\rcva s z} \pi'$, where $\pi'$ is a process distribution such that $\pi=\pi' \parallel \dirac{ P_2}$.

First we argue that rule \rulename{SensRead} can be used to derive a transition by $M$.
From $(S_1\uplus S_2) \in \invariantfun{}$, by definition of $E_1 \uplus E_2$,  we get both $ S_1\in \invariantfun{}^1$ and 
$S_2 \in \invariantfun{}^2$.
From $(\sensorfun^1{} \uplus \sensorfun^2{})(s) = \sum_{i \in I} p_i \cdot \dirac{v_i}$,  since $s$ is a sensor of 
$\sensorfun^1{} $,
we derive $ \sensorfun^1{}  (s) = \sum_{i \in I} p_i \cdot \dirac{v_i}$. 
Summarising, we have $P_1 \trans{\rcva s z} \pi'$,   $S_1\in \invariantfun{}^1$, and 
$ \sensorfun^1{}  (s) = \sum_{i \in I} p_i \cdot \dirac{v_i}$, which allows us to apply rule \rulename{SensRead} and derive 
$\confCPS{S_1}{P_1} \trans{\tau}   \confCPS {\dirac{ S_1 } }{\sum_{i \in I}p_i \cdot   {(\pi') \subst{v_i}{z}}} $,
 namely $M \trans{\tau} \gamma'' =  \confCPS { E_1  ;\dirac{ S_1 } }{\sum_{i \in I}p_i \cdot {(\pi') \subst{v_i}{z}}} $. 
  
Then, from $M \trans{\tau} \gamma''$ and $\metric (M,N) < 1$, there is a distribution $\gamma'''$ such that $N \Trans{\widehat{\tau}} \gamma'''$ with $\Kantorovich(\metric)(\gamma'',\gamma''' + (1-\size{\gamma'''})\dummyN) \le \metric(M,N)$.
Since $S_2 \in  \invariantfun{}^2 $, 
by \autoref{lem:aux1Pbis} it follows that $N \uplus O \TransS[\widehat{\tau}]  \gamma''' \uplus O$.
Finally, we conclude that  $\gamma''' \uplus O$ is the distribution $\gamma'$ we were looking for by $\Kantorovich(d)(\gamma,\gamma''' \uplus O + (1-\size{\gamma''' \uplus O})\dummyN) = \Kantorovich(d)(\gamma'' \uplus O ,\gamma''' \uplus O + (1-\size{\gamma''' \uplus O})\dummyN) =  \Kantorovich(\metric)(\gamma'',\gamma''' (1-\size{\gamma'''})\dummyN) \le \metric(M,N) = d(M \uplus O , N \uplus O)$.

\item 
$s$ is a sensor of $E_2$.
In this case, the transition $P_1 \parallel P_2 \trans{\rcva s z} \pi$ derives by rule \rulename{Par} in \autoref{tab:lts_processes} from 
$P_2 \trans{\rcva s z} \pi'$, where $\pi'$ is a process distribution such that $\pi=\dirac{P_1} \parallel \pi'$.

Assume  $N = E_3; \confCPS {\state_3}  {P_3}$ with $E_3 = \envCPS
{\evolmap^3{} }
{\measmap^3{} }   
{\invariantfun^3{} }$ 
and $\state_3 = \stateCPS {\statefun^3{}} {\sensorfun^3{} } {\actuatorfun^3{} }$.
We show that rule \rulename{SensRead} allow us to infer $N \uplus O \trans{\tau} N \uplus \gamma''  $ for some $\gamma''$. 
  
By the rule \rulename{Par} we get $P_3 \parallel P_2 \trans{\rcva s z} \dirac{ P_3} \parallel \pi'$.
From $(S_1\uplus S_2) \in \invariantfun{}$, by definition of $E_1 \uplus E_2$. we get both $S_1 \in \invariantfun{}^1$ and 
$S_2 \in \invariantfun{}^2$.
Let  $E_1 \uplus E_3 =  \envCPS 
{\evolmap{'} } 
{\measmap{'} }   
{\invariantfun{'} }
$.
From $\metric(M ,N) < 1$ and $ S_1 \in \invariantfun{}^1$,  by \autoref{lem:bis-invP} it follows that 
$S_3 \in \invariantfun{}^3$ and so  $ (S_3 \uplus S_2)  \in \invariantfun{'}$.
From $(\sensorfun^1{} \uplus \sensorfun^2{})(s) = \sum_{i \in I} p_i \cdot \dirac{v_i}$,  since $s$ is a sensor of 
$\sensorfun^2{} $,
we derive $ \sensorfun^2{}  (s) = \sum_{i \in I} p_i \cdot \dirac{v_i}$. 
Hence  we derive $(\sensorfun^3{} \uplus \sensorfun^2{})(s) = \sum_{i \in I} p_i \cdot \dirac{v_i}$.
Summarising we have $P_3 \parallel P_2 \trans{\rcva s z} \dirac{P_3} \parallel \pi'$, $(S_3 \uplus S_2) \in \invariantfun{'}$ and 
$(\sensorfun^3{} \uplus \sensorfun^2{})(s) = \sum_{i \in I} p_i \cdot \dirac{v_i}$.
Hence, we can apply rule
\rulename{SensRead} to infer $N \uplus O \trans{\tau} \confCPS { (E_3 \uplus E_2 ) ;\dirac{ S_3 \uplus S_2 } }{\sum_{i \in I}p_i \cdot  (\dirac{P_3} \parallel \pi') \subst{v_i}{z}}  = N\uplus \gamma'' $ with 
$\gamma'' =\confCPS {  E_2  ;\dirac{ S_2   } }{\sum_{i \in I}p_i \cdot   {(\pi') \subst{v_i}{z}}}  $.
Finally, we can conclude that  $  \gamma '= N\uplus \gamma''$ is the distribution we were looking for by $\Kantorovich(d)(M \uplus \gamma'',N \uplus \gamma'') = \Kantorovich(\metric)(\dirac{M},\dirac{N}) = \metric(M,N) = d(M \uplus O, N \uplus O)$.

\end{itemize}

\end{proof}

Also the $n$-weak bisimilarity metric is preserved by the parallel composition of physically disjoint  \CPS{s}.

\begin{proposition}
\label{lem:cong1Pk}
$\metric^n(M \uplus O ,N \uplus O) \le \metric^n(M , N) $, for any  physically disjoint \CPS{} $O$ and $n \ge 0$. 
\end{proposition}
\begin{proof}
We proceed by induction over $n$. 
The base case $n=0$ is immediate since $\metric^n(M,N) = \zeroF(M,N) = 0$ for all $M,N \in \nome$.
We consider the inductive step $n+1$.
The case $\metric^{n+1}(M , N) = 1$ is immediate, therefore we assume $\metric^{n+1}(M , N) < 1$. 
We have to show that any transition $M \uplus  O \trans{\alpha} \gamma$ is simulated by some transition $N \uplus  O \TransS[\widehat{\alpha}] \gamma'$ with $\Kantorovich(\metric^n)(\gamma,\gamma' + (1-\size{\gamma'})\dummyN ) \le \metric^{n+1}(M \uplus  O,N \uplus  O)$.
This can be shown precisely as in the proof of \autoref{lem:cong1P}. 
Essentially, we have to replace all occurrences of $\metric(M,N)$ by $\metric^{n+1}(M,N)$ and all occurrences of $\Kantorovich(d)(\gamma,\gamma')$ and $\Kantorovich(\metric)(\gamma,\gamma')$ by $\Kantorovich(\metric^n)(\gamma,\gamma')$.
\end{proof}

Now we prove that our weak bisimilarity metrics are preserved by parallel composition of pure-logical processes. 
These are special cases of \autoref{lem:cong1P} and \autoref{lem:cong1Pk}. 
\begin{proposition}
\label{lem:cong2P}
$\metric(M \parallel P , N \parallel P) \le \metric(M,N)$, for any pure-logical process $P$.
\end{proposition}
\begin{proof}
Let $E_\emptyset$ be the physical environment with an empty set of state variables, sensors and actuators. 
Let $S_\emptyset$ be the unique (empty)  physical state   of $E_\emptyset$.
We have
$
\metric(M \parallel P , N \parallel P) \le 
\metric(M \parallel P , M \uplus (\confCPS {E_\emptyset ; S_\emptyset}  {P })) + \metric(M \uplus (\confCPS {E_\emptyset ; S_\emptyset}  {P }) , N \parallel P) = 
\metric(M \uplus (\confCPS {E_\emptyset ; S_\emptyset}  {P }) , N \parallel P) \le
\metric(M \uplus (\confCPS {E_\emptyset ; S_\emptyset}  {P }),  N \uplus (\confCPS {E_\emptyset ; S_\emptyset}  {P })) + \metric(N \uplus (\confCPS {E_\emptyset ; S_\emptyset}  {P }) , N \parallel P)
= \metric(M \uplus (\confCPS {E_\emptyset ; S_\emptyset}  {P }) ,  N \uplus (\confCPS {E_\emptyset ; S_\emptyset}  {P }))  
\le \metric(M,N)
$
where the first two inequalities follow by the triangular properties of $\metric$, the last inequality follows by \autoref{lem:cong1P} and the two equalities are immediate.
\end{proof}

\begin{proposition}
\label{lem:cong2Pk}
$\metric^n(M \parallel P , N \parallel P) \le \metric^n(M,N)$, for any pure-logical process $P$ and $n \ge 0$.
\end{proposition}
\begin{proof}
The same arguments used in the proof of \autoref{lem:cong2P} apply. 
Essentially, we simply exploits \autoref{lem:cong1Pk} instead of \autoref{lem:cong1P}.
\end{proof}

Finally, we prove that weak bisimilarity metrics are preserved by channel restriction. 
\begin{proposition}
\label{lem:cong3P}
$\metric(M {\setminus} c  , N {\setminus} c) \le \metric(M,N)$, for any channel $c$.
\end{proposition}
\begin{proof}
We reason as in \autoref{lem:cong1P}.
The case $\metric(M , N) = 1$ is immediate, therefore we assume $\metric(M , N) < 1$. 
Let us define the function $d \colon \nome \times \nome \to [0,1]$ by $d(M{\setminus} c  , N{\setminus} c ) = \metric(M,N)$ for all $M,N,O \in \nome$.
To prove the thesis it is enough to show that $d$ is a weak bisimulation metric. 
In fact, since $\metric$ is the minimal weak bisimulation metric, this implies $\metric \sqsubseteq d$, thus giving $\metric(M {\setminus} c  , N {\setminus} c ) \le d(M {\setminus} c  , N {\setminus} c ) = \metric(M,N)$.
To prove that $d$ is a weak bisimulation metric, we show that any transition $M {\setminus} c  \trans{\alpha} \gamma$ is simulated by some transition $N {\setminus} c  \TransS[\widehat{\alpha}] \gamma'$ with $\Kantorovich(d)(\gamma,\gamma' + (1-\size{\gamma'})\dummyN ) \le d(M \setminus c,N \setminus c)$.
The proof proceeds by case analysis on why $M\setminus c \trans{\alpha} \gamma$. 
\end{proof}

\begin{proposition}
\label{lem:cong3Pk}
$\metric^n(M {\setminus} c  , N {\setminus} c) \le \metric^n(M,N)$, for any channel $c$ and $n \ge 0$.
\end{proposition}
\begin{proof}
We reason as in \autoref{lem:cong1Pk}.
Hence, we proceed by induction over $n$, where the base case $n=0$ is immediate and we consider the inductive step $n+1$.
The case $\metric^{n+1}(M , N) = 1$ is immediate, therefore we assume $\metric^{n+1}(M , N) < 1$. 
We have to show that any transition $M {\setminus} c  \trans{\alpha} \gamma$ is simulated by some transition $N {\setminus} c  \TransS[\widehat{\alpha}] \gamma'$ with $\Kantorovich(\metric^n)(\gamma,\gamma' + (1-\size{\gamma'})\dummyN ) \le \metric^{n+1}(M \setminus c,N \setminus c )$.
The proof proceeds by case analysis on why $M\setminus c \trans{\alpha} \gamma$. 
\end{proof}

\noindent
\textbf{Proof of \autoref{thm:congruenceP}} \hspace{0.2 cm}
By Propositions~\ref{lem:cong1P}--\ref{lem:cong3Pk}.
\qed
\vspace{0.2 cm}

Finally, as the bisimilarity $\approx$ coincides with the bisimulation metric $\approx_0$ it follows that \autoref{thm:congruence} is a special case of \autoref{thm:congruenceP}. As consequence, the proof of  \autoref{thm:congruence} follows from\autoref{thm:congruenceP}. \\

\noindent
\textbf{Proof of \autoref{thm:congruence}} \hspace{0.2 cm}
%This is a special case of \autoref{thm:congruenceP}.
%The thesis derives from \autoref{thm:congruenceP}.
%In detail, c
Consider  \autoref{thm:congruence}.1.
We have that 
\[
M \approx N 
 \; \Longrightarrow  \; \; 
\metric(M,N) = 0 
 \; \; \Longrightarrow  \;\;
\metric(M \uplus O,N\uplus O) = 0
 \;\; \Longrightarrow  \;\;
M \uplus O \approx N \uplus O
\]
by applying, respectively,  \autoref{prop:kernel}, \autoref{thm:congruenceP}.1, and \autoref{prop:kernel} again. The proofs of 
\autoref{thm:congruence}.2 and \autoref{thm:congruence}.3 are analogous.
\qed

\subsection{Proofs of  \autoref{sec:casebis}}

\label{app:sec:case-studybis}

\noindent
\textbf{Proof of \autoref{prop:case-propbis}} \hspace{0.2 cm}
The proof is analogous to that of \autoref{prop:case-prop} and \autoref{prop:air3}(\ref{prop:air3.1}).
\qed
\\

As the bisimilarity $\approx$ coincides with the bisimulation metric $\approx_0$ it follows that \autoref{prop:performances} is a special case of  \autoref{prop:case-propbis}.\\

\noindent
\textbf{Proof of \autoref{prop:performances}} \hspace{0.2 cm} 
Directly by \autoref{prop:case-propbis}(1) and  \autoref{prop:kernel}.
\qed
\\

\noindent
\textbf{Proof of \autoref{prop:case-prop}} \hspace{0.2 cm} 
Define the \CPS{} $\mathit{NIL}$ as $\mathit{NIL}=\confCPS{E_\emptyset ; S_\emptyset}{\nil}$, where $E_\emptyset$ is the empty physical  environment and $S_\emptyset$  the unique (empty)  physical state  of $E_\emptyset$.
The only transition by $\mathit{NIL}$ is $\mathit{NIL} \trans{\tick} \dirac{\mathit{NIL}}$.
By \autoref{prop:sys} and \autoref{prop:time}(d)  
we infer that $\metric^n(\mathit{Eng_g},\mathit{NIL})=0$.
 Therefore, by the triangular property of $\metric^n$, to show the thesis $\metric^n(\mathit{Eng_g}, \widehat{\mathit{Eng_g}}) \le  1- \left(1- q_g(p_g)^5\right)^n$ we can show  
$\metric^n(\mathit{NIL}, \widehat{\mathit{Eng_g}} )\leq   1- \left(1-q_g (p_g)^5\right)^n$.

The proof obligation $\metric^n(\mathit{NIL}, \widehat{\mathit{Eng_g}} )\leq   1- \left(1- (p_g)^5\right)^n$ follows from the following nine properties, by observing that the system $\widehat{\mathit{Eng_g}}$ satisfies the first one.
In the following we denote the process $\fix{Y} \tick^5 .   \rsens x {s_{\operatorname{t}}} . \ifelse {x>10} {\OUT{\mathit{warning}}{\mathrm{ID}}.Y}
{\wact{\off}{\mathit{cool}}.\tick.\mathit{Ctrl} }$ with $\mathit{RecY}$. 

\begin{enumerate}

\item  
\label{prop:case-prop1a}
$\metric^n(\mathit{NIL},\confCPS{\env_g;S}{P }) \le 1-\left( 1-  q_g (p_g)^5 \right)^n$ 
whenever the physical state   $S$ satisfies $\mathit{cool} = \off$ and $\mathit{temp} \in [0, 10.1] $, and the process $P$ is $\mathit{Ctrl}$, or $\tick.\mathit{Ctrl}$. 

\item  
\label{prop:case-prop1b}
$\metric^n(\mathit{NIL},\confCPS{\env_g;S}{P }) \le 1-\left( 1-    q_g(p_g)^5  \right)^n$ 
whenever the physical state $S$ satisfies $\mathit{cool} = \off$ and $\mathit{temp} \in (10.1, 11.4 ]$, and the process $P$ is $\mathit{Ctrl}$,   or $\mathit{Cooling}$. 

\item  
\label{prop:case-prop1c}
$\metric^n(\mathit{NIL},\confCPS{\env_g;S}{P }) \le 1-\left( 1-   (p_g)^5  \right)\left( 1-  q_g (p_g)^5  \right)^n$ 
whenever the physical state   $S$ satisfies $\mathit{cool} = \off$ and $\mathit{temp} \in (10.4, 11.5 ]$, and the process $P$ is $\mathit{Ctrl}$,   or $\mathit{Cooling}$. 

\item   
\label{prop:case-prop2a}
$\metric^n(\mathit{NIL},\confCPS{\env_g;S}{P }) \le  1-\left( 1-  q_g (p_g)^5  \right)^n$ 
whenever  the physical state   $S$ satisfies $\mathit{cool} = \on$ and $\mathit{temp} \in (9.9, 11.4 ]$, and the process $P$ is 
$\mathit{RecY}$.

\item   
\label{prop:case-prop2b}
$\metric^n(\mathit{NIL},\confCPS{\env_g;S}{P }) \le  1-\left( 1-   (p_g)^5  \right)\left( 1-  q_g (p_g)^5  \right)^n$ 
whenever  the physical state   $S$ satisfies $\mathit{cool} = \on$ and $\mathit{temp} \in (10.4, 11.5 ]$, and the process $P$ is 
$\mathit{RecY}$.

\item   
\label{prop:case-prop3}
$\metric^n(\mathit{NIL},\confCPS{\env_g;S}{P }) \le 1- \left( 1- (p_g)^{5-k} \right)   \left( 1- q_g (p_g)^5  \right)^n$,  
for all $n \in [1,4]$,
whenever  the physical state   $S$
satisfies $\mathit{cool}=\on$ and $\mathit{temp} \in  (11.4 -k(0.3),11.5 -k (0.3) ]$, and the process $P$ is
\[ 
P =\tick^{5-k} .   \rsens x {s_{\operatorname{t}}}  
\ifelse {x>10} {\OUT{\mathit{warning}}{\mathrm{ID}}.\mathit{RecY}}
{\wact{\off}{\mathit{cool}}.\tick.\mathit{Ctrl} } .
\]

\item   
\label{prop:case-prop4}
$\metric^n(\mathit{NIL},\confCPS{\env_g;S}{P }) \le1-\left( 1- q_g (p_g)^5  \right)^n$ 
whenever the physical state   $S$
satisfies $\mathit{cool} = \on$ and $\mathit{temp}  \le    11.4 -k (0.3)$, and the process $P$ is
 \[ P =\tick^{5-k} .   \rsens x {s_{\operatorname{t}}}  
\ifelse {x>10} {\OUT{\mathit{warning}}{\mathrm{ID}}.\mathit{RecY}}
{\wact{\off}{\mathit{cool}}.\tick.\mathit{Ctrl}} \]
for any $k\in [1,4]$.

\item   
\label{prop:case-prop5}
$\metric^n(\mathit{NIL},\confCPS{\env_g;S}{P }) \le1-\left( 1- q_g (p_g)^5  \right)^n$ 
whenever the physical state   $S$
satisfies  $\mathit{cool}=\on$ and $\mathit{temp}  \le  9.9$, and the process $P$ is
\[ 
P =   \rsens x {s_{\operatorname{t}}}  
\ifelse {x>10} {\OUT{\mathit{warning}}{\mathrm{ID}}.\mathit{RecY}}
{\wact{\off}{\mathit{cool}}.\tick.\mathit{Ctrl}} .
\]

\item   
\label{prop:case-prop6}
$\metric^n(\mathit{NIL},\confCPS{\env_g;S}{P}) \le  1-\left( 1- q_g (p_g)^5  \right)^n$ 
whenever the physical state  $S$
satisfies $\mathit{cool}= \on$ and $\mathit{temp}  \le  9.9$, and the process $P$ is
$P = {\wact{\off}{\mathit{cool}}.\tick.\mathit{Ctrl}}$.

\end{enumerate}

We prove these nine properties in parallel, by induction over $n$.
The base case $n=0$ is immediate since $\metric^0$ is the constant zero function $\zeroF$.
We consider the inductive step $n > 0$.
First we observe that, given any distribution $\sum_{i \in I} p_i \cdot \dirac{M_i}$ over $\CPS{s}$, the only matching $\omega \in \Omega(\sum_{i \in I} p_i \cdot \dirac{M_i},\dirac {\mathit{NIL}})$ is  $\omega(M_i,\mathit{NIL})= p_i$. 
It follows that $\Kantorovich(\metric^{n-1})(\sum_{i \in I} p_i \cdot \dirac{M_i},\dirac {\mathit{NIL}})=
 \sum_{i \in I} p_i  \metric^{n-1}(M_i,\mathit{NIL})$.
We show only the first property, the other are analogous.

We distinguish the cases $P=\mathit{Ctrl}$ and $P=\tick.\mathit{Ctrl}$.
\begin{itemize}

\item
Case $P=\mathit{Ctrl}$.\\
The only transition by $\confCPS{\env_g;S}{P}$ is 
$\confCPS{\env_g;S}{P} \trans{\tau}  \sum_{i \in I} p_i \cdot \dirac{M_i}$, 
where $M_i= \confCPS{\env_g;S}{P_i }$, with either $P_i=\tick.\mathit{Ctrl}$ or  $P_i=\mathit{Cooling}$.
The only transition by $\mathit{NIL}$ is $\mathit{NIL} \trans{\tau} \dirac{\mathit{NIL}}$.
Therefore we infer 
$\metric^n(\confCPS{\env_g;S}{P},\mathit{NIL}) \le \Kantorovich(\metric^{n-1})(\sum_{i \in I} p_i \cdot \dirac{M_i},\dirac {\mathit{NIL}})$.
By the inductive hypothesis on \autoref{prop:case-prop1a} we infer
$\metric^{n-1}(M_i,\mathit{NIL})  \le 1- \left( 1- q_g (p_g)^5  \right)^{n-1}$
in both cases, thus implying
\[
  \Kantorovich(\metric^{n-1})(\sum_{i \in I} p_i \cdot \dirac{M_i},\dirac {\mathit{NIL}})=\sum_{i \in I} p_i  \metric^{n-1}(M_i,\mathit{NIL})\le 1-\left( 1- q_g (p_g)^5  \right)^{n-1}  \le 1-\left( 1- q_g (p_g)^5  \right)^{n} .
\]
which completes the proof.

\item
Case $P=\tick.\mathit{Ctrl}$.\\
The only transition by $\confCPS{\env_g;S}{P}$ is 
$\confCPS{\env_g;S}{P} \trans{\tick} \confCPS{\mathit{next}_{ \env_g;}(S)}{ \dirac{\mathit{Ctrl}}}$.
Again, the only transition by $\mathit{NIL}$ is $\mathit{NIL} \trans{\tick} \dirac{\mathit{NIL}}$.
Therefore 
$\metric^n(\confCPS{\env_g;S}{P},\mathit{NIL}) \le \Kantorovich(\metric^{n-1})(\confCPS{\mathit{next}_{\env_g}(S)}{ \dirac{\mathit{Ctrl}}},\dirac{\mathit{NIL}})$.
By definition, 
$\mathit{next}_{ \env_g}(S)=  \sum_{v \in [0.3,1.1]_g}  \frac{1}{|[0.3,1.1]_g|} \dirac{S [\mathit{temp} \mapsto \statefun{}(temp) - v]}$.
Hence in all physical states $S'$  in the support of $\mathit{next}_{\env_g}(S)$ we have $\mathit{cool} = \off$ and
the temperature $\mathit{temp}$ lies   in the interval $[0+0.3 ,10.1+1.4]$.

We have two cases: $\mathit{temp}   \in [0+0.3 ,10.1]$, and $ \mathit{temp}   \in (10.1 ,10.5]$.
If $\mathit{temp}   \in [0+0.3 ,10.1]$, then by the inductive hypothesis on \autoref{prop:case-prop1a} we infer
$\metric^{n-1}(\confCPS{ \env_g;S' }{ \dirac{\mathit{Ctrl}}},\mathit{NIL})  \le 1- \left( 1- q_g (p_g)^5  \right)^{n-1}$,
for all $S' \in \support( \mathit{next}_{\env_g}(S))$,
 thus implying
\[
  \Kantorovich(\metric^{n-1})(\confCPS{\env_g;\mathit{next}_{ E}(S)}{ \dirac{\mathit{Ctrl}}},\dirac {\mathit{NIL}})
  \le 1-\left( 1- q_g (p_g)^5  \right)^{n-1}  \le 1-\left( 1- q_g (p_g)^5  \right)^{n} .
\]

If $\mathit{temp}   \in (10.1 ,10.5]$, then $\mathit{temp}   \in (10.4 ,10.5]$ 
with a probability bounded by $q_g$, whereas $\mathit{temp}  \in (10.1 ,10.4] $ with a probability not less that $1-q_g$.
If $\mathit{temp}   \in (10.4 ,10.5]$  we can apply  the inductive hypothesis on \autoref{prop:case-prop1c} to get
$
\metric^{n-1} (\confCPS{\env_g;S'}{ \mathit{Ctrl}},\mathit{NIL})   \le  1-\left( 1-  (p_g)^5  \right)\left( 1-  q_g (p_g)^5  \right)^{n-1} 
$, for all $S' \in \support( \mathit{next}_{\env_g}(S))$.
If $\mathit{temp}   \in (10.4 ,10.5]$ we can apply  the inductive hypothesis on \autoref{prop:case-prop1b} to get
$\metric^{n-1} (\confCPS{\env_g;S'}{ \mathit{Ctrl}},\mathit{NIL})   \le     1-\left( 1- q_g (p_g)^5  \right)^{n-1}$,
for all $S' \in \support( \mathit{next}_{\env_g}(S))$.
Therefore for some $q \le q_g$ we have
\[
\begin{array}{rlr}
&  \Kantorovich(\metric^{n-1})(  \confCPS{\env_g;\mathit{next}_{ E}(S)}{ \dirac{\mathit{Ctrl}}} ,\dirac {\mathit{NIL}})
\\
= &   (1-q)\left(  1-\left( 1- q_g (p_g)^5  \right)^{n-1 } \right) + 
q    \left( 1-   (p_g)^5  \right)  \left( 1-  q_g (p_g)^5  \right)^{n-1}
\\  
=   &    \left(  1-\left( 1- q_g (p_g)^5  \right)^{n-1 } \right) - q\left(  1-\left( 1- q_g (p_g)^5  \right)^{n-1 } \right)  +
 q   \left( 1-   (p_g)^5  \right) \left( 1-  q_g (p_g)^5  \right)^{n-1}
\\
=   &     1-\left( 1- q_g (p_g)^5  \right)^{n-1 }   - q+ q\left( 1- q_g (p_g)^5  \right)^{n-1 }    +
 q  -  \left(q- q(p_g)^5\right)    \left( 1- q_g (p_g)^5  \right)^{n-1} 
\\
=   &     1 -q+q   -   \left(1- q + q-q (p_g)^5\right)      \left( 1- q_g (p_g)^5  \right)^{n-1} 
\\
=   &     1   -   \left(1 -q (p_g)^5\right)      \left( 1- q_g (p_g)^5  \right)^{n-1} 
\\
\le   &     1   -   \left(1 - q_g (p_g)^5\right)      \left( 1- q_g (p_g)^5  \right)^{n-1} 
\\
=   &     1   -      \left( 1- q_g (p_g)^5  \right)^{n  } 
\end{array} 
\]
which completes the proof.

\end{itemize}
\qed

\noindent 
\textbf{Proof of \autoref{prop:air3}} \hspace{0.2 cm}
By \autoref{prop:case-prop} we derive 
$\metric^n(\mathit{Eng }_g, \widehat{\mathit{Eng }_g}) \le  1- \left(1- q_g(p_g)^5\right)^n=p$.
By simple 
$\alpha$-conversion it follows that 
$\metric^n(\mathit{Eng }_g^{\mathrm L} , \widehat{\mathit{Eng }_g^{\mathrm L}}) = p$ and 
$\metric^n(\mathit{Eng }_g^{\mathrm R} , \widehat{\mathit{Eng }_g^{\mathrm R}}) =  p$, respectively. 
%\remarkS{"(and the triangular property of $\metric$)" era "(and transitivity)". Ma era sbagliato, come fatto notare dal Rev.3}
By \autoref{thm:congruenceP}(\ref{thm:congruence1Pk}) 
(and the triangular property of $\metric^n$) it 
follows that 
$
\metric^n(\mathit{Eng}_g^{\mathrm L} \uplus \mathit{Eng}_g^{\mathrm R},
\widehat{\mathit{Eng}_g^{\mathrm L}} \uplus \widehat{\mathit{Eng}_g^{\mathrm R}}) \le 2p$.
By \autoref{thm:congruenceP}(\ref{thm:congruence2Pk}) it follows that 
\[
\metric^n\left(
 \left(  \mathit{Eng}_g^{\mathrm{L}}  
\uplus 
(   \mathit{Eng}_g^{\mathrm R}  \right) \parallel \mathit{Check},
 \left( \widehat{\mathit{Eng}_g^{\mathrm{L}} } 
\uplus 
(  \widehat{\mathit{Eng}_g^{\mathrm R}} \right) \parallel \mathit{Check}
\right) \le 2p.
 \]
By \autoref{thm:congruenceP}(\ref{thm:congruence3Pk}) we
obtain  
\begin{equation}
\label{eq:proof_of_thm_congruenceP}
\metric^n\left(  \mathit{Airplane}_g  , \widehat{\mathit{Airplane}_g} \right) \le 2p
\end{equation}
thus confirming that \autoref{prop:air3}(1) holds.

Finally, by \autoref{eq:proof_of_thm_congruenceP} and \autoref{eq:lim},  we derive
\[
\lim_{ g \rightarrow +\infty}   \metric^n( \mathit{Airplane}_g  ,  \widehat{\mathit{Airplane}_g}  )\leq 2\left(1- \left(1- \frac{1}{8^6}  \right)^n\right).
\]
namely \autoref{prop:air3}(2).
\qed

\end{document}